%% file: main.tex
\documentclass[10pt, draftclsnofoot, onecolumn]{IEEEtran}
\IEEEoverridecommandlockouts
\usepackage{amsmath}
\usepackage{amssymb}
\usepackage{amsthm}
\usepackage{bbm}
\usepackage{graphicx}
\usepackage{xcolor}
\usepackage{mathtools}
\usepackage{float}
\usepackage{cite}
\usepackage[hidelinks]{hyperref}
\usepackage[nameinlink,capitalise,noabbrev]{cleveref}
\usepackage{booktabs}
\usepackage{subcaption}

\theoremstyle{plain}
\newtheorem{definition}{Definition}[] 
\newtheorem{theorem}{Theorem}[]
\newtheorem{lemma}{Lemma}[]
\newtheorem{remark}{Remark}[]
\newtheorem{corollary}{Corollary}[]

\newcommand{\remove}[1]{}
\newcommand{\E}{\mathbb E}

\input defs

\title{Finite-Sample Binary Hypothesis Testing via R\'enyi Divergences:\\Strong Converse and Local Privacy}

\begin{document}

\author{Roberto~Bruno\textsuperscript{1}, Adrien~Vandenbroucque\textsuperscript{2}, and Amedeo~Roberto~Esposito\textsuperscript{3}\\[1ex]
\textsuperscript{1}Department of Computer Science, University of Salerno, Fisciano, Italy\\
\textsuperscript{2}School of Computer and Communication Sciences, EPFL, Lausanne, Switzerland\\
\textsuperscript{3}Okinawa Institute of Science and Technology, Onna, Okinawa, Japan\\[1ex]
Email: {\tt rbruno@unisa.it, adrien.vandenbroucque@epfl.ch, amedeo.esposito@oist.jp}}

\maketitle

\begin{abstract}
We study asymmetric simple binary hypothesis testing between $H_0:P_0^{ n}$ and $H_1:P_1^{ n}$, based on $n$ independent and identically distributed observations. Leveraging a variational representation of R\'enyi divergence of order $\alpha$, we derive our main result: a finite-sample converse with $\alpha>1$. The bound uses both directions of the divergence $D_\alpha(P_1\|P_0)$ and $D_\alpha(P_0\|P_1)$, tensorises under product measures, and contains familiar data-processing converses as boundary cases. For comparison, we apply the same variational approach to general $f$-divergences and specialise it to total variation, $E_\gamma$, Hellinger, and Kullback Leibler divergences, thereby recovering familiar converses within a unified framework. Together with an achievability bound involving R\'enyi divergence with $\alpha\in (0,1)$, the main converse recovers the phase transition of the optimal Type II error under the exponentially decaying Type I error constraint $\varepsilon_n=e^{-nr}$. Under regularity conditions, the optimal Type II error vanishes exponentially when $r<D(P_1\|P_0)$ and converges exponentially fast to one when $r>D(P_1\|P_0)$. We also derive sample-complexity bounds and extend both the converse and achievability analyses to locally differentially private observations, quantifying the cost of privacy and recovering the non-private achievability bound as the privacy constraint vanishes.
\end{abstract}

\section{Introduction}\label{sec:intro}
Given $n$ independent observations generated by an unknown probability distribution $P$, binary hypothesis testing asks for a decision between the null hypothesis $H_0:P=P_0$ and the alternative hypothesis $H_1:P=P_1$. A Type~I error occurs when one rejects $H_0$ even though it is true, whereas a Type~II error occurs when one accepts $H_0$ when $H_1$ is true. These errors cannot generally be minimised simultaneously. In the asymmetric Neyman--Pearson formulation, one therefore minimises the Type~II error subject to a constraint on the Type~I error. We write $\beta_n(\varepsilon)$ for the smallest Type~II error attainable when the Type~I error is at most $\varepsilon\in(0,1)$.

The asymptotic theory is well understood. For fixed $\varepsilon$, the Chernoff--Stein lemma identifies the Kullback-Leibler divergence $D(P_0\|P_1)$ as the optimal error exponent of $\beta_n(\varepsilon)$ \cite{chernoff1956large,chernoff1952measure,cover1999elements}. A different phenomenon appears when the Type~I constraint itself decays exponentially, $\varepsilon_n=e^{-nr}$. In that regime, $D(P_1\|P_0)$ becomes a critical threshold: if $r<D(P_1\|P_0)$, both errors vanish exponentially, whereas if $r>D(P_1\|P_0)$, the optimal Type~II error converges to one \cite{Blahut,Han_strongconverse}. The latter statement is what is typically called a strong converse.

These asymptotic results do not provide performance guarantees at a finite sample size. A useful finite-sample result for the strong converse regime should meet three requirements. It should remain informative at small to  moderate sample sizes, rely on divergences that tensorise under product measures, and recover convergence of the Type II error to one when the Type I exponent exceeds the critical threshold given by KL. Tensorisation allows divergences between the $n$-dimensional product distributions to be computed from their one-sample marginals. Most information measures meet these requirements only in part. Total variation and $E_\gamma$ can describe the finite-sample testing tradeoff sharply but do not tensorise~\cite{bhattacharyya2025total}. The Kullback–Leibler divergence tensorises, but its elementary finite-sample converses yield only a weak converse: when the Type I exponent exceeds the critical threshold, they can show that the optimal Type II error remains bounded away from zero, but not that it converges to one. Establishing the so-called strong converse requires additional tools such as blowing-up or smoothing-out arguments~\cite{ahlswede1976bounds,liu2020second}. Hellinger-type quantities tensorise and lead to useful sample-complexity estimates, but do not yield a strong converse either.

This paper develops a variational route to finite-sample bounds and specialises it to R\'enyi divergences of order $\alpha$. Orders above one yield converse bounds that tensorise and directly capture the strong-converse regime; orders below one yield complementary achievability bounds. After establishing the result, we extend the same variational argument to the full class of $f$-divergences, obtaining a unified converse. Specialising this converse to standard divergences recovers familiar bounds within one derivation and makes their comparison with the main bound transparent. Both the converse and the achievability bounds via R\'enyi quantify the behaviour on either side of the phase transition at $D(P_1\|P_0)$ and can be inverted to characterise sample complexity.

\subsection{Contributions}\label{sec:our_contribution}

Our main contributions are summarised as follows.
\begin{itemize}
   
    \item First we derive a finite-sample converse based on $D_\lambda$ for $\lambda>1$. The bound contains terms in both directions, $D_\lambda(P_1\|P_0)$ and $D_\lambda(P_0\|P_1)$, tensorises exactly for independent observations, and includes the corresponding data-processing bounds as boundary cases of its optimisation parameters.

    \remove{
    \item We then develop a unified converse from the variational representation of a general $f$-divergence. Specialising this construction to total variation, $E_\gamma$, and Hellinger divergences recovers their corresponding converses through a common argument and provides a systematic comparison with the main R\'enyi result.
    }
    \item For $\lambda\in(0,1)$, we establish an achievability bound for log-likelihood-ratio tests. Under regularity conditions, the achievability and converse bounds show that if the Type I error is constrained to decay exponentially like $e^{-nr}$ the optimal Type~II error probability $\beta_n(\varepsilon_n)$ approaches one exponentially when $r>D(P_1\|P_0)$ and vanishes exponentially when $r<D(P_1\|P_0)$.
    \item We invert the finite-sample bounds to obtain lower and upper bounds on the sample complexity. Under the same regularity conditions, the lower bound improves the leading constant in a recent Hellinger-based estimate~\cite{pensia2024sample} in strongly asymmetric regimes. Moreover, we can show, in the  same strongly asymmetric regimes, that our lower and upper bounds are asymptotically matching.
    \item Finally, we extend the analysis to hypothesis testing from locally differentially private observations. We derive a privacy-dependent converse and sample-complexity lower bounds, complemented by achievability guarantees based on two staircase mechanisms. These results also yield explicit achievability and strong-converse regimes under exponentially decaying Type~I constraints.
\end{itemize}

The remainder of the paper is organised as follows. Section~\ref{sec:rel_work} reviews the most closely related finite-sample bounds. Section~\ref{sec:notation} introduces the testing model and the relevant information measures. Section~\ref{sec:main} presents the finite-sample converse bound obtained through a variational representation of the R\'enyi divergence and establishes a complementary achievability bound. Section~\ref{sec:phase_transition} combines these bounds under exponentially decaying Type~I error constraints to recover, as $n\to\infty$, the sharp phase transition established in~\cite{Blahut,Han_strongconverse}. Appendix~\ref{app:f_div_framework} extends the variational approach to the broader class of $f$-divergences. 

Section~\ref{sec:comparisons} presents several numerical comparisons of the finite-sample converse bounds, showcasing how the bound induced by R\'enyi divergence improves the existing bounds in a variety of discrete and continuous settings.

Section~\ref{sec:sample_complexity} shows how the bounds derived in the previous sections can be translated into lower and upper bounds on the sample complexity, and compares them with the recent results established in~\cite{pensia2024sample,kazemi2025sample}. 

Finally, Section~\ref{sec:LDP} considers the locally private setting and investigates the behaviour of the R\'enyi divergence-based bounds established in this work in such a setting.

\section{Related Work}\label{sec:rel_work}

The Neyman--Pearson lemma characterises the optimal finite-sample decision rule as a likelihood-ratio test \cite{neyman1933ix}. For fixed Type~I error, the Chernoff--Stein lemma gives the asymptotic Type~II exponent \cite{chernoff1956large,chernoff1952measure,cover1999elements}; the same exponent persists when the Type~I constraint vanishes subexponentially \cite{Espinosa}. While, when the Type~I error constraint decays exponentially, the results of Blahut~\cite{Blahut} and of Han and Kobayashi~\cite{Han_strongconverse} identify $D(P_1\|P_0)$ as the strong-converse threshold.

At finite-sample sizes, total variation gives the exact minimum sum of the two errors:
\begin{equation}\label{eq:1-tv}
    \min_{A\in\mathcal{B}^{\otimes n}}\bigl\{P_0^n(A)+P_1^n(A^c)\bigr\}
    =1-\mathrm{TV}(P_0^n,P_1^n).
\end{equation}
The $E_\gamma$ divergence provides a related characterisation through likelihood-ratio threshold events \cite{asoodeh2020contraction,mullhaupt2025bounding}. Neither quantity, however, has a simple additive tensorisation rule, and exact evaluation under product measures can become computationally difficult as $n$ grows \cite{bhattacharyya2022approximating,bhattacharyya2025total}.

KL-based bounds avoid this difficulty because $D(P_0^n\|P_1^n)=nD(P_0\|P_1)$. Classical Fano-type arguments give explicit finite-length converses \cite{polyanskiy2010channel}, with later refinements improving their constants \cite{lungu2024optimal}. Elementary KL bounds generally do not recover the strong converse on their own. Blowing-up arguments \cite{ahlswede1976bounds} and smoothing-out methods based on reverse hypercontractivity \cite{liu2020second} strengthen them and allows KL to provide a strong converse. However, these methods require advanced machinery and the resulting finite-sample estimates may still be loose at moderate $n$.

Hellinger-based methods provide another tensorising alternative. Bar-Yossef used inequalities relating total variation and squared Hellinger distance to control testing errors \cite{bar2002complexity,sason2016f}. Pensia \textit{et al.} characterised the sample complexity of Bayesian and asymmetric binary testing up to universal constants using Hellinger and Jensen--Shannon quantities \cite{pensia2024sample}. Kazemi \textit{et al.} subsequently expressed related bounds through a Hellinger-$\lambda$ divergence in distributed settings \cite{kazemi2025sample}. These results are especially natural for the inverse problem of determining the number of samples needed to attain prescribed errors.

R\'enyi divergence is closely connected to hypothesis testing and to strong-converse exponents \cite{vanerven_renyi,mosonyi2015quantum,polyanskiy2010arimoto}. Our main focus is its use as a tractable finite-sample bridge: exact tensorisation permits product-measure calculations, while optimising the order around one recovers the phase transition governed by KL divergence. The subsequent $f$-divergence analysis serves a complementary purpose, unifying well-known converses under the same variational principle so that their relation to the R\'enyi bound can be seen directly.

\section{Preliminaries and Notation}\label{sec:notation}
Let $P_0$ and $P_1$ be probability measures on a measurable space $(\mathcal{X},\mathcal{B})$. We observe $\mathbf{X}=(X_1,\ldots,X_n)$ and test
\begin{equation*}
    H_0:\mathbf{X}\sim P_0^n \qquad\text{against}\qquad H_1:\mathbf{X}\sim P_1^n,
\end{equation*}
where $P_j^n$ denotes the $n$-fold product of $P_j$.

A randomised test is a measurable function $\phi:\mathcal{X}^n\to[0,1]$, where $\phi(\mathbf{x})$ is the conditional probability of deciding $H_1$ after observing $\mathbf{x}$. Its Type~I and Type~II error probabilities are
\begin{equation}\label{eq:error_definitions}
    \alpha_n(\phi)=\E_{P_0^n}[\phi], \qquad \beta_n(\phi)=\E_{P_1^n}[1-\phi].
\end{equation}
When the test is clear from context, we simply write $\alpha_n$ and $\beta_n$.

A deterministic test takes values in $\{0,1\}$ and is identified with its rejection region $A=\{\mathbf{x}:\phi(\mathbf{x})=1\}$. In that case, $\alpha_n=P_0^n(A)$ and $\beta_n=P_1^n(A^c)$.

More generally, for probability measures $P$ and $Q$ on $(\mathcal{Z},\mathcal{F})$, define the optimal Type~II error under a Type~I constraint as
\begin{equation}\label{eq:binary_testing_functional}
    \beta_\varepsilon(P,Q):=\inf_{\psi:\mathcal{Z}\to[0,1]}\left\{\E_Q[1-\psi]:\E_P[\psi]\leq\varepsilon\right\},\qquad \varepsilon\in(0,1),
\end{equation}
where the first argument is the distribution under the null hypothesis and the second is the distribution under the alternative. For the i.i.d. testing problem considered here, we use the abbreviation
\begin{equation}\label{eq:Beta_eps_def}
    \beta_n(\varepsilon) :=\beta_\varepsilon(P_0^n,P_1^n)=\inf_{\phi:\mathcal{X}^n\to[0,1]}\left\{\E_{P_1^n}[1-\phi]:\E_{P_0^n}[\phi]\leq\varepsilon\right\}.
\end{equation}
The constraint may depend on $n$, i.e., $\varepsilon=\varepsilon_n$.

Unless stated otherwise, $P_0$ and $P_1$ are dominated by a common measure $\mu$, with one-sample densities $p_0$ and $p_1$. We write
\begin{equation*}
    \imath_{1\|0}^{(n)}(\mathbf{x})
    :=\log\frac{dP_1^n}{dP_0^n}(\mathbf{x})
    =\sum_{i=1}^n\log\frac{p_1(x_i)}{p_0(x_i)}
\end{equation*}
whenever the likelihood ratio is well defined. All logarithms are natural.

The Neyman--Pearson lemma states that the infimum in \eqref{eq:Beta_eps_def} is attained by a possibly randomised log-likelihood-ratio test \cite{neyman1933ix}. For a threshold $\tau\in\mathbb{R}$ and $\eta\in[0,1]$, write
\begin{equation}\label{eq:randomised_llrt}
    \phi_{\tau,\eta}(\mathbf{x})
    :=\mathbbm{1}\!\left\{\imath_{1\|0}^{(n)}(\mathbf{x})>\tau\right\}
    +\eta\,\mathbbm{1}\!\left\{\imath_{1\|0}^{(n)}(\mathbf{x})=\tau\right\}.
\end{equation}
Thus the test decides $H_1$ above the threshold $\tau$ and randomises on the threshold set. Deterministic log-likelihood-ratio tests correspond to $\eta\in\{0,1\}$.

\subsection{Divergences}
Let $P$ and $Q$ be probability measures on $(\mathcal{X},\mathcal{B})$, with densities $p$ and $q$ with respect to a common dominating measure $\mu$.

\begin{definition}
    For $\lambda\in(0,\infty)\setminus\{1\}$, the R\'enyi divergence of order $\lambda$ of $P$ relative to $Q$ is \cite{vanerven_renyi}
\begin{equation}
  D_{\lambda}(P \| Q) := \frac{1}{\lambda - 1} \log  \int_{\mathcal{X}} p({x})^{\lambda} q({x})^{1-\lambda} \, d\mu({x}).
\end{equation}
We set $D_1(P\|Q)=D(P\|Q)$. The left limit always satisfies
\begin{equation}
    \lim_{\lambda\uparrow1}D_\lambda(P\|Q)=D(P\|Q).
\end{equation}
If $D_{1+s}(P\|Q)<\infty$ for some $s>0$, then the right limit also satisfies
\begin{equation}\label{eq:renyi_right_continuity}
    \lim_{\lambda\downarrow1}D_\lambda(P\|Q)=D(P\|Q).
\end{equation}
\end{definition}
For $\lambda\in(0,\infty)\setminus\{1\}$, a useful variational representation is \cite{anantharam2018variational}
\begin{equation}\label{eq:Renyi_var_formula}
        \frac{D_\lambda(P\|Q)}{\lambda} = \sup_{g \in \Gamma} \left\{ \frac{1}{\lambda-1} \log \E_P[e^{(\lambda-1)g}] - \frac{1}{\lambda}\log \E_Q[e^{\lambda g}] \right\},
\end{equation}
where $\Gamma$ is an admissible class containing the bounded measurable functions. R\'enyi divergence also tensorises exactly under products \cite{vanerven_renyi}:
\begin{equation}\label{eq:renyi_tensorisation}
    D_\lambda(P^n\|Q^n)=nD_\lambda(P\|Q),\qquad n\geq1.
\end{equation}

Similarly, $f-$divergences can be defined as follows.

\begin{definition}
    Let $f:[0,\infty)\rightarrow(-\infty,+\infty]$ be convex and lower semicontinuous, with $f(1)=0$. The $f$-divergence of $P$ relative to $Q$ is \cite{csiszar1967information}
    \begin{equation}\label{eq:f_divergence_def}
        D_f(P\|Q):=\int_{\mathcal{X}} f\left(\frac{dP}{dQ}\right)dQ= \int_{\mathcal{X}}
    q(x)f\left(\frac{p(x)}{q(x)}\right)d\mu(x),
    \end{equation}
    with the conventions $0f(0/0)=0$ and $0f(a/0)=a\lim_{u\to\infty}f(u)/u$ for $a>0$.
\end{definition}

Several classical information measures arise from particular choices of $f$:
\begin{itemize}
    \item total variation, with
    \[
        f_{\mathrm{TV}}(t)=\frac{1}{2}|t-1|;
    \]
    
    \item the $E_\gamma$ divergence, with
    \[
        f_\gamma(t)=(t-\gamma)_+-(1-\gamma)_+,
        \qquad \gamma>0;
    \]
    
    \item the squared Hellinger distance, with
    \[
        f_{\mathrm{H}}(t)
        =\frac{1}{2}\bigl(\sqrt{t}-1\bigr)^2,
    \]
    or, equivalently, $f_{\mathrm{H}}(t)=1-\sqrt{t}$;
    
    \item the Kullback--Leibler divergence, with
    \[
        f_{\mathrm{KL}}(t)=t\log t.
    \]
\end{itemize}

For the variational representation used in Appendix~\ref{app:f_div_framework}, we extend $f$ by setting $f(u)=+\infty$ for $u<0$ and let $f^\star$ be its convex conjugate:
\begin{equation}
    f^\star(t) = \sup_{u\in\mathbb{R}}\bigl\{tu-f(u)\bigr\} = \sup_{u\geq 0}\bigl\{tu-f(u)\bigr\}.\label{eq:f_convex_conjugate}
\end{equation}
Then, the $f$-divergence admits the representation \cite{nguyen2010estimating}
\begin{equation}
    \label{eq:f_divergence_variational}
    D_f(P\|Q) = \sup_{g\in\mathcal{G}} \bigl\{\E _{P}[g(X)] - \E_{Q}[f^\star(g(X))]\bigr\},
\end{equation}
where $\mathcal{G}$ denotes the class of all measurable functions $g: \mathcal{X} \to \mathbb{R}$ for which the above expectations are well defined.

\section{Converse and Achievability Bounds via R\'enyi Divergences}\label{sec:main}

In this section, leveraging the variational representation of R\'enyi divergences in \eqref{eq:Renyi_var_formula}, we derive a finite-sample converse bounds for binary hypothesis testing in terms of R\'enyi divergences of order $\lambda>1$.  The resulting bound recovers the corresponding data-processing finite-sample bound, yielding a more detailed characterisation of the Type~II error probability in certain finite-sample regimes.

The following theorem, whose proof is deferred to Appendix~\ref{app:thm:renyi_converse_bound_via_vr}, presents the converse result.

\begin{theorem}\label{thm:renyi_converse_bound_via_vr}
 Let $P_1^n$ and $P_0^n$ be mutually absolutely continuous. Then, for any $\varepsilon\in(0,1)$,

 \begin{align}\label{eq:renyi_converse_bound_via_vr}
     \beta_n(\varepsilon)&\geq\max\Bigg\{\sup_{\lambda>1,c>0} \frac{e^{(\lambda-1)c}-e^{\frac{\lambda-1}{\lambda}nD_\lambda(P_1\|P_0)}\left(1-\varepsilon+\varepsilon e^{\lambda c}\right)^{\frac{\lambda-1}{\lambda}}}{e^{(\lambda-1)c}-1},\nonumber\\&\hspace{10em}\sup_{\lambda>1,c>0} \frac{e^{-nD_\lambda(P_0\|P_1)}\left(\varepsilon +(1-\varepsilon)e^{(\lambda-1)c}\right)^{\frac{\lambda}{\lambda-1}}-1}{e^{\lambda c}-1}\Bigg\}.
 \end{align}
where each supremum is restricted to orders for which the displayed R\'enyi divergence is finite.
\end{theorem}

Optimising over $c>0$ allows us to derive a semi-closed form expression for the bound \eqref{eq:renyi_converse_bound_via_vr}, as shown in Corollary \ref{cor:semi-closed_form}. The proof is in Appendix \ref{app:semi-closed_form}.

\begin{corollary}\label{cor:semi-closed_form}
    Let $P_1^n$ and $P_0^n$ be mutually absolutely continuous, and define
    \begin{equation*}
        \Lambda_R:=\{\lambda>1:D_\lambda(P_1\|P_0)<\infty\},
        \qquad
        \Lambda_F:=\{\lambda>1:D_\lambda(P_0\|P_1)<\infty\}.
    \end{equation*}
      One has that for any $\varepsilon\in(0,1)$,
    \begin{equation}
        \beta_n(\varepsilon)\geq \max\Bigg\{\sup_{\lambda\in\Lambda_R}R(\lambda), \sup_{\lambda\in\Lambda_F}F(\lambda)\Bigg\}.
    \end{equation}
    Where, for any $\varepsilon\in(0,1)$, and $\lambda\in \Lambda_R$, we have
    \begin{align}
        R(\lambda)&:=\sup_{c>0}\frac{e^{(\lambda-1)c}-e^{\frac{\lambda-1}{\lambda}nD_\lambda(P_1\|P_0)}(1-\varepsilon+\varepsilon e^{\lambda c})^{\frac{\lambda-1}{\lambda}}}{e^{(\lambda-1)c}-1}\nonumber\\
        &= \begin{cases}
            1-\varepsilon\quad&\text{if } nD_\lambda(P_1\|P_0)=0,\\
            \frac{1-\varepsilon}{1-\varepsilon+\varepsilon e^{c_{\lambda, R}}}\quad&\text{if } 0<nD_{\lambda}(P_1\|P_0)<\log(1/\varepsilon),\\
            1-\left(\varepsilon e^{nD_\lambda(P_1\|P_0)}\right)^{\frac{\lambda-1}{\lambda}}\quad&\text{if } nD_\lambda(P_1\|P_0)\geq \log(1/\varepsilon),
        \end{cases}
    \end{align}
    where $c_{\lambda, R}>0$ is the unique positive solution to
    \begin{equation*}
        1-\varepsilon+\varepsilon e^{\lambda c_{\lambda, R}} = e^{(\lambda-1)nD_\lambda(P_1\|P_0)}\left(1-\varepsilon+\varepsilon e^{c_{\lambda, R}}\right)^\lambda.
    \end{equation*}
    Similarly, for any $\varepsilon\in(0,1)$ and $\lambda\in\Lambda_F$,  
    \begin{align}
        F(\lambda)&:=\sup_{c>0} \frac{e^{-nD_\lambda(P_0\|P_1)}\left(\varepsilon +(1-\varepsilon)e^{(\lambda-1)c}\right)^{\frac{\lambda}{\lambda-1}}-1}{e^{\lambda c}-1}=\begin{cases}
            1-\varepsilon\quad&\text{if } nD_\lambda(P_0\|P_1)=0,\\
            \frac{1-\varepsilon}{1-\varepsilon+\varepsilon e^{c_{\lambda,F}}}\quad&\text{if } nD_\lambda(P_0\|P_1)>0,
        \end{cases}
    \end{align}
    where $c_{\lambda,F}>0$ is the unique positive solution to
    \begin{equation*}
        \left(\varepsilon+(1-\varepsilon)e^{(\lambda-1)c_{\lambda, F}}\right)^{\frac{1}{\lambda-1}}\left(\varepsilon+(1-\varepsilon)e^{-c_{\lambda, F}}\right) = e^{nD_\lambda(P_0\|P_1)}.
    \end{equation*}
  
\end{corollary}

\begin{remark}\label{rem:DPI_recover}
    By considering the limit $c\to\infty$ in the two suprema in \eqref{eq:renyi_converse_bound_via_vr}, we obtain the following simpler bound:
    \begin{align}\label{eq:renyi_converse_bound_via_DPI}
      \beta_n(\varepsilon)\geq&  \max\Bigg\{1-\inf_{\lambda>1}\left(\varepsilon\, e^{nD_\lambda(P_1\| P_0)}\right)^{\frac{\lambda-1}{\lambda}}, \nonumber\sup_{\lambda>1}\left((1-\varepsilon)^{\frac{\lambda}{\lambda-1}}\,e^{-nD_\lambda(P_0\| P_1)}\right)\Bigg\}.
 \end{align}
 This recovers the converse induced by the data-processing inequality for R\'enyi divergence~\cite[Theorem 1]{bruno2026finite}.
The full converse in~\Cref{thm:renyi_converse_bound_via_vr} therefore includes the DPI bound and can be strictly tighter when its optimum occurs at finite $c$. The relevant calculation is given in Appendix~\ref{app:rem:DPI_recover}.
\end{remark}
We now shift our attention from converse bounds to achievability results. We first derive a finite-sample upper bound on the Type~II error probability of an LLRT with a prescribed threshold $\tau$, expressed in terms of the R\'enyi divergence of order $\lambda\in(0,1)$. The result is formalised in the following theorem, whose proof is deferred to Appendix~\ref{app:theorem:beta_upper_bound_renyi}.
\begin{theorem}\label{theorem:beta_upper_bound_renyi}
    Let $P_1^{n}$ be absolutely continuous with respect to $P_0^{n}$, and consider an LLRT $\phi:\mathcal{X}^n \to [0,1]$ with threshold $\tau\in \mathbb{R}$. Then, denoting by $\alpha_n(\phi)$ and $\beta_n(\phi)$ the Type~I and Type~II error probabilities of $\phi$, respectively, one has
    \begin{equation}\label{eq:beta_upper_bound_renyi}
        \beta_n(\phi) \leq \inf_{\lambda\in(0,1)} \frac{e^{(\lambda-1)nD_{\lambda}(P_1 \| P_0)} - \alpha_n(\phi) e^{\lambda \tau}}{e^{(\lambda-1)\tau}}.
    \end{equation}
\end{theorem}
Although Theorem~\ref{theorem:beta_upper_bound_renyi} is stated for a fixed likelihood-ratio test, the Neyman--Pearson lemma allows it to be converted into a bound on $\beta_n(\varepsilon)$. Indeed, applying the theorem to a possibly randomised Neyman--Pearson test with Type~I error $\varepsilon$ gives the following result.
\begin{corollary}\label{cor:beta_optimal_upper_bound_renyi}
    Let $P_0$ and $P_1$ be mutually absolutely continuous. Then, for every $\varepsilon\in(0,1)$,
    \begin{equation}\label{eq:beta_optimal_upper_bound_renyi}
        \beta_n(\varepsilon)\leq\inf_{\lambda\in(0,1)}\lambda\left(\frac{1-\lambda}{\varepsilon}\right)^{\frac{1-\lambda}{\lambda}}\exp\left(-\frac{1-\lambda}{\lambda}nD_\lambda(P_1\|P_0)\right).
    \end{equation}
\end{corollary}
The proof is deferred to Appendix~\ref{app:cor:beta_optimal_upper_bound_renyi}. By the skew-symmetry of R\'enyi divergence, the bound admits an equivalent formulation in terms of $D_\lambda(P_0\|P_1)$. We retain the form involving $D_\lambda(P_1\|P_0)$, as it is the one naturally used to establish the achievability side of the phase transition in Theorem~\ref{th:phase_transition}.

In the next section, we combine the converse and achievability bounds to show how our finite-sample estimates recover the classical phase transition described by \cite{Blahut} and \cite{Han_strongconverse}. 
\section{Phase transition of the Optimal Error}\label{sec:phase_transition}
In this section, we investigate the strong converse regime, where the Type~I error probability is constrained to decay exponentially with the sample size at a rate $r>0$, i.e., $\alpha_n\leq e^{-nr}$. 

The asymptotic behaviour of the optimal Type~II error probability in this regime is well understood, with the sharp phase transition at $D(P_1\|P_0)$ established by Blahut~\cite{Blahut} and Han and Kobayashi~\cite{Han_strongconverse}. Here, leveraging the finite-sample converse and achievability
bounds derived above, we provide explicit exponential estimates for the
optimal Type~II error and recover the same phase transition as $n\to\infty$.
Thus, the threshold $D(P_1|P_0)$, separating the achievability and strong-converse regimes, emerges directly from our nonasymptotic bounds.
This result is formalised in the following theorem, whose proof is deferred to Appendix~\ref{app:th:phase_transition}.
\begin{theorem}\label{th:phase_transition}
    Let $P_1$ and $P_0$ be mutually absolutely continuous, and assume that $D_{1+s}(P_1\|P_0)<\infty$ for some $s>0$. Let $\varepsilon_n=e^{-nr}$ for $r>0$.
    Then the optimal Type~II error satisfies

    \begin{equation}
         \lim_{n\to \infty} \beta_n\left(\varepsilon_n\right) = \begin{cases}
             1 &\text{if } r>D(P_1\| P_0),\\
             0 &\text{if } r<D(P_1\| P_0).
         \end{cases}
    \end{equation}
    In particular, for every finite $n$, if $r>D(P_1\|P_0)$, then
        \begin{equation}\label{eq:r>}
        \beta_n(\varepsilon_n)\geq 1-\inf_{\lambda>1} e^{-\frac{\lambda-1}{\lambda}n\big(r-D_{\lambda}(P_1\|P_0)\big)},
        \end{equation}
        whereas, if $r<D(P_1\|P_0)$, then
        \begin{equation}\label{eq:r<}
            \beta_n(\varepsilon_n)\leq\inf_{\lambda\in(0,1)}e^{-\frac{1-\lambda}{\lambda}n\big(D_{\lambda}(P_1\|P_0)-r\big)}.
        \end{equation}
\end{theorem}
We observe that the unconditional left-continuity of $D_\lambda$ at one, together with the assumed right-continuity, makes the relevant exponent strictly positive on the respective sides of the threshold. More precisely, for $r>D(P_1\|P_0)$,
\begin{equation}\label{eq:strong_converse_exponent}
1-\beta_n(\varepsilon_n)
\leq
\exp\left(-n\sup_{\lambda>1}\left\{
\frac{\lambda-1}{\lambda}
\bigl(r-D_\lambda(P_1\|P_0)\bigr)\right\}\right),
\end{equation}
whereas for $r<D(P_1\|P_0)$,
\begin{equation}\label{eq:achievability_exponent}
\beta_n(\varepsilon_n)
\leq
\exp\left(-n\sup_{\lambda\in(0,1)}\left\{
\frac{1-\lambda}{\lambda}
\bigl(D_\lambda(P_1\|P_0)-r\bigr)\right\}\right).
\end{equation}
Thus, the above bounds provide a finite-sample counterpart to the asymptotic results of Blahut \cite{Blahut} and Han and Kobayashi \cite{Han_strongconverse}, quantifying the rate at which the Type~II error probability converges to 1 or vanishes. We note that the boundary case $r=D(P_1\|P_0)$ is not resolved by this argument. Moreover, exchanging the hypotheses gives the corresponding statement for the swapped testing problem with threshold $D(P_0\|P_1)$, provided the moment assumption is exchanged as well: $D_{1+s}(P_0\|P_1)<\infty$ for some $s>0$.

\section{Comparison with Other Divergences}\label{sec:comparisons}
In this section, we provide a numerical comparison across various regimes of the finite-sample R\'enyi converse established in Theorem \ref{thm:renyi_converse_bound_via_vr}.
The relevant converse families we consider differ along three axes: whether they describe the finite-sample trade-off exactly, whether they remain tractable under product measures, and whether they recover the strong-converse limit directly. Table~\ref{tab:bound_comparison} summarises these distinctions. Total variation and $E_\gamma$ are sharp for a fixed pair of product measures and perfectly characterise the Type~II error probabilities, but their lack of additive tensorisation makes their computation intractable as $n$ grows. On the other hand, elementary KL converses tensorise and thereby are easily computable but generally yield only a weak converse. Blowing-up and smoothing-out arguments recover the strong converse. Nevertheless, these approaches yield bounds that may be loose at moderate sample sizes. The Hellinger measures tensorise and are effective for sample-complexity control, but do not force $\beta_n(\varepsilon_n)\to1$ above the critical rate. By contrast, the converse induced by R\'enyi in Theorem~\ref{thm:renyi_converse_bound_via_vr} tensorises and directly establishes \(\beta_n(\varepsilon_n)\to1\) whenever the Type I exponent exceeds \(D(P_1\|P_0)\).

\begin{table}[t]
    \centering
    \caption{Qualitative properties of finite-sample converse methods.}
    \label{tab:bound_comparison}
    \renewcommand{\arraystretch}{1.15}
    \begin{tabular}{lccc}
        \hline
        Method & Tensorisation for product measures & Strong converse \\
        \hline
        $\mathrm{TV}$ / $E_\gamma$  & no & yes \\
        Elementary KL  & yes & no \\
        KL with Smoothing Out & -- & yes \\
        Hellinger  & yes & no \\
        R\'enyi (this work) & yes & yes \\
        \hline
    \end{tabular}
    \par\vspace{2pt}
    \begin{minipage}{0.82\linewidth}
    \end{minipage}
\end{table}

For comparison, we evaluate the R\'enyi converse alongside the Hellinger-$\lambda$ and KL-based converses summarised in Table~\ref{tab:converse_summary}. As shown in Appendix~\ref{app:f_div_framework}, these bounds can be derived within a common framework based on a restricted variational representation of $f$-divergences. In particular, we include strengthened KL bounds obtained via the smoothing-out method, a reverse-hypercontractive alternative to the blowing-up method that extends beyond finite alphabets and yields sharper second-order estimates under bounded likelihood-ratio conditions; their explicit form is given in Appendix~\ref{app:KL_based_bound}.

In the numerical comparison, we do not include the $E_\gamma$ converse, although it is also summarised in Table~\ref{tab:converse_summary}. This omission is due to the lack of the tensorisation property, which makes the $E_\gamma$ divergence computationally intractable for large $n$~\cite{bhattacharyya2025total}.

We consider two main settings. The first is a Bernoulli testing problem, with
$$
P_0=\mathrm{Bern}(1/2)
\quad\text{and}\quad
P_1=\mathrm{Bern}(1/2+\Delta),
$$
where $\Delta>0$ is varied. The second is a Gaussian testing problem, with
$$
P_0=\mathcal{N}(\mu,1)
\quad\text{and}\quad
P_1=\mathcal{N}(\mu+\Delta,1),
$$
where $\mu=2$ and $\Delta>0$ is again varied.

For both settings, we evaluate the converse bounds across three different regimes for the Type~I error constraint $\varepsilon_n$: a constant regime ($\varepsilon_n=c$), a linear-decay regime ($\varepsilon_n=1/n$), and an exponential-decay regime ($\varepsilon_n=e^{-nr}$). The corresponding numerical results for the Bernoulli and Gaussian settings are reported in Figures~\ref{fig:bernoulli_converse} and~\ref{fig:gaussian_converse}, respectively.

As illustrated by the figures, the R\'enyi converse accurately captures the strong converse behaviour in the exponential regime, even for moderate sample sizes. Moreover, in the constant and linear-decay regimes, the R\'enyi converse generally provides a tighter characterisation of the Type~II error probability than the KL- and Hellinger-$\lambda$-based converses. These numerical results highlight the advantage of the R\'enyi formulation in capturing the finite-sample behaviour of hypothesis testing in asymmetric regimes.

\begin{table}[t]
\centering
\caption{Converse bounds for $\beta_n(\varepsilon)$ with $\varepsilon\in(0,1)$.}
\label{tab:converse_summary}
\renewcommand{\arraystretch}{1.6}
\setlength{\tabcolsep}{10pt}
\small
\begin{tabular}{@{}l l@{}}
\toprule
\textbf{Measure} & \textbf{Converse} \\
\midrule

$E_\gamma$ divergence
&
$\displaystyle
\beta_n(\varepsilon)\geq\max\left\{
\begin{array}{l}
\displaystyle\sup_{\gamma>0}
\frac{1-\varepsilon-E_\gamma(P_0^n\|P_1^n)-(1-\gamma)_+}{\gamma},\\[2ex]
\displaystyle\sup_{\gamma>0}
\Bigl(1-\gamma\varepsilon-E_\gamma(P_1^n\|P_0^n)-(1-\gamma)_+\Bigr)
\end{array}
\right\}$
\\
\midrule

Total variation
&
$\displaystyle
\beta_n(\varepsilon)\geq 1-\varepsilon-\mathrm{TV}(P_0^n,P_1^n)$
\\
\midrule

Hellinger divergence
&
$\displaystyle
\varepsilon^\lambda\bigl(1-\beta_n(\varepsilon)\bigr)^{1-\lambda}
+(1-\varepsilon)^\lambda\beta_n(\varepsilon)^{1-\lambda}\geq h_\lambda(P_0^n,P_1^n),
\quad \lambda\in(0,1)$
\\
\midrule

KL divergence
&
$\displaystyle
\beta_n(\varepsilon)\geq\max\left\{
\begin{array}{l}
\displaystyle\exp\left\{-\frac{nD(P_0\|P_1)+\log 2}{1-\varepsilon}\right\},\\[2ex]
\displaystyle 1-\frac{nD(P_1\|P_0)+\log 2}{\log(1/\varepsilon)}
\end{array}
\right\}$
\\

\midrule
KL divergence with smoothing-out
&
$\displaystyle
\beta_n(\varepsilon)\geq\max\left\{
\begin{array}{l}
\displaystyle
\exp\Bigg\{-nD(P_0\|P_1)-2\sqrt{\log\frac{1}{1-\varepsilon}}\sqrt{n\left(\left\|\frac{dP_0}{dP_1}\right\|_{\infty}-1\right)}-\log\frac{1}{1-\varepsilon}\Bigg\},\\[2ex]
\displaystyle
1- \exp\left\{-\left(\left\|\frac{dP_1}{dP_0}\right\|_{\infty}-1\right)n\left(\sqrt{1-\frac{nD(P_1\|P_0)+\log \varepsilon}{\left(\left\|\frac{dP_1}{dP_0}\right\|_{\infty}-1\right)n}}-1\right)^2\right\}
\end{array}
\right\}
$
\\
\bottomrule
\end{tabular}
\end{table}

\begin{figure}[t] 
    \centering
    \begin{subfigure}{\textwidth}
        \centering
        \includegraphics[width=\textwidth]{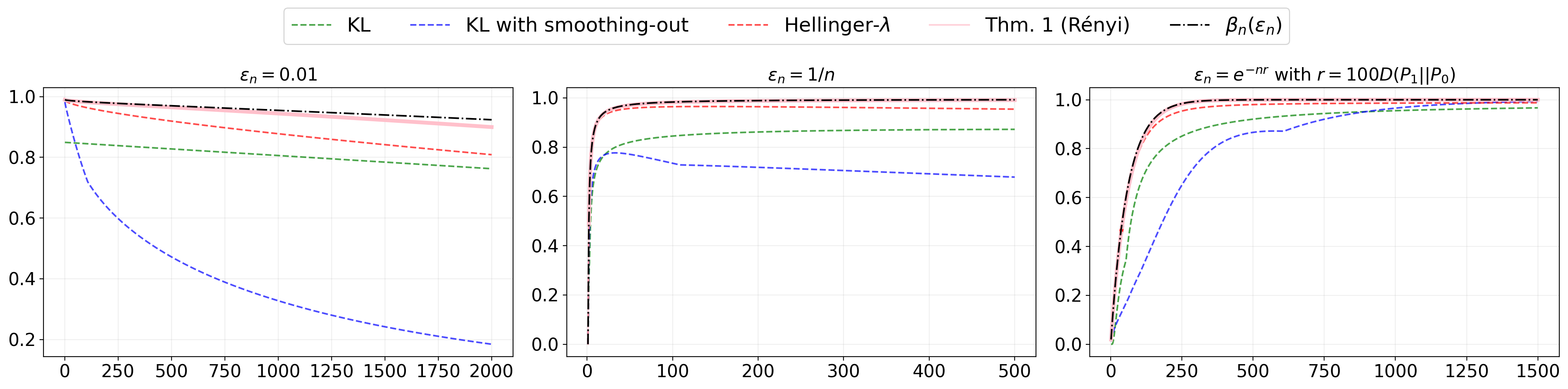}
    \end{subfigure}
    
    \par
    \par
    
    \begin{subfigure}{\textwidth}
        \centering
        \includegraphics[width=\textwidth]{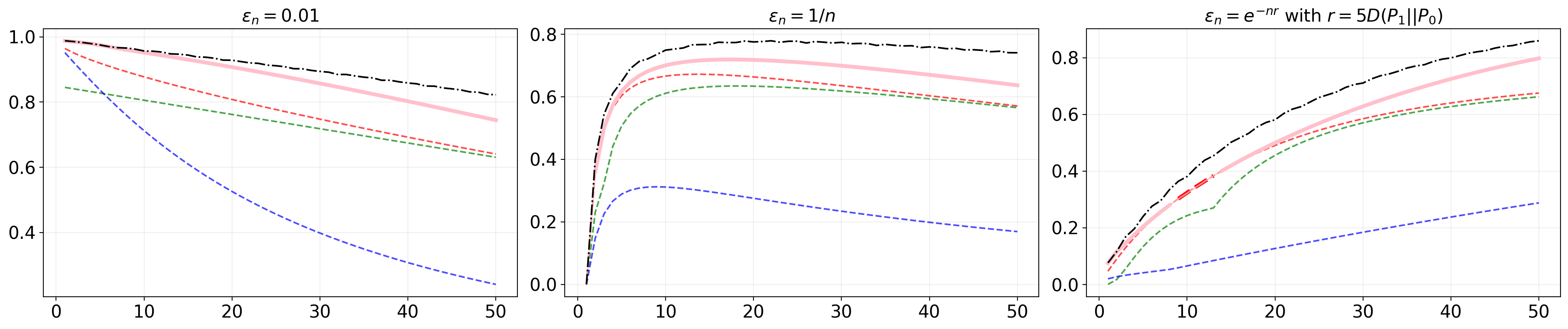}
    \end{subfigure}
    
    \caption{Comparison of the finite-sample 
    converse bounds for Bernoulli testing ($P_0=\text{Bern}(1/2)$ vs $P_1=\text{Bern}(1/2+\Delta)$) for  $\Delta\in\{0.01,0.1\}$ (top and bottom rows, respectively), and under three Type I error ($\varepsilon_n$) regimes: constant (left column), linear (middle column), and exponential (right column). The $x$-axis denotes the sample size $n$, and the $y$-axis the Type II error probability. Solid lines indicate that the corresponding lower-bound is the largest. All plotted bounds have been optimised whenever possible.}
    \label{fig:bernoulli_converse}
\end{figure}

\begin{figure}[t]
    \centering
    
    \begin{subfigure}{\textwidth}
        \centering
        \includegraphics[width=\textwidth]{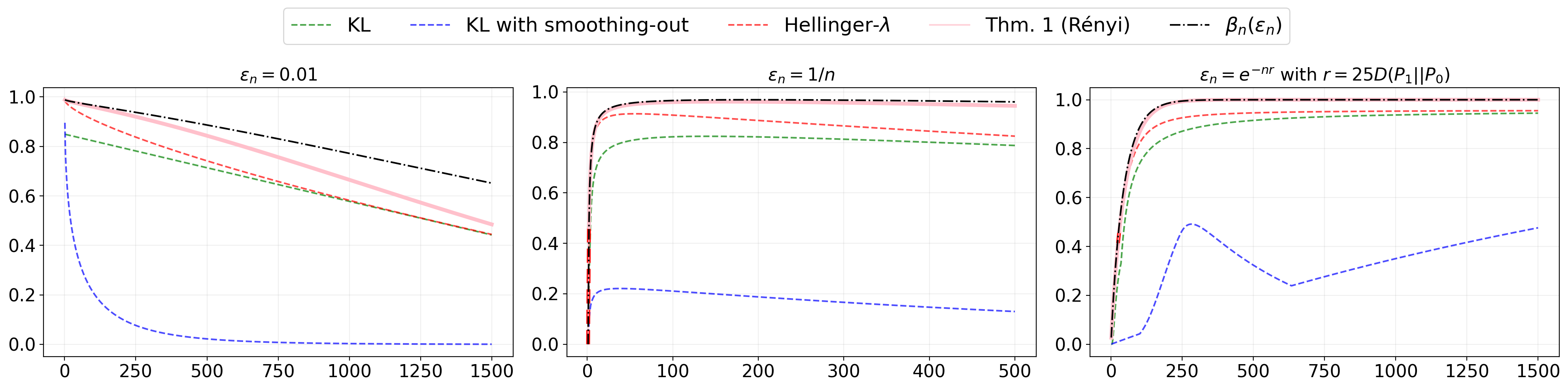}
    \end{subfigure}

     \par
    \par
    
    \begin{subfigure}{\textwidth}
        \centering
        \includegraphics[width=\textwidth]{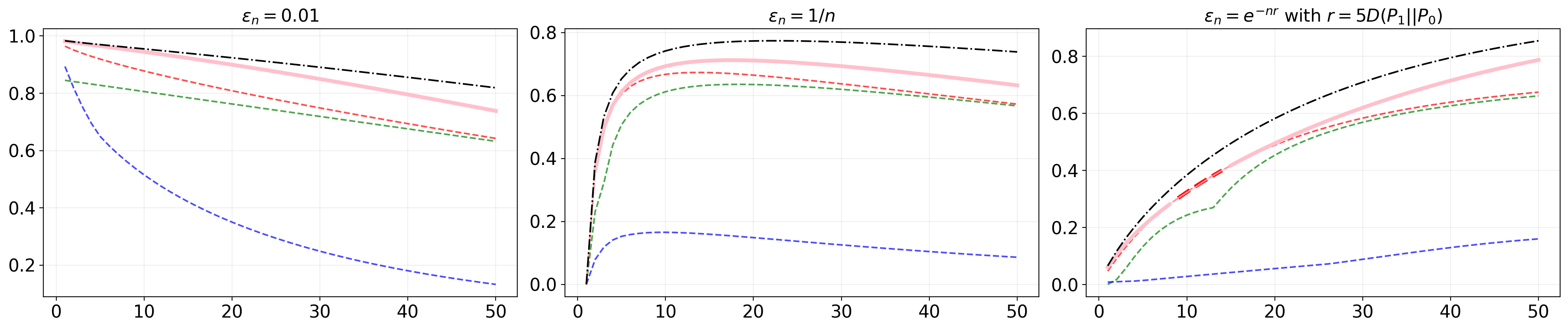}
    \end{subfigure}

    \caption{Comparison of the finite-sample 
    converse bounds for Gaussian testing ($P_0=\mathcal{N}(2,1)$ vs $P_1=\mathcal{N}(2+\Delta, 1)$) for  $\Delta\in\{0.05,0.2\}$ (top and bottom rows, respectively), and under three Type I error ($\varepsilon_n$) regimes: constant (left column), linear (middle column), and exponential (right column). The $x$-axis denotes the sample size $n$, and the $y$-axis the Type II error probability. Solid lines indicate that the corresponding lower-bound is the largest. All plotted bounds have been optimised whenever possible.}
    \label{fig:gaussian_converse}
\end{figure}

\remove{
Figure~\ref{fig:strong_converse} illustrates the finite-$n$ behaviour for Bernoulli testing. We compare the exact Neyman--Pearson tradeoff with the R\'enyi converse and representative forward- and reverse-KL baselines, including the smoothing-out estimates stated in Appendix~\ref{app:KL_based_bound}. In this example, the elementary KL bounds are nontrivial but remain far below the exact Type~II error at moderate sample sizes, whereas the optimised R\'enyi bound follows the strong-converse transition more closely.
}
\remove{
\begin{figure}[t]

    \centering
    \IfFileExists{plots/strong_converse_comparison.png}{%
        \includegraphics[width=0.75\linewidth]{plots/strong_converse_comparison.png}%
    }{%
        \fbox{\parbox[c][1.4in][c]{0.7\linewidth}{\centering Missing figure asset: \texttt{plots/strong\_converse\_comparison.png}}}%
    }
    \caption{Converse bounds for $P_0=\mathrm{Bern}(1/2)$ and $P_1=\mathrm{Bern}(1/2+\Delta)$ with $\Delta=0.05$, under $\varepsilon_n=e^{-nr}$ and $r=5D(P_1\|P_0)$.}
    \label{fig:strong_converse}
\end{figure}
}

\section{Sample Complexity}\label{sec:sample_complexity}

The sample complexity is the minimum number of i.i.d. observations required to meet prescribed Type~I and Type~II error constraints. We obtain lower and upper bounds by inverting the finite-sample converse and achievability estimates, and then compare them with recent Hellinger-based results \cite{pensia2024sample,kazemi2025sample}.

For $\varepsilon,\delta\in(0,1)$, define
\begin{equation*}
n(\varepsilon,\delta)
=
\min\left\{
n\in\mathbb{N}:
\exists\,\phi:\mathcal{X}^n\rightarrow[0,1]
\text{ such that }
\E_{P_0^n}[\phi]\leq\varepsilon,\
\E_{P_1^n}[1-\phi]\leq\delta
\right\}.
\end{equation*}

The finite-sample converse in Theorem~\ref{thm:renyi_converse_bound_via_vr} can be inverted to lower-bound $n(\varepsilon,\delta)$. The proof is in Appendix~\ref{app:cor:lower_bound_sample_complexity}.

\begin{corollary}\label{cor:lower_bound_sample_complexity}
     Let $P_1$ and $P_0$ be distinct and mutually absolutely continuous. Then, for any $\varepsilon,\delta\in(0,1)$,
     \begin{align}\label{eq:lb_sample_complexity}
         n(\varepsilon,\delta)\geq \max \Bigg\{&\sup_{\lambda>1,c>0}\left\{\frac{1}{D_{\lambda}(P_1\|P_0)}\left[\frac{\lambda}{\lambda-1}\log((1-\delta)e^{(\lambda-1)c}+\delta)-\log(\varepsilon e^{\lambda c}+1-\varepsilon)\right]\right\},\nonumber\\
         &\sup_{\lambda>1,c>0}\left\{\frac{1}{D_{\lambda}(P_0\|P_1)}\left[\frac{\lambda}{\lambda-1}\log((1-\varepsilon)e^{(\lambda-1)c}+\varepsilon)-\log(\delta e^{\lambda c}+1-\delta)\right]\right\}\Bigg\}.
     \end{align}
     where each supremum is restricted to orders for which the divergence in the denominator is finite and strictly positive.
     Furthermore, when $\varepsilon+\delta<1$, optimising over $c$ yields
     \begin{align}\label{eq:lb_sample_complexity_opt}
          &n(\varepsilon,\delta)\geq\\
          &\max\Bigg\{\sup_{\lambda>1}\Bigg\{\frac{1}{D_{\lambda}(P_1\|P_0)}\Bigg[\frac{\lambda}{\lambda-1}\log\left((1-\delta)^\lambda\left(\frac{1-\varepsilon}{\varepsilon\delta}\right)^{\lambda-1}\!\!\!+\delta\right)\nonumber-\log\left(\varepsilon^{1-\lambda} \left(\frac{(1-\varepsilon)(1-\delta)}{\delta}\right)^\lambda\!\!+1-\varepsilon\right)\Bigg]\Bigg\},\nonumber\\
          &\hspace{2em}\sup_{\lambda>1}\Bigg\{\frac{1}{D_{\lambda}(P_0\|P_1)}\Bigg[\frac{\lambda}{\lambda-1}\log\left((1-\varepsilon)^\lambda\left(\frac{1-\delta}{\varepsilon\delta}\right)^{\lambda-1}\!\!\!+\varepsilon\right)\nonumber-\log\left(\delta^{1-\lambda} \left(\frac{(1-\varepsilon)(1-\delta)}{\varepsilon}\right)^\lambda\!\!+1-\delta\right)\Bigg]\Bigg\}\Bigg\}.
     \end{align}
     In \eqref{eq:lb_sample_complexity_opt}, each supremum is restricted to orders $\lambda$ for which its corresponding R\'enyi divergence $D_{\lambda}(P_1\|P_0)$ or $D_{\lambda}(P_0\|P_1)$ is finite and strictly positive.
\end{corollary}


The upper bound on $\beta_n(\varepsilon)$ in Corollary~\ref{cor:beta_optimal_upper_bound_renyi} can be inverted to obtain a corresponding upper bound on the sample complexity. The proof is deferred to Appendix~\ref{app:cor:upper_bound_sample_complexity}.
\begin{corollary}\label{cor:upper_bound_sample_complexity}
    Let $P_0$ and $P_1$ be distinct and mutually absolutely continuous. Then, for any $\varepsilon,\delta\in(0,1)$ such that $\varepsilon+\delta<1$,
    \begin{equation}\label{eq:ub_sample_complexity}
        n(\varepsilon,\delta)\leq\inf_{\lambda\in(0,1)}\left\lceil\frac{\log\left(\frac{1-\lambda}{\varepsilon}\right)+\frac{\lambda}{1-\lambda}\log\left(\frac{\lambda}{\delta}\right)}{D_\lambda(P_1\|P_0)}\right\rceil.
    \end{equation}
\end{corollary}
By the skew-symmetry of R\'enyi divergence and the change of variables $\lambda\mapsto1-\lambda$, the bound admits an equivalent expression in terms of $D_\lambda(P_0\|P_1)$.

\subsection{Comparisons with existing results}

In this section, we compare the bounds on the sample complexity established in the previous section with the Hellinger-based bounds of Pensia \textit{et al.}~\cite{pensia2024sample}, using the closed-form reformulation of Kazemi \textit{et al.}~\cite{kazemi2025sample}. Under the condition $D_{1+\eta}(P_1\|P_0)<\infty$ for some $\eta>0$, our lower bound improves the leading constant of~\cite{pensia2024sample,kazemi2025sample} in a strongly asymmetric regime, where one error constraint is fixed and the other vanishes. Under the same condition, our upper and lower bounds match asymptotically in this regime, thereby fully characterising the sample complexity to first order.

We begin by recalling the closed-form reformulation of~\cite[Theorem 2]{kazemi2025sample}. Their bounds are formulated in the Bayesian setting and are expressed in terms of the Hellinger affinity of order $\lambda\in(0,1)$. Using the equivalence between the Bayesian and asymmetric settings established in~\cite[Claim 4.6]{pensia2024sample}, together with the relation
\begin{equation*}
    D_\lambda(P\|Q)=\frac{1}{\lambda-1}\log h_{\lambda}(P,Q),\quad\text{for }\lambda\in(0,1)
\end{equation*}
where $h_{\lambda}(P,Q)=\E_Q\left[\left(\frac{dP}{dQ}\right)^\lambda\right]$ is the Hellinger affinity of order $\lambda$, we can rewrite their bounds as follows.

For $\varepsilon\leq \delta\leq 1/32$, the resulting bounds are
\begin{equation}\label{eq:pensia_sc_eps<delta}
    \Bigg\lceil\frac{1}{2}\frac{\lambda_*}{1-\lambda_*}\frac{\log(1/(2\delta))}{D_{\lambda_*}(P_1\|P_0)}\Bigg\rceil\leq n(\varepsilon,\delta)\leq \Bigg\lceil2\frac{\lambda_*}{1-\lambda_*}\frac{\log(1/(2\delta))}{D_{\lambda_*}(P_1\|P_0)}\Bigg\rceil,
\end{equation}
where
\begin{equation}
    \lambda_*= \frac{\log(1/(2\varepsilon))}{\log(1/(2\delta))+\log(1/(2\varepsilon))}\in[0.5,1).
\end{equation}
Similarly, for $\delta\leq \varepsilon\leq 1/32$, they become
\begin{equation}\label{eq:pensia_sc_eps>delta}
    \Bigg\lceil\frac{1}{2}\frac{\lambda_*}{1-\lambda_*}\frac{\log(1/(2\varepsilon))}{D_{\lambda_*}(P_0\|P_1)}\Bigg\rceil\leq n(\varepsilon,\delta)\leq \Bigg\lceil2\frac{\lambda_*}{1-\lambda_*}\frac{\log(1/(2\varepsilon))}{D_{\lambda_*}(P_0\|P_1)}\Bigg\rceil,
\end{equation}
where 
\begin{equation}
    \lambda_*= \frac{\log(1/(2\delta))}{\log(1/(2\delta))+\log(1/(2\varepsilon))}\in[0.5,1).
\end{equation}

We now compare the lower bound in \eqref{eq:lb_sample_complexity} with the one in \eqref{eq:pensia_sc_eps<delta}. In particular, we focus on the asymmetric regime in which the Type~II error constraint $\delta$ is fixed while the Type~I error constraint $\varepsilon\to0$. Our lower bound improves the leading constant by a factor of two.
The argument is formalised in the following remark, whose proof is in Appendix \ref{app:remark:sc_bounds_comparison}.

\begin{remark}\label{remark:sc_bounds_comparison}
    Consider the strongly asymmetric regime where $\delta\in(0,1/32]$ is fixed and $\varepsilon\to0$. Assume $0<D(P_1\|P_0)<\infty$ and $D_{1+\eta}(P_1\|P_0)<\infty$ for some $\eta>0$. In this regime, the asymptotic behaviour of our lower bound \eqref{eq:lb_sample_complexity} of Corollary~\ref{cor:lower_bound_sample_complexity} reduces to
    \begin{equation}
        n(\varepsilon,\delta)\gtrsim \frac{\log(1/\varepsilon)}{D(P_1\|P_0)}.
    \end{equation}
    By contrast, the bound \eqref{eq:pensia_sc_eps<delta} retrieved from~\cite{kazemi2025sample} gives
    \begin{equation}
         n(\varepsilon,\delta)\gtrsim \frac{1}{2}\frac{\log(1/\varepsilon)}{D(P_1\|P_0)}.
    \end{equation}
    
    Consequently, the bound \eqref{eq:lb_sample_complexity} of Corollary~\ref{cor:lower_bound_sample_complexity} is asymptotically tighter than the result provided by~\cite{kazemi2025sample} by a factor of 2.
\end{remark}
By symmetry, an analogous conclusion holds when $\varepsilon$ is fixed and $\delta\to0$, provided $D_{1+\eta}(P_0\|P_1)<\infty$ for some $\eta>0$. In this case, the relevant divergence is $D(P_0\|P_1)$.

We next compare the upper bounds. The proof is in Appendix~\ref{app:remark:sc_upper_comparison}. 

\begin{remark}\label{remark:sc_upper_comparison}
   In the symmetric regime $\varepsilon=\delta$, the upper bound in \eqref{eq:pensia_sc_eps<delta} reduces to 
    \begin{equation}\label{eq:symmetric_case_sc}n(\varepsilon,\delta)\leq \Bigg\lceil\frac{2}{D_{1/2}(P_1\|P_0)}\log\left(\frac{1}{2\delta}\right)\Bigg\rceil.
    \end{equation}
    Similarly, evaluating \eqref{eq:ub_sample_complexity} at $\lambda=1/2$ also yields \Cref{eq:symmetric_case_sc}. Hence,  the two bounds have the same behaviour in this regime.

    In the asymmetric regime $\varepsilon<\delta$, the bound in \eqref{eq:pensia_sc_eps<delta} reduces to
    \begin{equation}\label{eq:bound_case_2}
        n(\varepsilon,\delta)\leq \Bigg\lceil\frac{2}{D_{\lambda_*}(P_1\|P_0)}\log\left(\frac{1}{2\varepsilon}\right)\Bigg\rceil,
    \end{equation}
    where
    $$
    \lambda_*= \frac{\log(1/(2\varepsilon))}{\log(1/(2\delta))+\log(1/(2\varepsilon))}\in[0.5,1).$$
    By contrast, whenever
    $$\varepsilon\leq \frac{1}{2}\left(\frac{\delta}{\lambda_*}\right)^{\frac{\lambda_*}{1-\lambda_*}},$$
    evaluating \eqref{eq:ub_sample_complexity} at $\lambda=\lambda_*$ gives exactly \eqref{eq:bound_case_2}. Consequently, in this specific asymmetric regime, our upper bound is at
least as tight as the bound in~\eqref{eq:pensia_sc_eps<delta}.
\end{remark}

\remove{
\begin{remark}\label{remark:sc_upper_comparison}
    In the symmetric regime $\varepsilon=\delta$, evaluating \eqref{eq:ub_sample_complexity} at $\lambda=1/2$ gives
    \begin{equation}\label{eq:symmetric_case_sc}
        n(\varepsilon,\varepsilon)
        \leq
        \left\lceil\frac{2\log(1/\varepsilon)}{D_{1/2}(P_1\|P_0)}\right\rceil.
    \end{equation}
    By contrast, the upper bound in \eqref{eq:pensia_sc_eps<delta} contains $2\log(1/(2\varepsilon))$ in the numerator. The two bounds therefore have the same leading-order behaviour as $\varepsilon\to0$, while the latter has a better additive constant at finite error levels. In asymmetric regimes, no distribution-independent ordering follows directly from the closed forms; optimised bounds can instead be compared for specific pairs $(P_0,P_1)$.
\end{remark}
}

We now demonstrate that our lower and upper bounds on the sample complexity established in Corollaries~\ref{cor:lower_bound_sample_complexity} and~\ref{cor:upper_bound_sample_complexity} are asymptotically matching in strongly asymmetric regimes.  In particular, they yield the exact first-order scaling of the sample complexity when the Type~I error is required to decay exponentially and the Type~II error is fixed. This result is shown in the corollary below, whose proof is in Appendix~\ref{app:cor:sc_gap}.

\begin{corollary}
    \label{cor:sc_gap}
    Let $\varepsilon_m=e^{-mR}$ for some $R>0$, and fix $\delta\in(0,1)$. Let $L_m$ and $U_m$ denote the lower and upper bounds in Corollaries~\ref{cor:lower_bound_sample_complexity} and~\ref{cor:upper_bound_sample_complexity}, respectively, so that
    \begin{equation*}
        L_m \leq n(\varepsilon_m, \delta) \leq U_m
    \end{equation*}
    Assume that $0<D(P_1\|P_0)<\infty$ and $D_{1+\eta}(P_1\|P_0)<\infty$ for some $\eta>0$. Then
    \begin{equation}
        \lim_{m\to\infty}\frac{L_m}{m}
        =\lim_{m\to\infty}\frac{U_m}{m}
        =\frac{R}{D(P_1\|P_0)},
    \end{equation}
    and consequently $\lim_{m\to\infty}U_m/L_m=1$.
\end{corollary}

\remove{
\begin{figure}[t] 
    \centering
    \begin{subfigure}{\textwidth}
        \centering
        \includegraphics[width=\textwidth]{plots/sample_complexity/Bernoulli_sc_fixed_delta_0.png}
    \end{subfigure}
    
    \par
    \par
    
    \begin{subfigure}{\textwidth}
        \centering
        \includegraphics[width=\textwidth]{plots/sample_complexity/Bernoulli_sc_fixed_epsilon_0.png}
    \end{subfigure}
    
    \caption{Comparison of the sample complexity lower bounds for Bernoulli testing ($P_0=\text{Bern}(1/2)$ vs $P_1=\text{Bern}(1/2+\Delta)$) for $\Delta=0.01$.
   The top row corresponds to the asymmetric regime in which the Type~II error bound $\delta$ is fixed and the Type~I error bound $\varepsilon$ decays. Instead, in the bottom row, the left and middle columns correspond to the asymmetric regime in which $\varepsilon$ is fixed and $\delta$ decays, and the right column to the symmetric regime $\varepsilon=\delta$. Both axes are shown on a logarithmic scale. Solid lines indicate that the corresponding lower-bound is the largest.}
    \label{fig:bernoulli_sc_converse_1}
\end{figure}
}

\begin{figure}[t] 
    \centering
    \begin{subfigure}{\textwidth}
        \centering
        \includegraphics[width=\textwidth]{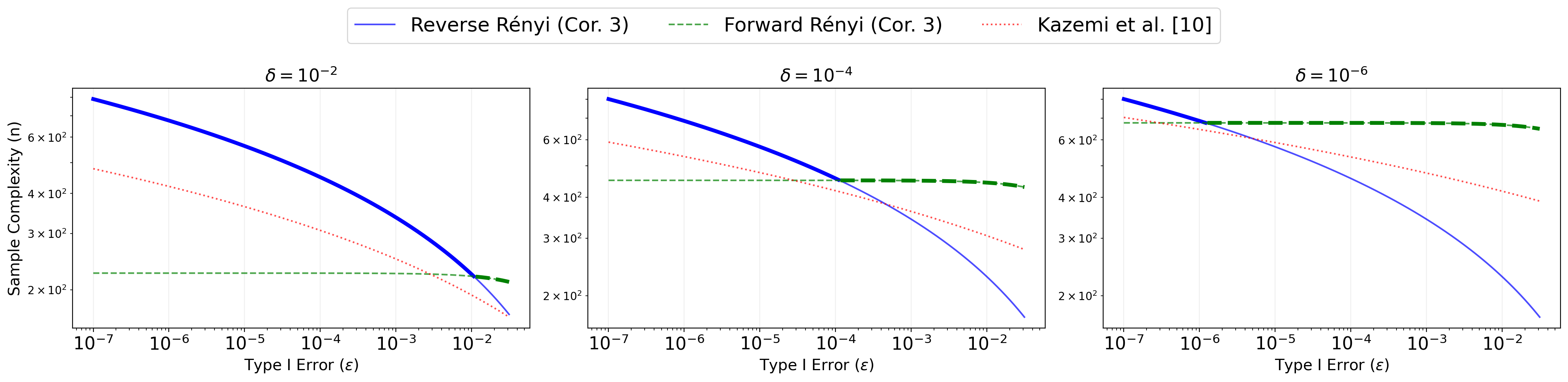}
    \end{subfigure}
    
    \par
    \par
    
    \begin{subfigure}{\textwidth}
        \centering
        \includegraphics[width=\textwidth]{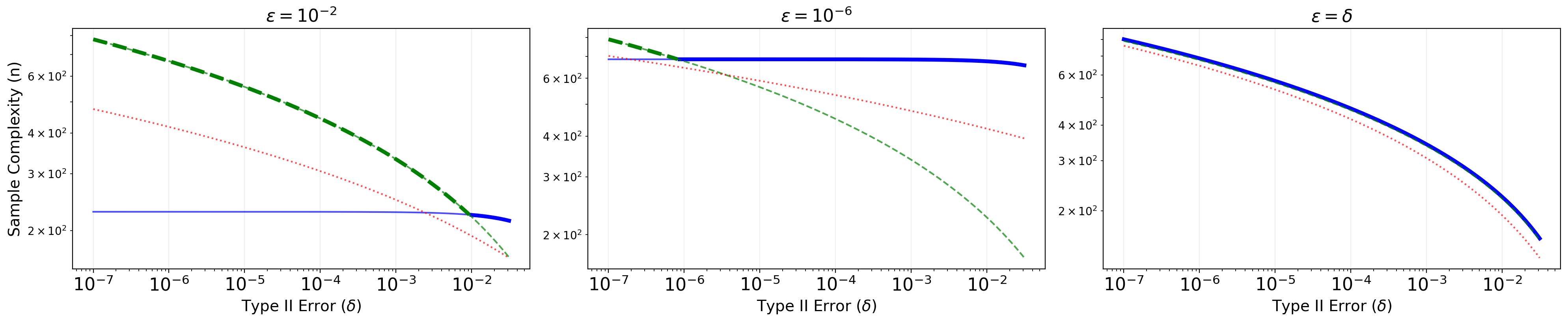}
    \end{subfigure}
    
    \caption{Comparison of the sample complexity lower bounds for Bernoulli testing ($P_0=\text{Bern}(1/2)$ vs $P_1=\text{Bern}(1/2+\Delta)$) for $\Delta=0.1$.
   The top row corresponds to the asymmetric regime in which the Type~II error bound $\delta$ is fixed and the Type~I error bound $\varepsilon$ decays. Instead, in the bottom row, the left and middle columns correspond to the asymmetric regime in which $\varepsilon$ is fixed and $\delta$ decays, and the right column to the symmetric regime $\varepsilon=\delta$. Both axes are shown on a logarithmic scale. Solid lines indicate that the corresponding lower-bound is the largest. The plotted bounds of Corollary~\ref{cor:lower_bound_sample_complexity} in~\eqref{eq:lb_sample_complexity_opt} have been optimised over $\lambda>1$.}
    \label{fig:bernoulli_sc_converse_2}
\end{figure}

\remove{
\begin{figure}[t] 
    \centering
    \begin{subfigure}{\textwidth}
        \centering
        \includegraphics[width=\textwidth]{UB_Bernoulli_sc_fixed_delta_0.png}
    \end{subfigure}
    
    \par
    \par
    
    \begin{subfigure}{\textwidth}
        \centering
        \includegraphics[width=\textwidth]{UB_Bernoulli_sc_fixed_epsilon_0.png}
    \end{subfigure}
    
    \caption{Comparison of the sample complexity upper bounds for Bernoulli testing ($P_0=\text{Bern}(1/2)$ vs $P_1=\text{Bern}(1/2+\Delta)$) for $\Delta=0.01$.
   The top row corresponds to the asymmetric regime in which the Type~II error bound $\delta$ is fixed and the Type~I error bound $\varepsilon$ decays. Instead, in the bottom row, the left and middle columns correspond to the asymmetric regime in which $\varepsilon$ is fixed and $\delta$ decays, and the right column to the symmetric regime $\varepsilon=\delta$. Both axes are shown on a logarithmic scale. Solid lines indicate that the corresponding upper bound is the smallest.}
    \label{fig:bernoulli_sc_upper_1}
\end{figure}
}

\begin{figure}[t] 
    \centering
    \begin{subfigure}{\textwidth}
        \centering
        \includegraphics[width=\textwidth]{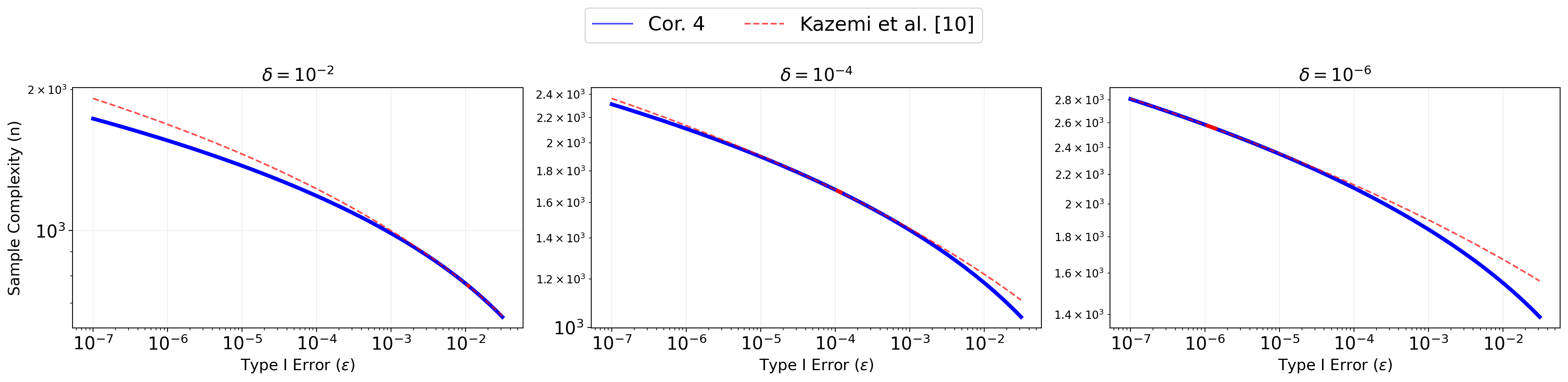}
    \end{subfigure}
    
    \par
    \par
    
    \begin{subfigure}{\textwidth}
        \centering
        \includegraphics[width=\textwidth]{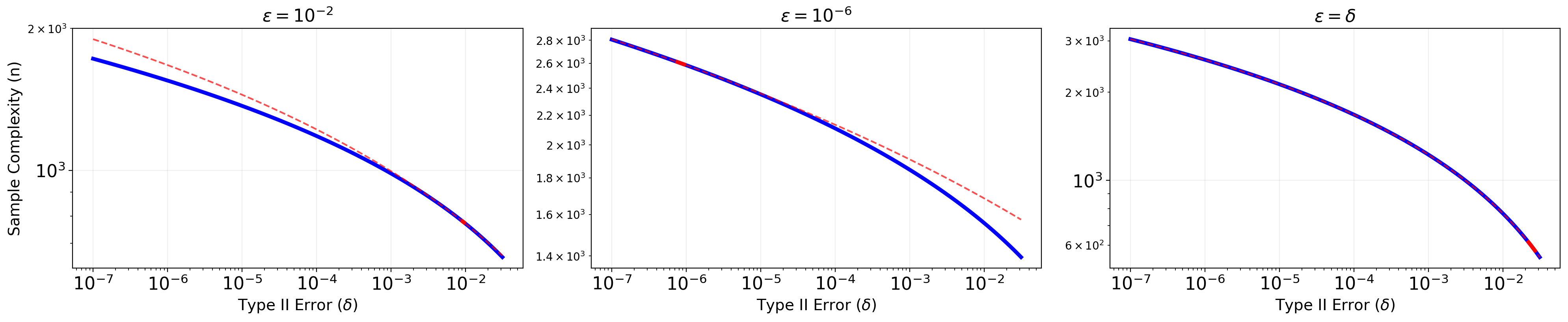}
    \end{subfigure}
    
    \caption{Comparison of the sample complexity upper bounds for Bernoulli testing ($P_0=\text{Bern}(1/2)$ vs $P_1=\text{Bern}(1/2+\Delta)$) for $\Delta=0.1$.
   The top row corresponds to the asymmetric regime in which the Type~II error bound $\delta$ is fixed and the Type~I error bound $\varepsilon$ decays. Instead, in the bottom row, the left and middle columns correspond to the asymmetric regime in which $\varepsilon$ is fixed and $\delta$ decays, and the right column to the symmetric regime $\varepsilon=\delta$. Both axes are shown on a logarithmic scale. Solid lines indicate that the corresponding upper bound is the smallest. The bound of Corollary~\ref{cor:upper_bound_sample_complexity} has been optimised over $\lambda\in(0,1)$.}
    \label{fig:bernoulli_sc_upper_2}
\end{figure}

We conclude the section with a numerical comparison of the sample complexity bounds derived in this work. Specifically, we compare them with the lower and upper bounds in  \eqref{eq:pensia_sc_eps<delta} and \eqref{eq:pensia_sc_eps>delta} of \cite{kazemi2025sample}.
For the comparison, we consider the same Bernoulli testing problem as in Section~\ref{sec:comparisons}. We omit the Gaussian setting, as it exhibits a similar behaviour. 

Figure~\ref{fig:bernoulli_sc_converse_2} compares the sample complexity lower bounds both in asymmetric and symmetric regimes. The numerical results highlight the strict dependence that the two bounds, established in Corollary~\ref{cor:lower_bound_sample_complexity}, have on the Type~I and Type~II error constraints. 
This behaviour is clearly reflected in the asymmetric regime. 
When the Type~II constraint is substantially more stringent than the Type~I constraint, i.e., when $\delta\ll\varepsilon$, the lower bound based on the forward divergence $D_\lambda(P_0\|P_1)$ provides the tighter bound. Conversely, when the Type~I constraint is substantially more stringent, i.e., $\varepsilon\ll\delta$, the bound based on the reverse divergence $D_\lambda(P_1\|P_0)$ becomes tighter.
Thus, the proposed bounds, especially in strongly asymmetric regimes, provide a tighter estimate of the sample complexity than the corresponding bounds of \cite{kazemi2025sample}. This behaviour is consistent with the asymptotic analysis in Remark~\ref{remark:sc_bounds_comparison}. Moreover, in this example, our bounds are substantially tighter than those in~\cite{kazemi2025sample} across the entire range of sample sizes.

Figure~\ref{fig:bernoulli_sc_upper_2} presents the corresponding comparison for the sample complexity upper bounds both in the asymmetric and symmetric regimes. Here too, the numerical results show that the proposed upper bound is particularly effective in strongly asymmetric regimes. In fact, in these regimes, the bound of Corollary \ref{cor:upper_bound_sample_complexity} provides
a tighter estimate of the sample complexity. This behaviour is consistent with the analytical comparison discussed in Remark~\ref{remark:sc_upper_comparison}.
On the other hand, in the symmetric regime, the two bounds exhibit similar behaviour, in agreement with the analysis in Remark~\ref{remark:sc_upper_comparison}. Overall, these numerical comparisons support the findings  that the R\'enyi-based bounds provide a particularly accurate characterisation of the sample complexity in asymmetric regimes.

\section{Local Differential Privacy}\label{sec:LDP}

In this section, we extend our results to the local differential privacy (LDP)
setting. Throughout, $\mathcal{X}$ and $\mathcal{Y}$ are finite.

In the classical binary hypothesis testing problem considered so far, we
observe a vector $\mathbf{X}=(X_1,\dots,X_n)$ of $n$ i.i.d. random variables,
each distributed according to $P$ on $(\mathcal{X}, \mathcal{B})$, and we wish to distinguish
the null hypothesis $H_0: P=P_0$ from the alternative $H_1: P=P_1$. Instead, in the local
privacy setting, $\mathbf{X}$ is no longer directly observable. This happens because the samples are
typically sensitive data, for instance, medical records, browsing histories, or device telemetry, and the individuals generating them may be unwilling to disclose them to any
central party. Local differential privacy \cite{dwork2014algorithmic} addresses this problem by requiring each observation to be
randomised on the user's own device before release.  Such a
constraint restricts the information available to the test and can therefore only
degrade the achievable error trade-off. This behaviour has already been investigated in the
literature \cite{duchi2018minimax, asoodeh2024contraction, pensia2024simple}.

In what follows, we show that the results established in the previous sections can be
generalized to the LDP setting. Among the various models of privacy, we focus in this work on pure LDP, which we now define.
\begin{definition}[\cite{pensia2024simple}]
A channel $T:\mathcal{X}\to\mathcal{Y}$ is
$\epsilon_{\rm dp}$-locally differentially private if, for every
$x,x'\in\mathcal{X}$ and every measurable $S\subseteq\mathcal{Y}$,
\begin{equation}\label{eq:defLDP}
    T(S\mid x)
    \leq
    e^{\epsilon_{\rm dp}}T(S\mid x').
\end{equation}
We denote the class of such channels by
$\mathcal{C}_{\epsilon_{\rm dp}}$.
\end{definition}
Each observation $X_i$ is passed independently through a channel $T_i\in\mathcal{C}_{\epsilon_{\rm dp}}$, and the analyst observes only the privatised output $Y_i$. For a probability measure $P$ on $\mathcal{X}$, we write $T_iP$ for the distribution of the output of $T_i$ when its input has distribution $P$. Hence, for fixed channels $T_1,\ldots,T_n$, the privatised observations are independent and, under hypothesis $H_j$, have joint distribution
\begin{equation}\label{eq:private_product_law}
    \bigotimes_{i=1}^nT_iP_j,
    \qquad j\in\{0,1\}.
\end{equation}
For this fixed choice of channels, the optimal Type~II error under the Type~I constraint $\varepsilon$ is therefore
\begin{equation*}
    \beta_\varepsilon\left(\bigotimes_{i=1}^nT_iP_0,\bigotimes_{i=1}^nT_iP_1\right).
\end{equation*}
Optimising over all admissible choices of the privacy channels that satisfy $\epsilon_{\rm dp}$-LDP, we define
\begin{equation}\label{eq:private_beta_definition}
    \beta_n^{\rm LDP}(\varepsilon, \epsilon_{\rm dp})=\inf_{\substack{T_i\in\mathcal{C}_{\epsilon_{\rm dp}}\\i=1,\ldots,n}}\beta_\varepsilon\left(\bigotimes_{i=1}^nT_iP_0,\bigotimes_{i=1}^nT_iP_1\right).
\end{equation}
This formulation keeps the testing problem unchanged: the original product measures $P_0^n$ and $P_1^n$ are simply replaced by their corresponding privatised product measures. Consequently, the finite-sample bounds developed earlier apply directly, with the Rényi divergences between the original distributions replaced by the Rényi divergences between their privatised versions.

We start with the converse. For any fixed choice of channels $T_1,\ldots,T_n\in\mathcal{C}_{\epsilon_{\rm dp}}$, the R\'enyi divergence between the corresponding product output distributions tensorises:
\begin{equation}\label{eq:tensorisation_ldp}
    D_\lambda\left(\bigotimes_{i=1}^nT_iP_1\,\middle\|\,\bigotimes_{i=1}^nT_iP_0\right)=\sum_{i=1}^nD_\lambda(T_iP_1\|T_iP_0).
\end{equation}

As such, using Theorem~\ref{thm:renyi_converse_bound_via_vr} gives, for every $\lambda>1$ and $c>0$,
\begin{equation}\label{eq:private_converse_before_infimum}
\beta_\varepsilon\left(
\bigotimes_{i=1}^nT_iP_0,
\bigotimes_{i=1}^nT_iP_1
\right)
\geq
\frac{
e^{(\lambda-1)c}
-
e^{\frac{\lambda-1}{\lambda}\sum_{i=1}^nD_\lambda(T_iP_1\|T_iP_0)}
\left(1-\varepsilon+\varepsilon e^{\lambda c}\right)^{\frac{\lambda-1}{\lambda}}
}{
e^{(\lambda-1)c}-1
}.
\end{equation}
Taking the infimum over the privacy channels in~\eqref{eq:private_converse_before_infimum} and using
\begin{equation}\label{eq:supremum}
    \sup_{T_1, \dots, T_n\in\mathcal C_{\epsilon_{\rm dp}}}\sum_{i=1}^nD_\lambda(T_iP_1\|T_iP_0)= n\sup_{T\in\mathcal C_{\epsilon_{\rm dp}}}D_\lambda(TP_1\|TP_0)
\end{equation}
 therefore yields
\begin{equation}\label{eq:private_converse_after_infimum}
\beta_n^{\rm LDP}(\varepsilon,\epsilon_{\rm dp})
\geq
\frac{
e^{(\lambda-1)c}
-
e^{\frac{\lambda-1}{\lambda}n
\sup_{T\in\mathcal C_{\epsilon_{\rm dp}}}
D_\lambda(TP_1\|TP_0)}
\left(1-\varepsilon+\varepsilon e^{\lambda c}\right)^{\frac{\lambda-1}{\lambda}}
}{
e^{(\lambda-1)c}-1
}.
\end{equation}
Thus, minimising over the privacy channels amounts to maximising the corresponding R\'enyi divergence. Interchanging $P_0$ and $P_1$ gives the analogous expression involving $\sup_{T\in\mathcal C_{\epsilon_{\rm dp}}}D_\lambda(TP_0\|TP_1)$. Since these two suprema are generally not available in closed form, we upper-bound them using the distribution-independent contraction coefficient associated with R\'enyi divergence, defined as
\begin{equation}
    \eta_\lambda(T) := \sup_{\substack{P, Q:\\ 0 < D_\lambda(P \| Q) < \infty}} \frac{D_\lambda(T P \| T Q)}{D_\lambda(P \| Q)}.
\end{equation}
By~\cite[Corollary~2 and Theorem~8]{vandenbroucque2026journal}, one has that $\eta_\lambda(T) \leq \eta_\infty(T)$ for every $\lambda\in[0,\infty]$ and moreover, on finite alphabets, \cite[Theorem~13]{vandenbroucque2026journal} shows that $\epsilon_{\rm dp}$-LDP is equivalent to
\begin{equation}\label{eq:ldp_contraction_assumption}
    \eta_\infty(T)\leq1-e^{-\epsilon_{\rm dp}}.
\end{equation}
Under \eqref{eq:ldp_contraction_assumption}, we thus obtain
\begin{equation}\label{eq:bound_supremum}
    \sup_{T \in \mathcal{C}_{\epsilon_{\rm dp}}} D_\lambda(T P_1 \| T P_0) \leq \sup_{T \in \mathcal{C}_{\epsilon_{\rm dp}}} \eta_{\infty}(T) D_\lambda(P_1\|P_0) \leq (1-e^{-\epsilon_{\rm dp}})D_\lambda(P_1\|P_0).
\end{equation}
The same bound holds after interchanging $P_0$ and $P_1$. Substituting into~\eqref{eq:private_converse_after_infimum} and optimising over $\lambda>1$ and $c>0$ gives the following converse.
\begin{corollary}\label{cor:LDP_lower_bound}
    Let $P_0\neq P_1$ be mutually absolutely continuous. Then, for any $\varepsilon\in(0,1)$ and $\epsilon_{\rm dp}\geq0$,
    \begin{align}
        \beta_n^{\rm LDP}(\varepsilon, \epsilon_{\rm dp}) \geq \max\Bigg\{  &\sup_{\lambda>1,c>0} \frac{e^{(\lambda-1)c}-e^{\frac{\lambda-1}{\lambda}n (1-e^{-\epsilon_{\rm dp}})D_\lambda(P_1 \|P_0)}\left(1-\varepsilon+\varepsilon e^{\lambda c}\right)^{\frac{\lambda-1}{\lambda}}}{e^{(\lambda-1)c}-1},\nonumber\\
        &\sup_{\lambda>1,c>0} \frac{e^{-n(1-e^{-\epsilon_{\rm dp}})D_\lambda(P_0 \| P_1)}\left(\varepsilon +(1-\varepsilon)e^{(\lambda-1)c}\right)^{\frac{\lambda}{\lambda-1}}-1}{e^{\lambda c}-1}\Bigg\}.
     \end{align}
\end{corollary}
As in Remark~\ref{rem:DPI_recover}, taking $c\to\infty$ gives
\begin{equation*}
      \beta_n^{\rm LDP}(\varepsilon, \epsilon_{\rm dp}) \geq  \max\bigg\{1-\inf_{\lambda>1}\left(\varepsilon\, e^{n(1-e^{-\epsilon_{\rm dp}})D_\lambda(P_1\| P_0)}\right)^{\frac{\lambda-1}{\lambda}}, \sup_{\lambda>1}\left\{(1-\varepsilon)^{\frac{\lambda}{\lambda-1}}\,e^{-n(1-e^{-\epsilon_{\rm dp}})D_\lambda(P_0\| P_1)}\right\}\bigg\}.
 \end{equation*}

 Figure~\ref{fig:bernoulli_privacy} provides a numerical comparison of the private lower bound established in Corollary~\ref{cor:LDP_lower_bound} with the corresponding non-private converse in Theorem~\ref{thm:renyi_converse_bound_via_vr}. The comparison illustrates the impact of the privacy constraint $\epsilon_{\rm dp}$ on the difficulty of the testing problem. As $\epsilon_{\rm dp}$ decreases towards zero, the privacy requirement becomes more stringent, resulting in much larger lower bounds.

Let $n^{\rm LDP}(\varepsilon,\delta,\epsilon_{\rm dp})$ be the smallest $n$ for which there exist mechanisms $T_1, \dots, T_n\in\mathcal{C}_{\epsilon_{\rm dp}}^n$ and a randomised test whose errors do not exceed $\varepsilon$ and $\delta$. More formally, define
\begin{equation}\label{eq:ldp_sample_complexity_definition}
n^{\rm LDP}(\varepsilon,\delta,\epsilon_{\rm dp}):=\min\left\{n\in\mathbb{N}:\beta_n^{\rm LDP}(\varepsilon, \epsilon_{\rm dp})\leq\delta\right\}.
\end{equation}
Substituting~\eqref{eq:bound_supremum} in the sample complexity bounds from Corollary~\ref{cor:lower_bound_sample_complexity} gives the following lower bound on the private sample complexity.
\begin{corollary}\label{cor:LDP_sc_lb}
 Suppose \eqref{eq:ldp_contraction_assumption} holds. Let $P_0\ne P_1$ be mutually absolutely continuous. For $\varepsilon,\delta\in(0,1)$ with $\varepsilon+\delta<1$ and $\epsilon_{\rm dp}>0$,
    \begin{align*}
        &n^{\rm LDP}(\varepsilon,\delta,\epsilon_{\rm dp})\geq\\
        &\left(\frac{1}{1-e^{-\epsilon_{\rm dp}}}\right)\cdot \max\Bigg\{\sup_{\lambda>1,c>0}\Bigg\{\frac{1}{D_{\lambda}(P_1\|P_0)}\Bigg[\frac{\lambda}{\lambda-1}\log\left((1-\delta)e^{(\lambda-1)c}+\delta\right)-\log\left(\varepsilon e^{\lambda c}+1-\varepsilon\right)\Bigg]\Bigg\},\nonumber\\
        &\hspace{9.5em}\sup_{\lambda>1,c>0}\Bigg\{\frac{1}{D_{\lambda}(P_0\|P_1)}\Bigg[\frac{\lambda}{\lambda-1}\log\left((1-\varepsilon)e^{(\lambda-1)c}+\varepsilon\right)-\log\left(\delta e^{\lambda c}+1-\delta\right)\Bigg]\Bigg\}\Bigg\}.
    \end{align*}
    Each supremum is restricted to orders for which the divergence in the denominator is finite and strictly positive.
\end{corollary}
Hence, the bound in Corollary~\ref{cor:LDP_sc_lb} is simply the nonprivate lower-bound from Corollary~\ref{cor:lower_bound_sample_complexity} multiplied by
\begin{equation*}
    \frac{1}{1-e^{-\epsilon_{\rm dp}}}>1.
\end{equation*}

\begin{figure}[t] 
    \centering 
    \begin{subfigure}{\textwidth}
        \centering
        \includegraphics[width=\textwidth]{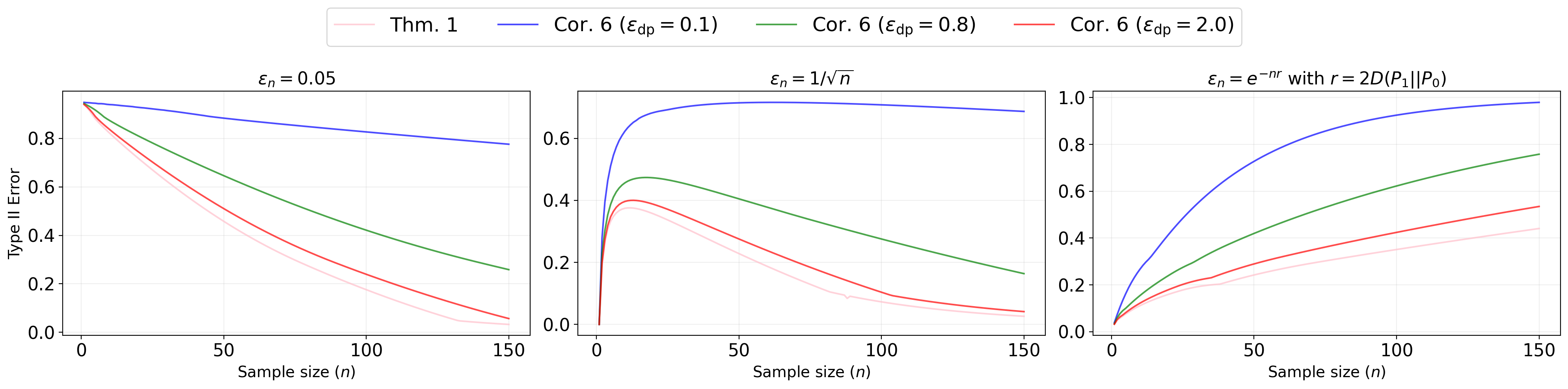}
    \end{subfigure}
    
    \caption{Comparison of the finite-sample converse bounds for Bernoulli testing ($P_0=\text{Bern}(1/2)$ vs $P_1=\text{Bern}(1/2+\Delta)$) for $\Delta=0.1$ under local differential privacy (LDP) constraints. The figure compares the non-private bound (Theorem~\ref{thm:renyi_converse_bound_via_vr}) against the private bound (Corollary~\ref{cor:LDP_lower_bound}) for $\epsilon_{\rm dp}\in\{0.1,0.8, 2.0\}$ across three Type~I error regimes: $\varepsilon_n=0.05$ (left column), $ \varepsilon_n = 1/\sqrt{n}$ (middle column), and $\varepsilon_n=e^{-nr}$ with $r=2D(P_1\|P_0)$ (right column). All plotted bounds have been optimised over $\lambda>0$ and $c>0$.}
    \label{fig:bernoulli_privacy}
\end{figure}

We now turn to achievability. Fix $\lambda\in(0,1)$ and let the same channel $T\in\mathcal C_{\epsilon_{\rm dp}}$ be applied to every observation. Applying Corollary~\ref{cor:beta_optimal_upper_bound_renyi} to $TP_0$ and $TP_1$ gives
\begin{equation}\label{eq:private_converse_before_infimum_1}
    \beta_\varepsilon\big((TP_0)^n,(TP_1)^n\big)
    \leq
    \lambda\left(\frac{1-\lambda}{\varepsilon}\right)^{\frac{1-\lambda}{\lambda}}
    \exp\left(-\frac{1-\lambda}{\lambda}nD_\lambda(TP_1\|TP_0)\right).
\end{equation}
The definition of $\beta_n^{\rm LDP}(\varepsilon,\epsilon_{\rm dp})$ allows the channel to vary across observations. Restricting the infimum to a common channel and then using~\eqref{eq:private_converse_before_infimum_1} yields
\begin{equation}
    \beta_n^{\rm LDP}(\varepsilon,\epsilon_{\rm dp})
    \leq
    \inf_{T\in\mathcal C_{\epsilon_{\rm dp}}}
    \beta_\varepsilon\big((TP_0)^n,(TP_1)^n\big)\leq
    \lambda\left(\frac{1-\lambda}{\varepsilon}\right)^{\frac{1-\lambda}{\lambda}}
    \exp\left(
        -\frac{1-\lambda}{\lambda}n
        \sup_{T\in\mathcal C_{\epsilon_{\rm dp}}}
        D_\lambda(TP_1\|TP_0)
    \right).
\end{equation}
Thus, the achievability bound is governed by $\sup_{T\in\mathcal{C}_{\epsilon_{\rm dp}}}D_\lambda(TP_1\|TP_0)$,
the same optimisation encountered in~\eqref{eq:supremum}.

For the converse, this quantity was upper-bounded uniformly over all admissible channels via contraction coefficients. For achievability, we will instead evaluate $D_\lambda(TP_1\|TP_0)$ for a particular admissible channel, thereby obtaining a lower bound on $\sup_{T\in\mathcal{C}_{\epsilon_{\rm dp}}}D_\lambda(TP_1\|TP_0)$ and hence an explicit upper bound on $\beta_n^{\rm LDP}(\varepsilon, \epsilon_{\rm dp})$. Our next aim is thus to identify channels which approach the supremum.

For $\lambda\in(0,1)$, maximising the R\'enyi divergence is equivalent to maximising the $f$-divergence
\begin{equation}
1-\exp\big(-(1-\lambda)D_\lambda(TP_1\|TP_0)\big) = 1-h_\lambda(TP_1,TP_0).
\end{equation}
Therefore, one can apply the extremal-mechanism result~\cite[Theorem 2]{kairouz2016extremal} which states that an optimal channel may be chosen from the class of staircase mechanisms. This structural result considerably restricts the class of channels that must be considered, but does not in general provide a closed-form optimiser. We consequently examine two explicit staircase mechanisms.

The first one is the binary mechanism. Its relevance is not limited to its analytical simplicity: by~\cite[Theorem 5]{kairouz2016extremal}, for every pair $P_0,P_1$, it maximises every $f$-divergence between the privatised distributions whenever $\epsilon_{\rm dp}$ is sufficiently small. It therefore also maximises $D_\lambda(TP_1\|TP_0)$ in this high-privacy regime. We recall that, letting $A=\big\{x\in\mathcal X:p_1(x)>p_0(x)\big\}$, the binary mechanism $T_{\rm B}:\mathcal X\to\{0,1\}$ is defined by
\begin{equation}\label{eq:T_A_def}
    T_{\rm B}(1 |x) := \begin{cases} \dfrac{e^{\epsilon_{\rm dp}}}{1+e^{\epsilon_{\rm dp}}}, & \text{if } x \in A\\[10pt] \dfrac{1}{1+e^{\epsilon_{\rm dp}}}, & \text{if } x \notin A \end{cases}, \qquad
T_{\rm B}(0\mid x):=1-T_{\rm B}(1\mid x).
\end{equation}
The channel $T_{\rm B}$ maps the original probability measures $P_1$ and $P_0$ into Bernoulli distributions over the privatised space $\mathcal{Y} = \{0, 1\}$:
\begin{equation}
    T_{\rm B} P_1 = \mathrm{Bern}(p_{\epsilon_{\rm dp}}) \quad \text{and}\quad T_{\rm B} P_0 = \mathrm{Bern}(q_{\epsilon_{\rm dp}}),
\end{equation}
where
\begin{equation}
    p_{\epsilon_{\rm dp}}
    :=\frac{e^{\epsilon_{\rm dp}}P_1(A)+P_1(A^{\mathsf c})}
    {1+e^{\epsilon_{\rm dp}}},
    \qquad
    q_{\epsilon_{\rm dp}}
    :=\frac{e^{\epsilon_{\rm dp}}P_0(A)+P_0(A^{\mathsf c})}
    {1+e^{\epsilon_{\rm dp}}}.
\end{equation}
Although the binary mechanism is well suited to the high-privacy regime, it first reduces each observation $x$ to the binary statistic $\mathbbm{1}_A(x)$. Consequently, the privatised distributions retain only the probabilities $P_0(A)$ and $P_1(A)$, while all information within $A$ and $A^{\mathsf c}$ is discarded. In particular, in the limit $\epsilon_{\rm dp}\to\infty$, we get
\begin{equation}
    \lim_{\epsilon_{\rm dp}\to\infty} D_\lambda(T_{\rm B}P_1\|T_{\rm B}P_0) = d_\lambda\bigl(P_1(A)\|P_0(A)\bigr)\leq D_\lambda(P_1\|P_0),
\end{equation}
where $d_\lambda(\cdot\|\cdot)$ denotes the binary R\'enyi divergence of order $\lambda$, and the inequality follows from data processing and is generally strict. Thus, the binary mechanism does not in general recover the non-private bound in the low-privacy limit.

To complement the binary mechanism, we consider the $k$-ary randomised response, where $k=\lvert\mathcal X\rvert$, which is the channel \(T_{\rm RR}:\mathcal X\to\mathcal X\) defined by
\begin{equation}\label{eq:kary_rr}
    T_{\rm RR}(y| x) := 
    \begin{cases}
        \dfrac{e^{\epsilon_{\rm dp}}}{e^{\epsilon_{\rm dp}}+k-1}
        & \text{if }y=x,\\[10pt]
        \dfrac{1}{e^{\epsilon_{\rm dp}}+k-1}
        & \text{if }y\neq x.
    \end{cases}
\end{equation}
This mechanism is particularly relevant in the low-privacy regime, as proven in~\cite[Theorem 8]{kairouz2016extremal}. Specifically, it maximises the KL divergence between the privatised distributions when $\epsilon_{\rm dp}$ is sufficiently large. Although this does not imply optimality for R\'enyi divergence of every order, the mechanism also has the important property that it converges to the identity channel as $\epsilon_{\rm dp}\to\infty$. Thus, the achievability bound obtained from $T_{\rm RR}$ recovers the corresponding non-private bound as the privacy constraint vanishes.

Because $T_{\rm B}$ and $T_{\rm RR}$ belong to the feasible set of $\epsilon_{\rm dp}$-LDP channels, they provide a lower bound to the supremum:
\begin{equation}\label{eq:lower_bound_supremum}
    D_\lambda^{\rm ach}(\epsilon_{\rm dp}):=\max\big\{D_\lambda(T_{\rm B}P_1\|T_{\rm B}P_0), D_\lambda(T_{\rm RR} P_1 \| T_{\rm RR} P_0)\big\}\leq\sup_{T \in \mathcal{C}_{\epsilon_{\rm dp}}} D_\lambda(TP_1 \| TP_0), 
\end{equation}
Applying Corollary~\ref{cor:beta_optimal_upper_bound_renyi} to the mechanism attaining the maximum defining $D_\lambda^{\rm ach}(\epsilon_{\rm dp})$ gives the following result, whose proof is in Appendix~\ref{app:cor:ldp_upper_bound}.
\begin{corollary}\label{cor:ldp_upper_bound}
    Let $P_0\neq P_1$ and let $\epsilon_{\rm dp}>0$. Then, for every $\varepsilon\in(0,1)$,
    \begin{equation}\label{eq:ldp_beta_upper_bound}
        \beta_n^{\rm LDP}(\varepsilon, \epsilon_{\rm dp})\leq\inf_{\lambda\in(0,1)}\lambda\left(\frac{1-\lambda}{\varepsilon}\right)^{\frac{1-\lambda}{\lambda}}\exp\left(-\frac{1-\lambda}{\lambda}nD_\lambda^{\rm ach}(\epsilon_{\rm dp})\right).
    \end{equation}
\end{corollary}
Inverting~\eqref{eq:ldp_beta_upper_bound} gives the corresponding sample-complexity bound.
\begin{corollary}\label{cor:ldp_sample_complexity_upper_bound}
    Let $P_0\neq P_1$ be mutually absolutely continuous. For every $\varepsilon,\delta\in(0,1)$ such that $\varepsilon+\delta<1$, and every $\epsilon_{\rm dp}>0$,
    \begin{equation}
        n^{\rm LDP}(\varepsilon,\delta,\epsilon_{\rm dp})\leq\inf_{\lambda\in(0,1)}\left\lceil\frac{\log\bigl(\frac{1-\lambda}\varepsilon\bigr)+\dfrac{\lambda}{1-\lambda}\log\bigl(\frac\lambda\delta\bigr)}{D_\lambda^{\rm ach}(\epsilon_{\rm dp})}\right\rceil.
    \end{equation}
\end{corollary}
Since the randomised response mechanism converges to the identity channel,
\begin{equation}\label{eq:dach_limit}
    \lim_{\epsilon_{\rm dp}\to\infty} D_\lambda^{\rm ach}(\epsilon_{\rm dp})= D_\lambda(P_1\|P_0).
\end{equation}
Consequently, the bounds above recover their non-private counterparts as $\epsilon_{\rm dp}\to\infty$.
\begin{figure}[t] 
    \centering 
    \begin{subfigure}{\textwidth}
        \centering
        \includegraphics[width=\textwidth]{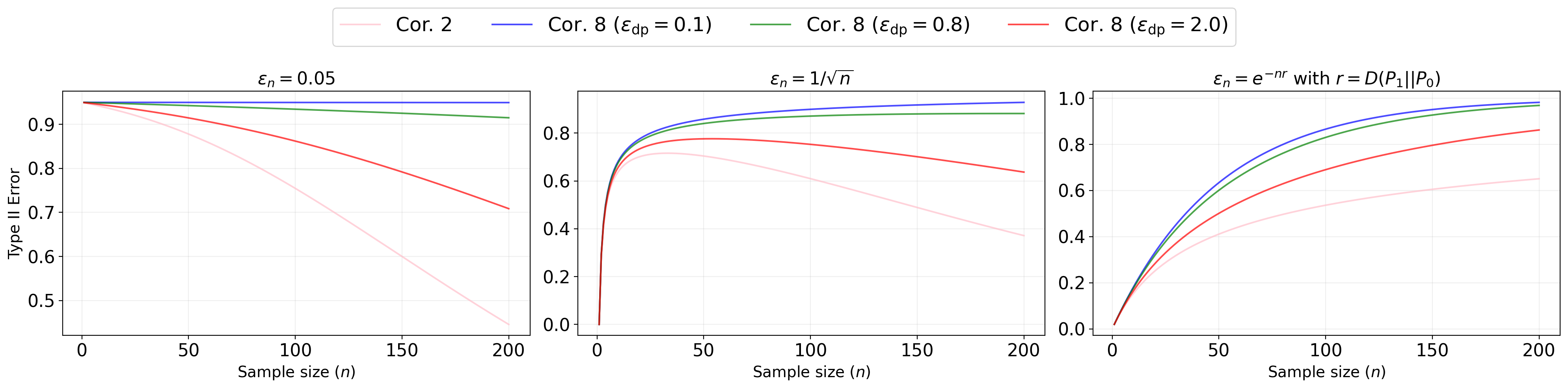}
    \end{subfigure}
    
    \caption{Comparison of the finite-sample upper bounds for Bernoulli testing ($P_0=\text{Bern}(1/2)$ vs $P_1=\text{Bern}(1/2+\Delta)$) for $\Delta=0.1$ under local differential privacy (LDP) constraints. The figure compares the non-private bound (Corollary~\ref{cor:beta_optimal_upper_bound_renyi}) against the private bound (Corollary~\ref{cor:ldp_upper_bound}) for $\epsilon_{\rm dp}\in\{0.1,0.8, 2.0\}$ across three Type~I error regimes: $\varepsilon_n=0.05$ (left column), $ \varepsilon_n = 1/\sqrt{n}$ (middle column), and $\varepsilon_n=e^{-nr}$ with $r=D(P_1\|P_0)$ (right column). All plotted bounds have been optimised over $\lambda\in(0,1)$.}
    \label{fig:ub_bernoulli_privacy}
\end{figure}

Figure~\ref{fig:ub_bernoulli_privacy} compares the private achievability bound of Corollary~\ref{cor:ldp_upper_bound} with its non-private counterpart in Corollary~\ref{cor:beta_optimal_upper_bound_renyi}. As in the converse comparison of Figure~\ref{fig:bernoulli_privacy}, the privacy constraint degrades the achievable error trade-off, yielding larger upper bounds as $\epsilon_{\rm dp}$ decreases towards zero.

Figure~\ref{fig:bernoulli_privacy_eps_func} examines the private achievability bound of Corollary~\ref{cor:ldp_upper_bound} from a complementary perspective by fixing the sample size and varying the privacy level. One can observe that each private upper bound is nonincreasing in $\epsilon_{\rm dp}$ and approaches the corresponding non-private bound as $\epsilon_{\rm dp}\to\infty$, consistently with~\eqref{eq:dach_limit}. Conversely, as the privacy constraint becomes more stringent ($\epsilon_{\rm dp}\to 0$), the privatised observations become uninformative and the bounds approach $1-\varepsilon_n$.
 

Finally, Figure~\ref{fig:bernoulli_privacy_opt_comp} compares the achievability bound of Corollary~\ref{cor:ldp_upper_bound} with the exact optimal private Type~II error. In this Bernoulli example, $\mathcal X=\{0,1\}$ and $A=\{1\}$, so $T_{\rm B}$ coincides with the binary randomised response. By~\cite[Theorem~18, specialised to pure LDP]{kairouz2016extremal}, this mechanism is optimal among all admissible privacy channels, and the exact private error can be evaluated using the Neyman--Pearson lemma. Under the fixed Type~I constraint, the upper bound remains close to $1-\varepsilon_n$ and becomes increasingly conservative as $n$ grows. For the two vanishing constraints, it captures the finite-sample behaviour of the optimal error more closely over the range shown. The agreement is particularly close in the exponential regime, where both the bound and the exact error increase towards one for $r=D(P_1\|P_0)$. We next establish a private phase-transition result that explains this behaviour.

 \begin{figure}[t] 
    \centering 
    \begin{subfigure}{\textwidth}
        \centering
        \includegraphics[width=\textwidth]{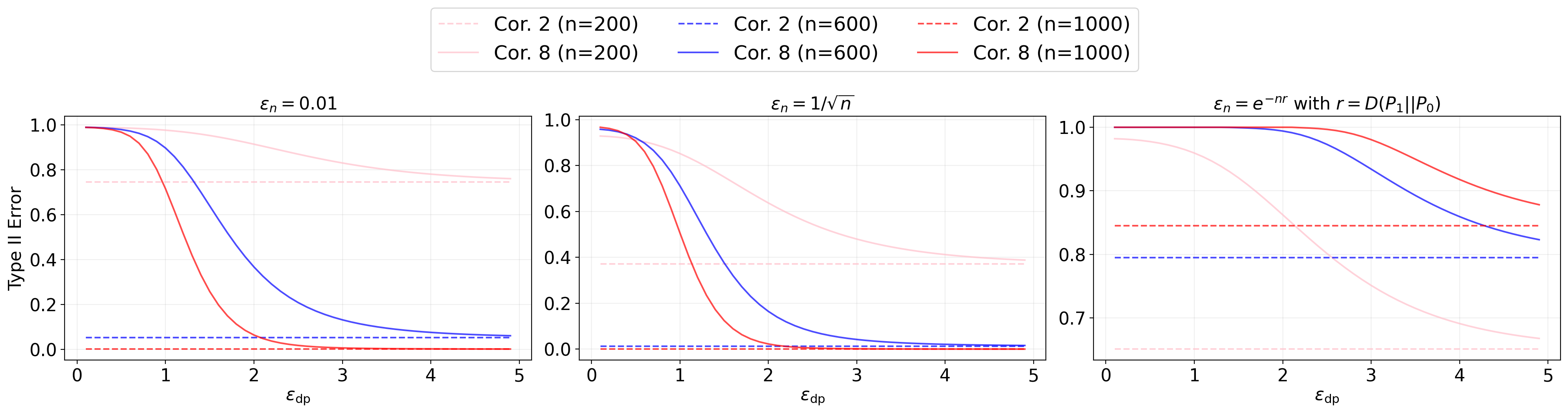}
    \end{subfigure}
    
    \caption{Finite-sample private achievability bounds for Bernoulli testing ($P_0=\text{Bern}(1/2)$ vs $P_1=\text{Bern}(1/2+\Delta)$) with $\Delta=0.1$, plotted as a function of the privacy constraint $\epsilon_{\rm dp}$. Specifically, we compare the private bounds (Corollary~\ref{cor:ldp_upper_bound}), in solid lines, against their corresponding non-private bound (Corollary~\ref{cor:beta_optimal_upper_bound_renyi}), in dashed lines, across three Type~I error regimes: $\varepsilon_n=0.01$ (left column), $\varepsilon_n = 1/\sqrt{n}$ (middle column), and $\varepsilon_n=e^{-nr}$ with $r=D(P_1\|P_0)$ (right column). Three fixed sample sizes $n\in\{200,600,1000\}$ are considered, distinguished by different colours. All plotted bounds have been optimised over $\lambda\in(0,1)$.}
    \label{fig:bernoulli_privacy_eps_func}
\end{figure}

 \begin{figure}[t] 
    \centering 
    \begin{subfigure}{\textwidth}
        \centering
        \includegraphics[width=\textwidth]{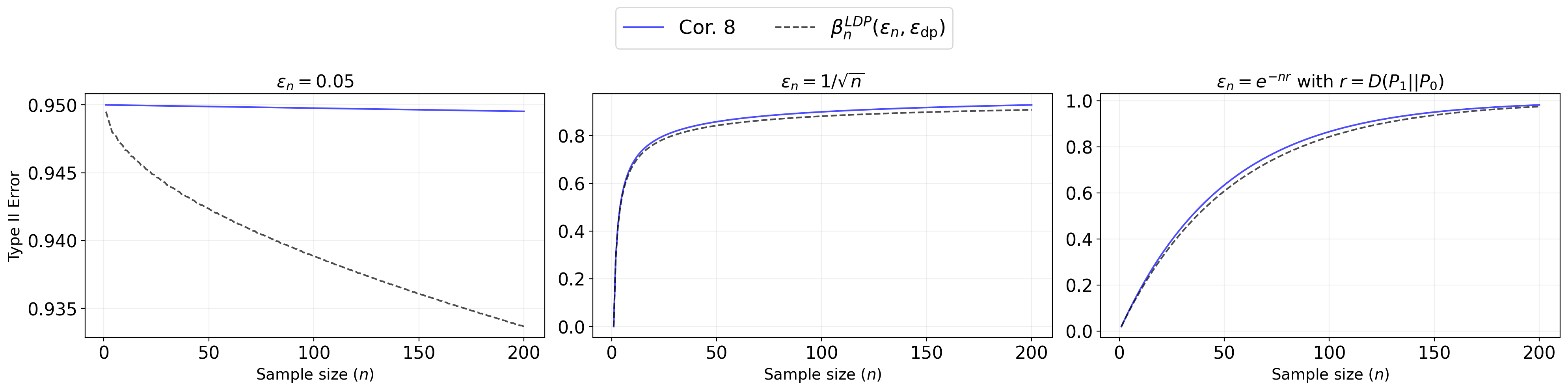}
    \end{subfigure}
    
    \caption{Comparison of the finite-sample private achievability bound of Corollary~\ref{cor:ldp_upper_bound} with the exact optimal private Type~II error $\beta_n^{\rm LDP}(\varepsilon_n,\epsilon_{\rm dp})$ for Bernoulli testing ($P_0=\text{Bern}(1/2)$ vs $P_1=\text{Bern}(1/2+\Delta)$) with $\Delta=0.1$ and $\epsilon_{\rm dp}=0.1$. Three Type~I error regimes are considered: $\varepsilon_n=0.05$ (left column), $\varepsilon_n = 1/\sqrt{n}$ (middle column), and $\varepsilon_n=e^{-nr}$ with $r=D(P_1\|P_0)$ (right column). The bound of Corollary~\ref{cor:ldp_upper_bound} has been optimised over $\lambda\in(0,1)$.}
    \label{fig:bernoulli_privacy_opt_comp}
\end{figure}

 To formulate the corresponding phase-transition threshold, define the achievable KL divergence
\begin{equation}\label{eq:ldp_achievable_kl}
    D^{\rm ach}(\epsilon_{\rm dp}):=\lim_{\lambda\uparrow1}D_\lambda^{\rm ach}(\epsilon_{\rm dp})=\max\big\{D(T_{\rm B}P_1\|T_{\rm B}P_0),D(T_{\rm RR}P_1\|T_{\rm RR}P_0)\big\}.
\end{equation}
We now combine \Cref{cor:LDP_lower_bound} with \Cref{cor:ldp_upper_bound} to examine exponentially decaying Type~I constraints. The following theorem identifies regimes in which the optimal private Type~II error converges exponentially to zero or one, with a possible gap between the corresponding thresholds. Its proof is deferred to Appendix~\ref{app:th:ldp_phase_transition}.
\begin{theorem}
\label{th:ldp_phase_transition}
Let $P_0\neq P_1$ be mutually absolutely continuous probability
measures on a finite alphabet. Fix $0<\epsilon_{\mathrm{dp}}<\infty$,
and consider the independent-channel LDP model defined above.
For $\varepsilon_n=e^{-nr}$ with $r>0$, the optimal private
Type~II error satisfies
\begin{equation}\label{eq:ldp_phase_transition}
\lim_{n\to\infty}\beta_n^{\rm LDP}(\varepsilon_n, \epsilon_{\rm dp})
=
\begin{cases}
1, & \text{if } r>r_{\mathrm{con}},\\
0, & \text{if } 0<r<r_{\mathrm{ach}},
\end{cases}
\end{equation}
where $r_{\mathrm{ach}}:=D^{\rm ach}(\epsilon_{\rm dp})$ and $r_\mathrm{con}:=(1-e^{-\epsilon_{\rm dp}})D(P_1\|P_0).$
More precisely, for every $n\ge1$,
\begin{equation}\label{eq:ldp_phase_converse}
\beta_n^{\rm LDP}(\varepsilon_n, \epsilon_{\rm dp})
\ge
1-\inf_{\lambda>1}
\exp\!\left(
-\frac{\lambda-1}{\lambda}n
\left(
r-(1-e^{-\epsilon_{\mathrm{dp}}})D_\lambda(P_1\|P_0)
\right)
\right),
\end{equation}
and
\begin{equation}\label{eq:ldp_phase_achievability}
\beta_n^{\rm LDP}(\varepsilon_n, \epsilon_{\rm dp})
\le
\inf_{\lambda\in(0,1)}
\exp\!\left(
-\frac{1-\lambda}{\lambda}n
\big(D_\lambda^{\rm ach}(\epsilon_{\rm dp})-r\big)
\right).
\end{equation}
\end{theorem}
 \begin{remark}
     Since $r_{\mathrm{ach}}\le r_{\mathrm{con}}$, these bounds identify achievability and strong-converse regimes, while leaving the interval $[r_{\mathrm{ach}},r_{\mathrm{con}}]$ unresolved. Thus, the possible gap between the thresholds does not prevent establishing vanishing error below $r_{\mathrm{ach}}$, but it precludes identifying a sharp transition from these explicit bounds alone. We conjecture that the exact threshold occurs at $r_\star := \sup_{T\in \mathcal{C}_{\epsilon_{\rm dp}}} D(TP_1 \|TP_0).$
 \end{remark}

The bounds established in this section provide valuable insights into the local differential privacy setting, quantifying how the privacy constraint strictly increases the difficulty of the hypothesis testing problem.

The effect of privacy enters the converse through the distribution-independent contraction bound, whereas the achievability result uses the explicit lower bound $D_\lambda^{\rm ach}(\epsilon_{\rm dp})$ obtained from the binary mechanism and $k$-ary randomised response. The binary mechanism is optimal in the high-privacy regime, while randomised response attains $\sup_{T\in \mathcal{C}_{\epsilon_{\rm dp}}} D(TP_1 \|TP_0)$ in the low-privacy regime and recovers the non-private bound as $\epsilon_{\rm dp}\to\infty$. Outside these settings, however, the two mechanisms need not attain $\sup_{T\in \mathcal{C}_{\epsilon_{\rm dp}}} D_\lambda(TP_1 \|TP_0)$. The possible gap between the achievability and converse bounds may therefore arise from both the contraction estimate used for the converse and the restriction to these two mechanisms for achievability. Obtaining matching bounds throughout the full range of privacy levels remains an open question.


\section{Conclusion}\label{sec:conclusion}

The primary contribution of this paper is the derivation of a finite-sample R\'enyi converse bound  for asymmetric binary hypothesis testing. Established via a variational formulation approach, the bound leverages both divergence directions, tensorises under independent observations, and recovers the standard data-processing converses as a limiting case of its optimisation parameters ($c\to \infty$).  Furthermore, by applying a similar variational approach, we also recover total-variation, $E_\gamma$, and Hellinger-based converses within a unified comparison framework for general $f$-divergences.

In terms of R\'enyi divergences of order $\lambda\in(0,1)$, we also provide a complementary achievability bound.

Under an exponentially decaying Type I error constraint $\varepsilon_n = e^{-nr}$, we demonstrate that the two R\'enyi bounds yield explicit exponential control of $\beta_n(\varepsilon_n)$ and $1-\beta_n(\varepsilon_n)$ on the respective sides of the known phase-transition threshold $D(P_1\|P_0)$.

Subsequently, by inverting these new finite-sample bounds, we obtain lower and upper bounds on sample complexity. Moreover, under regularity conditions, we demonstrate that these bounds are asymptotically matching in strongly asymmetric regimes.



Finally, we extend our results to the local differential privacy setting. A distribution-independent contraction coefficient yields finite-sample converse and sample-complexity lower bounds, while the binary mechanism and $k$-ary randomised response provide explicit achievability and sample-complexity upper bounds. These mechanisms are suited to complementary privacy regimes, and moreover the $k$-ary randomised response ensures that the non-private achievability bound is recovered as the privacy constraint vanishes. Combining the private converse and achievability bounds also identifies regimes in which the optimal Type~II error converges exponentially to zero or one, with a possible gap between the corresponding thresholds.

Several questions remain open for future research. In the non-private setting, the finite-sample behaviour at the critical rate $r=D(P_1\|P_0)$ is not resolved by our bounds. Furthermore, in the locally private setting, the present results do not determine the exact phase-transition threshold throughout the full range of privacy levels. Closing this gap would require improving one or both sides of the analysis: finding an admissible channel that strengthens the achievability bound beyond the two mechanisms considered here, or tightening the contraction-based converse.

\newpage
\appendices

\section{Proofs for the Converse and Achievability via R\'enyi Divergence}
\subsection{Proof of Theorem \ref{thm:renyi_converse_bound_via_vr}}\label{app:thm:renyi_converse_bound_via_vr}
    Let $\phi$ be a possibly randomised optimal rule achieving $\beta_n(\varepsilon)$, with Type~I error $\alpha_n\leq\varepsilon$. Adjoin an independent $U\sim\mathrm{Unif}[0,1]$ under both hypotheses, let $\mathsf U$ denote its law, and define
    \begin{equation*}
        A=\{(\mathbf{x},u):u\leq\phi(\mathbf{x})\}.
    \end{equation*}
    Then $P_0^n(A)=\alpha_n$ and $P_1^n(A^c)=\beta_n(\varepsilon)$ under the augmented laws. Moreover,
    \begin{equation*}
        D_\lambda(P_1^n\otimes\mathsf U\|P_0^n\otimes\mathsf U)
        =D_\lambda(P_1^n\|P_0^n).
    \end{equation*}
    In what follows, probabilities and expectations involving $A$ are taken under the augmented laws, with the common auxiliary factor suppressed from the notation.

    Setting $g=c\mathbbm{1}_{A}$ for an arbitrary constant $c>0$ in \eqref{eq:Renyi_var_formula}, we obtain the following lower bound for $D_{\lambda}(P_1^n\|P_0^n)$:
    \begin{align}
        D_\lambda(P_1^n\|P_0^n)&\geq \frac{\lambda}{\lambda-1} \log \mathbb{E}_{P_1^n}[e^{(\lambda-1)c\mathbbm{1}_{A}}] - \log \mathbb{E}_{P_0^n}[e^{\lambda c\mathbbm{1}_{A}}]\nonumber\\
        &=\frac{\lambda}{\lambda-1}\log\left(\int_A e^{(\lambda-1)c}dP_1^n +\int_{A^c} dP_1^n\right)-\log\left(\int_A e^{\lambda c}dP_0^n +\int_{A^c} dP_0^n\right)\nonumber\\
        &=\frac{\lambda}{\lambda-1}\log\left(e^{(\lambda-1)c}P_1^n(A)+P_1^n(A^c)\right)-\log\left(e^{\lambda c}P_0^n(A)+P_0^n(A^c)\right).\label{eq:first_step_P_1_P_0}
    \end{align}
    Recalling that $P_0^n(A)=\alpha_n$ and $P_1^n(A^c)=\beta_n(\varepsilon)$, and applying the tensorisation property for i.i.d. observations ($D_\lambda(P_1^n\|P_0^n)=nD_\lambda(P_1\|P_0)$),  \eqref{eq:first_step_P_1_P_0} becomes
    \begin{equation}\label{eq:second_step_P_1_P_0}
        nD_\lambda(P_1\|P_0)\geq \frac{\lambda}{\lambda-1}\log(\beta_n(\varepsilon)+(1-\beta_n(\varepsilon))e^{(\lambda-1)c})-\log(1-\alpha_n+\alpha_ne^{\lambda c}).
    \end{equation}
    Rearranging \eqref{eq:second_step_P_1_P_0} to isolate the term containing $\beta_n(\varepsilon)$ yields
    \begin{align}
        \frac{\lambda}{\lambda-1}\log(\beta_n(\varepsilon)+(1-\beta_n(\varepsilon))e^{(\lambda-1)c})&\leq nD_\lambda(P_1\|P_0)+\log(1-\alpha_n+\alpha_ne^{\lambda c}).\label{eq:third_step_P_1_P_0}
    \end{align}
    Assuming $\lambda>1$, we multiply by $\frac{\lambda-1}{\lambda}$ and exponentiate  both sides of the inequality \eqref{eq:third_step_P_1_P_0}, obtaining
    \begin{equation}
        \beta_n(\varepsilon)+(1-\beta_n(\varepsilon))e^{(\lambda-1)c}\leq e^{\frac{\lambda-1}{\lambda}nD_\lambda(P_1\|P_0)}(1-\alpha_n+\alpha_ne^{\lambda c})^{\frac{\lambda-1}{\lambda}}.
    \end{equation}
    Grouping the $\beta_n(\varepsilon)$ terms on the left side, we get
    \begin{equation}\label{eq:fourth_step_P_1_P_0}
        \beta_n(\varepsilon)(1-e^{(\lambda-1)c})\leq e^{\frac{\lambda-1}{\lambda}nD_\lambda(P_1\|P_0)}(1-\alpha_n+\alpha_ne^{\lambda c})^{\frac{\lambda-1}{\lambda}}-e^{(\lambda-1)c}.
    \end{equation}
    For $c>0$ and $\lambda>1$, the term $1-e^{(\lambda-1)c}<0$. Consequently, dividing both sides of \eqref{eq:fourth_step_P_1_P_0} by $1-e^{(\lambda-1)c}$, we get
    \begin{align}
        \beta_n&\geq \frac{e^{\frac{\lambda-1}{\lambda}nD_\lambda(P_1\|P_0)}(1-\alpha_n+\alpha_ne^{\lambda c})^{\frac{\lambda-1}{\lambda}}-e^{(\lambda-1)c}}{1-e^{(\lambda-1)c}}\nonumber\\
        &=\frac{e^{(\lambda-1)c}-e^{\frac{\lambda-1}{\lambda}nD_\lambda(P_1\|P_0)}(1-\alpha_n+\alpha_ne^{\lambda c})^{\frac{\lambda-1}{\lambda}}}{e^{(\lambda-1)c}-1}.\label{eq:last_step_P_1_P_0}
    \end{align}
    Recalling that $\alpha_n\leq\varepsilon$, and since for $c>0$ and $\lambda>1$, the term $e^{\lambda c}>1$, it follows that the right-hand side of \eqref{eq:last_step_P_1_P_0} decreases as $\alpha_n$ increases. This implies
    \begin{equation}\label{eq:final_formula_P_1_P_0}
        \beta_n(\varepsilon)\geq \frac{e^{(\lambda-1)c}-e^{\frac{\lambda-1}{\lambda}nD_\lambda(P_1\|P_0)}(1-\varepsilon+\varepsilon e^{\lambda c})^{\frac{\lambda-1}{\lambda}}}{e^{(\lambda-1)c}-1}.
    \end{equation}
    Because the inequality \eqref{eq:final_formula_P_1_P_0} holds for any arbitrary $c>0$ and $\lambda>1$, taking the supremum over these two parameters gives

    \begin{equation}\label{eq:converse_renyi_P_1_P_0}
        \beta_n(\varepsilon)\geq\sup_{\lambda>1,c>0} \Bigg\{\frac{e^{(\lambda-1)c}-e^{\frac{\lambda-1}{\lambda}nD_\lambda(P_1\|P_0)}(1-\varepsilon+\varepsilon e^{\lambda c})^{\frac{\lambda-1}{\lambda}}}{e^{(\lambda-1)c}-1}\Bigg\}.
    \end{equation}

    To obtain the corresponding bound in terms of $D_\lambda(P_0\|P_1)$, we proceed in a similar manner. By setting $g=c(1-\mathbbm{1}_{A})$ for $c>0$ in \eqref{eq:Renyi_var_formula}, we obtain the following lower bound for $D_{\lambda}(P_0^n\|P_1^n)$:
    \begin{align}
        D_\lambda(P_0^n\|P_1^n)&\geq \frac{\lambda}{\lambda-1} \log \mathbb{E}_{P_0^n}[e^{(\lambda-1)c(1-\mathbbm{1}_{A})}] - \log \mathbb{E}_{P_1^n}[e^{\lambda c(1-\mathbbm{1}_{A})}]\nonumber\\
        &=\frac{\lambda}{\lambda-1}\log\left(\int_{A^c} e^{(\lambda-1)c}dP_0^n +\int_{A} dP_0^n\right)-\log\left(\int_{A^c} e^{\lambda c}dP_1^n +\int_{A} dP_1^n\right)\nonumber\\
        &=\frac{\lambda}{\lambda-1}\log\left(e^{(\lambda-1)c}P_0^n(A^c)+P_0^n(A)\right)-\log\left(e^{\lambda c}P_1^n(A^c)+P_1^n(A)\right).\label{eq:first_step_P_0_P_1}
    \end{align}
    Recalling that $P_0^n(A)=\alpha_n$ and $P_1^n(A^c)=\beta_n(\varepsilon)$, and applying the tensorisation property for i.i.d. observations ($D_\lambda(P_0^n\|P_1^n)=nD_\lambda(P_0\|P_1)$),  \eqref{eq:first_step_P_0_P_1} becomes
    \begin{equation}\label{eq:second_step_P_0_P_1}
        nD_\lambda(P_0\|P_1)\geq \frac{\lambda}{\lambda-1}\log\left(\alpha_n+(1-\alpha_n)e^{(\lambda-1)c}\right)-\log\left(1-\beta_n(\varepsilon)+\beta_n(\varepsilon)e^{\lambda c}\right).
    \end{equation}
    Rearranging \eqref{eq:second_step_P_0_P_1} to isolate the term containing $\beta_n(\varepsilon)$ yields
    \begin{align}
        \log\left(1-\beta_n(\varepsilon)+\beta_n(\varepsilon)e^{\lambda c}\right)&\geq \frac{\lambda}{\lambda-1}\log\left(\alpha_n+(1-\alpha_n)e^{(\lambda-1)c}\right)-nD_\lambda(P_0\|P_1).\label{eq:third_step_P_0_P_1}
    \end{align}
    Exponentiating both sides of inequality \eqref{eq:third_step_P_0_P_1} gives
    \begin{equation}
        1-\beta_n(\varepsilon)+\beta_n(\varepsilon)e^{\lambda c}\geq\left(\alpha_n+(1-\alpha_n)e^{(\lambda-1)c}\right)^{\frac{\lambda}{\lambda-1}}e^{-nD_\lambda(P_0\|P_1)}.
    \end{equation}
    Grouping the $\beta_n$ terms on the left side, we get
    \begin{equation}\label{eq:fourth_step_P_0_P_1}
        \beta_n(\varepsilon)(e^{\lambda c}-1)\geq \left(\alpha_n+(1-\alpha_n)e^{(\lambda-1)c}\right)^{\frac{\lambda}{\lambda-1}}e^{-nD_\lambda(P_0\|P_1)}-1.
    \end{equation}
    For $c>0$ and $\lambda>1$, the term $e^{\lambda c}-1>0$. Dividing by $e^{\lambda c}-1$ both sides of \eqref{eq:fourth_step_P_0_P_1}, we get
    \begin{align}
        \beta_n&\geq \frac{\left(\alpha_n+(1-\alpha_n)e^{(\lambda-1)c}\right)^{\frac{\lambda}{\lambda-1}}e^{-nD_\lambda(P_0\|P_1)}-1}{e^{\lambda c}-1}\nonumber\\
        &=\frac{\bigg(e^{(\lambda-1)c}-\alpha_n\left(e^{(\lambda-1)c}-1\right)\bigg)^{\frac{\lambda}{\lambda-1}}e^{-nD_\lambda(P_0\|P_1)}-1}{e^{\lambda c}-1}.\label{eq:last_step_P_0_P_1}
    \end{align}
    For $c>0$ and $\lambda>1$, the term $e^{(\lambda-1)c}-1$ is strictly positive, and therefore the right-hand side of \eqref{eq:last_step_P_0_P_1} decreases as $\alpha_n$ increases. Thus, since $\alpha_n\leq\varepsilon$, we obtain the following inequality
    \begin{align}\label{eq:final_formula_P_0_P_1}
        \beta_n(\varepsilon)&\geq\frac{\bigg(e^{(\lambda-1)c}-\varepsilon\left(e^{(\lambda-1)c}-1\right)\bigg)^{\frac{\lambda}{\lambda-1}}e^{-nD_\lambda(P_0\|P_1)}-1}{e^{\lambda c}-1}\nonumber\\
        &=\frac{e^{-nD_\lambda(P_0\|P_1)}\left(\varepsilon +(1-\varepsilon)e^{(\lambda-1)c}\right)^{\frac{\lambda}{\lambda-1}}-1}{e^{\lambda c}-1}.
    \end{align}
    The above inequality \eqref{eq:final_formula_P_0_P_1} holds for any $c>0$ and $\lambda>1$. Therefore, taking the supremum over both parameters yields

    \begin{equation}\label{eq:converse_renyi_P_0_P_1}
        \beta_n(\varepsilon)\geq\sup_{\lambda>1,c>0}\Bigg\{\frac{e^{-nD_\lambda(P_0\|P_1)}\left(\varepsilon +(1-\varepsilon)e^{(\lambda-1)c}\right)^{\frac{\lambda}{\lambda-1}}-1}{e^{\lambda c}-1}\Bigg\}.
    \end{equation}

    Combining \eqref{eq:converse_renyi_P_1_P_0} and \eqref{eq:converse_renyi_P_0_P_1} immediately gives \eqref{eq:renyi_converse_bound_via_vr}, which concludes the proof.

\subsection{Proof of Corollary \ref{cor:semi-closed_form}}\label{app:semi-closed_form}
    Fix $\varepsilon\in(0,1)$. We analyse the optimisation over $c>0$ for the two branches in \eqref{eq:renyi_converse_bound_via_vr} separately.
    
    For the first branch, fix $\lambda\in\Lambda_R$ and let
    \begin{equation}
        R_\lambda(c)=\frac{e^{(\lambda-1)c}-e^{\frac{\lambda-1}{\lambda}nD_\lambda(P_1\|P_0)}\left(1-\varepsilon+\varepsilon e^{\lambda c}\right)^{\frac{\lambda-1}{\lambda}}}{e^{(\lambda-1)c}-1},\quad\text{for } c>0.
    \end{equation}
    
    Differentiating $R_\lambda(c)$ with respect to $c$, we obtain
    
    \[\frac{d}{dc}R_\lambda(c) = \frac{(\lambda-1)e^{(\lambda-1)c}}{\left(e^{(\lambda-1)c}-1\right)^2} \left[ e^{\frac{\lambda-1}{\lambda}nD_\lambda(P_1\|P_0)} \frac{1-\varepsilon+\varepsilon e^c}{\left(1-\varepsilon+\varepsilon e^{\lambda c}\right)^{1/\lambda}} - 1 \right].\]
    
    Since  $\frac{(\lambda-1)e^{(\lambda-1)c}}{(e^{(\lambda-1)c}-1)^2} > 0$ for all $c > 0$ and $\lambda > 1$, the sign of $\frac{d}{dc}R_\lambda(c)$ is entirely determined by the sign of 
    \begin{equation}\label{eq:app:first_expr}
        e^{\frac{\lambda-1}{\lambda}nD_\lambda(P_1\|P_0)} - \frac{\left(1-\varepsilon+\varepsilon e^{\lambda c}\right)^{1/\lambda}}{1-\varepsilon+\varepsilon e^c}.
    \end{equation}
    To study the sign of the expression \eqref{eq:app:first_expr}, 
    let 
    \[\psi(c) = \frac{\left(1-\varepsilon+\varepsilon e^{\lambda c}\right)^{1/\lambda}}{1-\varepsilon+\varepsilon e^c}.\]
    Taking the derivative of the logarithm of $\psi(c)$, we get
    \begin{align}
        \frac{d}{dc}\log \psi(c) &= \frac{\varepsilon e^{\lambda c}}{1-\varepsilon+\varepsilon e^{\lambda c}}-\frac{\varepsilon e^c}{1-\varepsilon+\varepsilon e^c}\nonumber\\
        &=\frac{\varepsilon(1-\varepsilon)(e^{\lambda c}-e^c)}{(1-\varepsilon+\varepsilon e^{\lambda c})(1-\varepsilon+\varepsilon e^{c})}.
    \end{align}
    Since $\lambda>1$ and $c>0$, we have that $e^{\lambda c}>e^{c}>1$. Therefore, $\frac{d}{dc}\log \psi(c)>0$, and this implies that $\psi(c)$ is strictly increasing in $c$ on $(0,\infty)$. Moreover, we have
    \begin{equation}
        \lim_{c\to0^+} \psi(c)=1\quad\text{and}\quad\lim_{c\to \infty} \psi(c)=\varepsilon^{-\frac{\lambda-1}{\lambda}}.
    \end{equation}
    We therefore distinguish three cases according to the value of  $nD_\lambda(P_1\|P_0)$:
    \begin{itemize}
        \item Case $(nD_\lambda(P_1\|P_0)=0)$: In this case, since $e^{\frac{\lambda-1}{\lambda}nD_{\lambda}(P_1\|P_0)}=1$, and since $\psi(c)>1$ for all $c>0$, the derivative $\frac{d}{dc}R_\lambda(c)$ is strictly negative for all $c>0$. Thus, the supremum is attained as $c\to 0^+$, yielding
        
        \[\sup_{c>0}R_\lambda(c)=\lim_{c\to 0^+}R_\lambda(c)=1-\varepsilon.\]
    
        \item  Case $(0<nD_\lambda(P_1\|P_0)<\log(1/\varepsilon))$: In this case, we have
        \[
            1<e^{\frac{\lambda-1}{\lambda}nD_\lambda(P_1\|P_0)}<\varepsilon^{-\frac{\lambda-1}{\lambda}}.
        \]
        Since $\psi(c)$ is strictly increasing from 1 to $\varepsilon^{-\frac{\lambda-1}{\lambda}}$, there exists a unique positive solution $c_{\lambda,R}\in(0,\infty)$ to the equation $\psi(c)=e^{\frac{\lambda-1}{\lambda}nD_\lambda(P_1\|P_0)}$, which  can be written as
    
        \begin{equation}\label{eq:app:stazionary_eq_1}
            \left(1-\varepsilon+\varepsilon e^{\lambda c_{\lambda,R}}\right)^{1/\lambda}=e^{\frac{\lambda-1}{\lambda}nD_\lambda(P_1\|P_0)}\left(1-\varepsilon+\varepsilon e^{ c_{\lambda,R}}\right).
        \end{equation}
        Thus, since the derivative $\frac{d}{dc}R_\lambda(c)$ is positive for $c<c_{\lambda,R}$ and negative for $c>c_{\lambda,R}$, $c_{\lambda,R}$ is also the unique global maximiser, i.e.,
        \[\sup_{c>0} R_\lambda(c)=R_{\lambda}(c_{\lambda,R}).\]

        To evaluate $R_\lambda(c_{\lambda,R})$, we observe that from \eqref{eq:app:stazionary_eq_1}, we have
        \begin{equation}
            e^{\frac{\lambda-1}{\lambda}nD_\lambda(P_1\|P_0)}\left(1-\varepsilon+\varepsilon e^{\lambda c_{\lambda,R}}\right)^{\frac{\lambda-1}{\lambda}}=\frac{1-\varepsilon+\varepsilon e^{\lambda c_{\lambda,R}}}{1-\varepsilon+\varepsilon e^{c_{\lambda,R}}}.
        \end{equation}
        Substituting this into the numerator of $R_\lambda(c_{\lambda,R})$, we obtain
        \begin{align*}
            e^{(\lambda-1)c_{\lambda,R}}-e^{\frac{\lambda-1}{\lambda}nD_\lambda(P_1\|P_0)}(1-\varepsilon+\varepsilon e^{\lambda c_{\lambda,R}})^{\frac{\lambda-1}{\lambda}}&=e^{(\lambda-1)c_{\lambda,R}}-\frac{1-\varepsilon+\varepsilon e^{\lambda c_{\lambda,R}}}{1-\varepsilon+\varepsilon e^{c_{\lambda,R}}}\\
            &=\frac{(1-\varepsilon)(e^{(\lambda-1)c_{\lambda,R}}-1)}{1-\varepsilon +\varepsilon e^{c_{\lambda,R}}}.
        \end{align*}
        Finally, dividing by the denominator $e^{(\lambda-1)c_{\lambda,R}}-1$, we obtain
        \begin{equation*}
            R_{\lambda}(c_{\lambda,R})=\frac{1-\varepsilon}{1-\varepsilon+\varepsilon e^{c_{\lambda,R}}}.
        \end{equation*}
    
        \item  Case ($nD_\lambda(P_1\|P_0)\geq \log (1/\varepsilon)$): In this case, $e^{\frac{\lambda-1}{\lambda}nD_{\lambda}(P_1\|P_0)}\geq \varepsilon^{-\frac{\lambda-1}{\lambda}}>\psi(c)$ for all $c>0$. Hence, the derivative $\frac{d}{dc}R_{\lambda}(c)$ is strictly positive for all $c>0$. Thus, the supremum is attained as $c\to \infty$, yielding \[\sup_{c>0}R_\lambda(c)=\lim_{c\to\infty} R_{\lambda}(c)=1-\left(\varepsilon e^{nD_{\lambda}(P_1\|P_0)}\right)^{\frac{\lambda-1}{\lambda}}.\]
    \end{itemize}
    
    For the second branch, fix $\lambda\in\Lambda_F$ and proceed similarly. Let
    
    \begin{equation}
        F_\lambda(c)=\frac{e^{-nD_\lambda(P_0\|P_1)}(\varepsilon +(1-\varepsilon)e^{(\lambda-1)c})^{\frac{\lambda}{\lambda-1}}-1}{e^{\lambda c}-1},\quad\text{for }c>0.
    \end{equation}
    Differentiating $F_\lambda(c)$ with respect to $c$, we get
    \begin{equation*}
        \frac{d}{dc}F_\lambda(c)=\frac{\lambda e^{\lambda c}}{(e^{\lambda c}-1)^2}\Bigg[1-e^{-nD_\lambda(P_0\|P_1)}\left(\varepsilon+(1-\varepsilon)e^{(\lambda-1)c}\right)^{\frac{1}{\lambda-1}}\left(\varepsilon+(1-\varepsilon)e^{-c}\right)\Bigg]
    \end{equation*}
    The sign of $\frac{d}{dc}F_\lambda(c)$ is determined by
    \begin{equation}\label{eq:app:intermediate_sign}
        1-e^{-nD_\lambda(P_0\|P_1)}\Psi(c),
    \end{equation}
    where 
    \[
        \Psi(c)=\left(\varepsilon+(1-\varepsilon)e^{(\lambda-1)c}\right)^{\frac{1}{\lambda-1}}\left(\varepsilon+(1-\varepsilon)e^{-c}\right).
    \]
    
    To study the sign of \eqref{eq:app:intermediate_sign}, let us analyse $\Psi(c)$. Differentiating the logarithm of $\Psi(c)$ gives
    \begin{equation*}
        \frac{d}{dc}\log \Psi(c)=\frac{\varepsilon(1-\varepsilon)(e^{(\lambda-1)c}-e^{-c})}{(\varepsilon+(1-\varepsilon)e^{(\lambda-1)c})(\varepsilon+(1-\varepsilon)e^{-c})}>0,
    \end{equation*}
    since for $\lambda>1$ and $c>0$, we have $e^{(\lambda-1)c}>e^{-c}$. Consequently, $\Psi(c)$ is strictly increasing on $(0,\infty)$, and
    \begin{equation*}
        \lim_{c\to 0^+} \Psi(c)=1\quad\text{and}\quad\lim_{c\to \infty}\Psi(c)=+\infty.
    \end{equation*}
    Thus, we have to distinguish two cases according to the value of $nD_{\lambda}(P_0\|P_1)$:
    \begin{itemize}
        \item \textbf{Case} $nD_{\lambda}(P_0\|P_1)=0$: Since $e^{-nD_{\lambda}(P_0\|P_1)}=1$ and $\Psi(c)>1$ for all $c>0$, the derivative $\frac{d}{dc}F_\lambda(c)$ is always negative. Consequently, the supremum is attained as $c\to 0^+$, yielding
        \begin{equation*}
            \sup_{c>0}F_\lambda(c)=\lim_{c\to 0^+}F_\lambda(c)=1-\varepsilon.
        \end{equation*}
    
        \item  \textbf{Case} $nD_{\lambda}(P_0\|P_1)>0$: Since $e^{nD_{\lambda}(P_0\|P_1)}>1$ and $\Psi(c)$ is strictly increasing from $1$ to $\infty$, there exists a unique positive value $c_{\lambda, F}\in(0,\infty)$ for which $1-e^{-nD_\lambda(P_0\|P_1)}\Psi(c_{\lambda, F})=0$, or equivalently
        \begin{equation}\label{eq:app:stationary_eq_2}
             \Psi(c_{\lambda, F})=e^{nD_{\lambda}(P_0\|P_1)}.
        \end{equation}
        Consequently, the derivative $\frac{d}{dc} F_{\lambda}(c)$ is positive for $c<c_{\lambda, F}$ and negative for $c>c_{\lambda, F}$. Thereby, $c_{\lambda, F}$ is the unique global maximiser, i.e.,
        \[
            \sup_{c>0}F_{\lambda}(c)=F_{\lambda}(c_{\lambda, F}).
        \]
        Finally, to explicitly evaluate $F_\lambda(c_{\lambda, F})$, we use \eqref{eq:app:stationary_eq_2} to substitute the divergence term, obtaining
        \begin{equation*}
            e^{-nD_\lambda(P_0\|P_1)}\left(\varepsilon+(1-\varepsilon)e^{(\lambda-1)c_{\lambda,F}}\right)^{\frac{\lambda}{\lambda-1}} = \frac{\varepsilon+(1-\varepsilon)e^{(\lambda-1)c_{\lambda,F}}}{\varepsilon+(1-\varepsilon)e^{-c_{\lambda,F}}}.
        \end{equation*}
        Substituting this into the numerator of $F_\lambda(c_{\lambda, F})$ and dividing by $e^{\lambda c_{\lambda,F}}-1$ yields
        \begin{equation*}
            F_\lambda(c_{\lambda, F}) = \frac{(1-\varepsilon)e^{-c_{\lambda, F}}}{\varepsilon+(1-\varepsilon)e^{-c_{\lambda, F}}} = \frac{1-\varepsilon}{1-\varepsilon+\varepsilon e^{c_{\lambda, F}}}.
        \end{equation*}
    \end{itemize}

\subsection{Proof of Remark \ref{rem:DPI_recover}}\label{app:rem:DPI_recover}
    We demonstrate that the two converse bounds
    
    \begin{equation}\label{eq:converse_renyi_P_1_P_0_1}
            \beta_n(\varepsilon)\geq\sup_{\lambda>1,c>0} \Bigg\{\frac{e^{(\lambda-1)c}-e^{\frac{\lambda-1}{\lambda}nD_\lambda(P_1\|P_0)}\left(1-\varepsilon+\varepsilon e^{\lambda c}\right)^{\frac{\lambda-1}{\lambda}}}{e^{(\lambda-1)c}-1}\Bigg\}.
        \end{equation}
    and
    \begin{equation}\label{eq:converse_renyi_P_0_P_1_1}
            \beta_n(\varepsilon)\geq\sup_{\lambda>1,c>0}\Bigg\{\frac{e^{-nD_\lambda(P_0\|P_1)}\left(\varepsilon +(1-\varepsilon)e^{(\lambda-1)c}\right)^{\frac{\lambda}{\lambda-1}}-1}{e^{\lambda c}-1}\Bigg\}.
        \end{equation}
    recover their corresponding bound obtained via DPI.

    Let us start from \eqref{eq:converse_renyi_P_1_P_0_1}.
    Fix $\lambda>1$. Dividing both the numerator and denominator of the right-hand side of \eqref{eq:converse_renyi_P_1_P_0_1} by $e^{(\lambda-1)c}$, and  taking the limit as  $c\to \infty$, we get
    \begin{align}
        \beta_n(\varepsilon)&
        \geq\lim_{c\to\infty}\frac{e^{(\lambda-1)c} -\left((1-\varepsilon(1-e^{\lambda c})) e^{nD_\lambda(P_1\|P_0)}\right)^{\frac{\lambda-1}{\lambda}}}{e^{(\lambda-1)c}-1}\nonumber\\
        &=\lim_{c\to\infty}\frac{1 -e^{\frac{\lambda-1}{\lambda}nD_\lambda(P_1\|P_0)}\left(\varepsilon+(1-\varepsilon)e^{-\lambda c}\right)^{\frac{\lambda-1}{\lambda}}}{1-e^{-(\lambda-1)c}}\label{eq:red_step_P_1_P_0}
    \end{align}
   As $c\to \infty$ the terms $e^{-(\lambda-1)c}$ and $(1-\varepsilon)e^{-\lambda c}$ vanish. Thus, it follows from \eqref{eq:red_step_P_1_P_0} that
    \begin{equation*}
        \beta_n(\varepsilon)\geq\lim_{c\to\infty}\frac{1 -e^{\frac{\lambda-1}{\lambda}nD_\lambda(P_1\|P_0)}\left(\varepsilon+(1-\varepsilon)e^{-\lambda c}\right)^{\frac{\lambda-1}{\lambda}}}{1-e^{-(\lambda-1)c}}=1-\left(\varepsilon\, e^{nD_\lambda(P_1\| P_0)}\right)^{\frac{\lambda-1}{\lambda}}.
    \end{equation*}
    Since the above inequality holds for any $\lambda>1$, we can take the supremum over $\lambda$ (or equivalently, the infimum for the subtracted term), obtaining
    \begin{equation}\label{app:eq:DPI_renyi_P_1_P_0}
        \beta_n(\varepsilon)\geq1-\inf_{\lambda>1}\left(\varepsilon\, e^{nD_\lambda(P_1\| P_0)}\right)^{\frac{\lambda-1}{\lambda}},
    \end{equation}
    which perfectly recovers the corresponding bound achieved via DPI.

    Similarly, let us consider \eqref{eq:converse_renyi_P_0_P_1_1}. Fix $\lambda>1$. Dividing both the numerator and denominator on the right-hand side of \eqref{eq:converse_renyi_P_0_P_1_1} by $e^{\lambda c}$ and taking the limit as $c\to\infty$, we obtain
    \begin{align}
        \beta_n(\varepsilon)&\geq\lim_{c\to\infty} \frac{(e^{(\lambda-1)c}-\varepsilon(e^{(\lambda-1)c}-1))^{\frac{\lambda}{\lambda-1}}e^{-nD_\lambda(P_0\|P_1)}-1}{e^{\lambda c}-1}\nonumber\\
        &=
        \lim_{c\to\infty}
        \frac{
        e^{-nD_\lambda(P_0\|P_1)}
        \left(1-\varepsilon+\varepsilon e^{-(\lambda-1)c}\right)^{\frac{\lambda}{\lambda-1}}
        -e^{-\lambda c}
        }
        {1-e^{-\lambda c}}.
                \label{eq:red_step_P_0_P_1}
            \end{align}
            As $c\to \infty$ the terms $e^{-\lambda c}$ and $\varepsilon e^{-(\lambda-1) c}$ vanish. Thus, from \eqref{eq:red_step_P_0_P_1} we obtain the following:
            \begin{equation*}
                \beta_n(\varepsilon)\geq\lim_{c\to\infty}
        \frac{
        e^{-nD_\lambda(P_0\|P_1)}
        \left(1-\varepsilon+\varepsilon e^{-(\lambda-1)c}\right)^{\frac{\lambda}{\lambda-1}}
        -e^{-\lambda c}
        }
        {1-e^{-\lambda c}}=(1-\varepsilon)^{\frac{\lambda}{\lambda-1}}\,e^{-nD_\lambda(P_0\| P_1)}.
    \end{equation*}
    Optimising over $\lambda>1$, we get
    \begin{equation}
    \beta_n(\varepsilon)\geq\sup_{\lambda>1}\left((1-\varepsilon)^{\frac{\lambda}{\lambda-1}}\,e^{-nD_\lambda(P_0\| P_1)}\right),    
    \end{equation}
    which perfectly recovers the corresponding bound achieved via DPI, thus concluding the proof.

    \subsection{Proof of Theorem~\ref{theorem:beta_upper_bound_renyi}}\label{app:theorem:beta_upper_bound_renyi}
    Let $L_n=dP_1^n/dP_0^n$. Since $\phi$ is a possibly randomised LLRT with threshold $\tau$, it satisfies
    \begin{equation}
        \phi(\mathbf{x})=1\quad\text{if }\log L_n(\mathbf{x})>\tau, \qquad \phi(\mathbf{x})=0\quad\text{if }\log L_n(\mathbf{x})<\tau,
    \end{equation}
    and may take values in $[0,1]$ when $\log L_n(\mathbf{x})=\tau$. For every $\lambda\in(0,1)$,
    \begin{align}
        e^{(\lambda-1)nD_\lambda(P_1\|P_0)}
        &=\E_{P_0^n}[L_n^\lambda]\nonumber\\
        &=\E_{P_0^n}[L_n^\lambda\phi]+\E_{P_0^n}[L_n^\lambda(1-\phi)]\nonumber\\
        &\geq e^{\lambda\tau}\E_{P_0^n}[\phi]+e^{(\lambda-1)\tau}\E_{P_0^n}[L_n(1-\phi)]\nonumber\\
        &=e^{\lambda\tau}\alpha_n(\phi)+e^{(\lambda-1)\tau}\beta_n(\phi).
    \end{align}
    Indeed, $L_n\geq e^\tau$ whenever $\phi>0$, while $L_n\leq e^\tau$ whenever $1-\phi>0$. Isolating $\beta_n(\phi)$ and taking the infimum over $\lambda\in(0,1)$ proves the result.


\subsection{Proof of Corollary~\ref{cor:beta_optimal_upper_bound_renyi}}\label{app:cor:beta_optimal_upper_bound_renyi}
    Let $\phi_\varepsilon$ be a possibly randomised Neyman--Pearson test attaining $\beta_n(\varepsilon)$, with Type~I error equal to $\varepsilon$, and let $\tau\in\mathbb{R}$ denote its likelihood-ratio threshold. Fix $\lambda\in(0,1)$ and define
    \begin{equation*}
        M_\lambda=\exp\big((\lambda-1)nD_\lambda(P_1\|P_0)\big).
    \end{equation*}
    Applying Theorem~\ref{theorem:beta_upper_bound_renyi} to $\phi_\varepsilon$ and setting $t=e^\tau$ gives
    \begin{align}
        \beta_n(\varepsilon)
        &\leq\frac{M_\lambda-\varepsilon e^{\lambda\tau}}{e^{(\lambda-1)\tau}}\nonumber\\
        &=M_\lambda t^{1-\lambda}-\varepsilon t\nonumber\\
        &\leq\sup_{s>0}\left\{M_\lambda s^{1-\lambda}-\varepsilon s\right\}.\label{eq:legendre_type_supremum}
    \end{align}
    The function inside the supremum is strictly concave and attains its maximum at
    \begin{equation*}
        s_\lambda^\star=\left(\frac{(1-\lambda)M_\lambda}{\varepsilon}\right)^{1/\lambda}.
    \end{equation*}
    Substituting $s_\lambda^\star$ into~\eqref{eq:legendre_type_supremum} yields
    \begin{equation*}
        \beta_n(\varepsilon)\leq\lambda\left(\frac{1-\lambda}{\varepsilon}\right)^{\frac{1-\lambda}{\lambda}}M_\lambda^{1/\lambda} =\lambda\left(\frac{1-\lambda}{\varepsilon}\right)^{\frac{1-\lambda}{\lambda}}\exp\left(-\frac{1-\lambda}{\lambda}nD_\lambda(P_1\|P_0)\right).
    \end{equation*}
    Since this holds for every $\lambda\in(0,1)$, taking the infimum over $\lambda$ proves the result.

\section{Unified $f$-Divergence Framework}\label{app:f_div_framework}

In this section, we develop a common variational framework for deriving converse results for $f$-divergences. Indeed, restricting the representation of an $f$-divergence to two-valued functions produces the binary reduction needed for hypothesis testing.

Specifically, restrict the variational representation to functions of the form
\[
g(x)=a\mathbbm{1}_A(x)+b\mathbbm{1}_{A^c}(x),
\]
for a measurable set $A$. This recovers the binary data-processing inequality induced by the map $x\mapsto\mathbbm{1}_A(x)$.

\begin{lemma}
\label{lemma:variational_fdiv}
Let $P$ and $Q$ be probability measures on $(\mathcal{X},\mathcal{B})$, and let $P$ be  absolutely continuous with respect to $Q$. Let
$f:[0,\infty)\to(-\infty,+\infty]$ be a
convex and lower semicontinuous function satisfying $f(1)=0$, and let $f^\star$ denote the convex conjugate
of $f$,
\[
    f^\star(t)
    =
    \sup_{u\geq 0}\{tu-f(u)\}.
\]
Then, for every measurable set $A\in\mathcal{B}$, it holds
\begin{align}
D_f(P\|Q)&\geq\sup_{a,b\in\operatorname{dom}(f^\star)}\Bigg\{b+(a-b)P(A) -f^\star(b)-(f^\star(a)-f^\star(b))Q(A)\Bigg\}
\label{eq:binary_variational_bound}\\
&=
D_f\bigg(\mathrm{Bern}\big(P(A)\big)\big\|\mathrm{Bern}\big(Q(A)\big)\bigg),
\label{eq:binary_fdiv_DPI}
\end{align}
where $\mathrm{Bern}(x)$ denotes the Bernoulli distribution with parameter $x\in[0,1]$ and the endpoint cases are interpreted using the conventions in \eqref{eq:f_divergence_def}.
\end{lemma}
\begin{proof}
    By the variational representation of the $f$-divergence in \eqref{eq:f_divergence_variational}, we have
\begin{equation}\label{eq:app:vr_f_div_def} 
D_f(P\|Q)
=
\sup_{g}
\left\{
    \mathbb{E}_P[g]
    -
    \mathbb{E}_Q[f^\star(g)]
\right\},
\end{equation}

where the supremum is over all admissible measurable functions $g$.

For a measurable set $A$ and $a,b\in\operatorname{dom}(f^\star)$, let
\begin{equation}\label{eq:app:g_def}
    g(x)=a\mathbbm{1}_A(x)+b\mathbbm{1}_{A^c}(x),
\end{equation}
where $\mathbbm{1}_{A}$ denotes the indicator function of the set $A$.

Substituting \eqref{eq:app:g_def} into \eqref{eq:app:vr_f_div_def} yields
\begin{align}\label{eq:app:vr_f_div_sub}
D_f(P\|Q)
&\geq
\mathbb{E}_P[g]
-
\mathbb{E}_Q[f^\star(g)]
\nonumber\\
&=
aP(A)+bP(A^c)
-f^\star(a)Q(A)-f^\star(b)Q(A^c).
\end{align}
Using $P(A^c)=1-P(A)$ and $Q(A^c)=1-Q(A)$, we obtain
\begin{equation}\label{eq:app:vr_f_div_sub_1}
    D_f(P\|Q)\geq b+(a-b)P(A) -f^\star(b)-(f^\star(a)-f^\star(b))Q(A).
\end{equation}
Since \eqref{eq:app:vr_f_div_sub_1} holds for every
$a,b\in\operatorname{dom}(f^\star)$, taking the supremum over these
parameters yields 
\begin{align}\label{eq:app:vr_f_div_gen_low}
D_f(P\|Q)
\geq
\sup_{a,b\in\operatorname{dom}(f^\star)}
\Big\{
b+(a-b)P(A)
-f^\star(b)
-\bigl(f^\star(a)-f^\star(b)\bigr)Q(A)
\Big\},
\end{align}
which proves \eqref{eq:binary_variational_bound}.

We next show that optimising \eqref{eq:app:vr_f_div_gen_low} recovers the
$f$-divergence between the corresponding binary distributions.  

Let
$p=P(A)$ and $q=Q(A)$. 
Assume first that $q\in(0,1)$.
optimising the right-hand side of \eqref{eq:app:vr_f_div_sub_1}  over $a$ and $b$ gives
\begin{align}\label{eq:app:right_hand_side}
\sup_{a,b}\left\{b+(a-b)p-f^\star(b)-(f^\star(a)-f^\star(b))q\right\}&=
\sup_{a,b}\left\{
ap+b(1-p)-qf^\star(a)-(1-q)f^\star(b)
\right\}
\nonumber\\
&=
\sup_a\left\{ap-qf^\star(a)\right\}
+
\sup_b\left\{b(1-p)-(1-q)f^\star(b)\right\}.
\end{align}
Since $f$ is proper, convex, and lower semicontinuous, the Fenchel--Moreau theorem gives $f=f^{**}$, i.e.,
\[
    f(u)=\sup_t\{tu-f^\star(t)\},
\]
and, since $q>0$, it follows that
\begin{align*}
    \sup_a \left\{ap-qf^\star(a)\right\}&=\sup_a \left\{q\left(a\frac{p}{q}-f^\star(a)\right)\right\}\\
    &=q \sup_a \left\{\left(a\frac{p}{q}-f^\star(a)\right)\right\}\\
    &=qf\left(\frac{p}{q}\right).
\end{align*}
    
Analogously, we also obtain
\begin{equation*}
    \sup_b\left\{b(1-p)-(1-q)f^\star(b)\right\}=(1-q)f\left(\frac{1-p}{1-q}\right).
\end{equation*}
Substituting these into \eqref{eq:app:right_hand_side} gives
\begin{equation}
    \sup_{a,b}\left\{b+(a-b)p-f^\star(b)-(f^\star(a)-f^\star(b))q\right\} = qf\left(\frac{p}{q}\right)+(1-q)f\left(\frac{1-p}{1-q}\right).
\end{equation}
Hence, recalling $p=P(A)$ and $q=Q(A)$, we finally get 
\begin{align}
D_f(P\|Q)
&\geq\sup_{a,b\in\operatorname{dom}(f^\star)}\Bigg\{b+(a-b)P(A) -f^\star(b)-(f^\star(a)-f^\star(b))Q(A)\Bigg\}\nonumber\\
&=
Q(A)f\left(\frac{P(A)}{Q(A)}\right)
+
(1-Q(A))f\left(\frac{1-P(A)}{1-Q(A)}\right)
\nonumber\\
&=
D_f(\mathrm{Bern}(P(A))\|\mathrm{Bern}(Q(A)))\nonumber,
\end{align}
which proves the claim for $q\in(0,1)$. If $q\in\{0,1\}$, absolute continuity forces the corresponding value of $p$ at the zero-mass atom, and the result follows by the conventions in \eqref{eq:f_divergence_def}, or equivalently by lower-semicontinuous extension from $q\in(0,1)$.
\end{proof}

Applying Lemma~\ref{lemma:variational_fdiv} to the test output yields the following general converse.
\begin{theorem}\label{thm:fdiv_beta_lowerbound}
Let $P_1^n$ and $P_0^n$ be mutually absolutely continuous, and let $f$ be the generator of an $f$-divergence. For any possibly randomised test $\phi:\mathcal{X}^n\rightarrow[0,1]$ with error probabilities $\alpha_n$ and $\beta_n$,
\begin{align}
    \beta_n&\geq \max\Bigg\{\sup_{\substack{a,b\in\operatorname{dom}(f^\star)\\a>b}}\frac{
a-\alpha_n f^\star(a)
-(1-\alpha_n)f^\star(b)
-D_f(P_1^n\|P_0^n)
}{a-b},\nonumber\\
&\qquad \sup_{\substack{a,b\in\operatorname{dom}(f^\star)\\
f^\star(a)>f^\star(b)}}\frac{a(1-\alpha_n)+b\alpha_n-f^\star(b)-
D_f(P_0^n\|P_1^n)
}{f^\star(a)-f^\star(b)}\Bigg\}.\label{eq:f_div_gen_beta_lowerbound}
\end{align}
Moreover, for any $\varepsilon\in(0,1)$,  we have
\begin{align}
    \beta_n(\varepsilon)&\geq \max\Bigg\{\sup_{\substack{a,b\in\operatorname{dom}(f^\star)\\a>b,f^\star(a)\geq f^\star(b)}}\frac{
a-\varepsilon f^\star(a)
-(1-\varepsilon)f^\star(b)
-D_f(P_1^n\|P_0^n)
}{a-b},\nonumber\\
&\qquad \sup_{\substack{a,b\in\operatorname{dom}(f^\star)\\a\geq b,f^\star(a)>f^\star(b)}}\frac{a(1-\varepsilon)+b\varepsilon-f^\star(b)-
D_f(P_0^n\|P_1^n)
}{f^\star(a)-f^\star(b)}\Bigg\}.\label{eq:f_div_gen_beta_eps_lowerbound}
\end{align}
\end{theorem}
\begin{proof}
    Let $U\sim\mathrm{Unif}[0,1]$ be independent of the observations under both hypotheses, and define
\begin{equation*}
    A=\{(\mathbf{x},u)\in\mathcal{X}^n\times[0,1]:u\leq\phi(\mathbf{x})\}.
\end{equation*}
Under $P_j^n\otimes\mathrm{Unif}[0,1]$, this realises the randomised test as a deterministic decision region. Hence $P_0^n(A)=\alpha_n$ and $P_1^n(A)=1-\beta_n$, where the product with the common auxiliary law is suppressed from the notation. Adding the same independent auxiliary variable does not change any divergence appearing below.

Now, from Lemma \ref{lemma:variational_fdiv}, setting $P=P_1^n$, $Q=P_0^n$ and using the set $A$, we obtain
\begin{align}
    D_f(P_1^n\|P_0^n)&\geq
b+(a-b)P_1^n(A)-f^\star(b)
-\bigl(f^\star(a)-f^\star(b)\bigr)P_0^n(A)
\nonumber\\
&=
b+(a-b)(1-\beta_n)-f^\star(b)
-\bigl(f^\star(a)-f^\star(b)\bigr)\alpha_n\nonumber\\
&=
a-\alpha_n f^\star(a)
-(1-\alpha_n)f^\star(b)
-(a-b)\beta_n,\label{eq:app:f_div_P_1_P_0_lb}
\end{align}
for any $a,b\in \operatorname{dom}(f^\star)$.

From \eqref{eq:app:f_div_P_1_P_0_lb}, whenever $a>b$, we obtain the following bound on $\beta_n$:
\begin{equation}
    \beta_n
\geq
\frac{
a-\alpha_n f^\star(a)
-(1-\alpha_n)f^\star(b)
-D_f(P_1^n\|P_0^n)
}{a-b}.
\end{equation}

Since the above inequality holds for every
$a,b\in\operatorname{dom}(f^\star)$ such that $a>b$, taking the supremum yields
\begin{equation}
\beta_n
\geq
\sup_{\substack{a,b\in\operatorname{dom}(f^\star)\\a>b}}
\frac{
a-\alpha_n f^\star(a)
-(1-\alpha_n)f^\star(b)
-D_f(P_1^n\|P_0^n)
}{a-b}.
\label{eq:first-corollary-bound-proof}
\end{equation}

Similarly, for the second bound, we apply Lemma~\ref{lemma:variational_fdiv} with $P=P_0^n$ and $Q=P_1^n$ and using the complementary set $A^c$. Thus, since $P_0^n(A^c)=1-\alpha_n$ and $P_1^n(A^c)=\beta_n$, we get
\begin{align*}
D_f(P_0^n\|P_1^n)
&\geq
b+(a-b)(1-\alpha_n)-f^\star(b)
-\bigl(f^\star(a)-f^\star(b)\bigr)\beta_n\\
&=
a(1-\alpha_n)+b\alpha_n-f^\star(b)
-\bigl(f^\star(a)-f^\star(b)\bigr)\beta_n,
\end{align*}
for any $a,b\in \operatorname{dom}(f^\star)$.

From which, whenever $f^\star(a)>f^\star(b)$, we obtain the following bound on $\beta_n$:
\begin{equation}
    \beta_n
\geq
\frac{
a(1-\alpha_n)+b\alpha_n-f^\star(b)
-D_f(P_0^n\|P_1^n)
}{
f^\star(a)-f^\star(b)
}.
\end{equation}
Taking the supremum over all admissible $a,b$ such that $f^\star(a)>f^\star(b)$ gives
\begin{equation}
\beta_n
\geq
\sup_{\substack{a,b\in\operatorname{dom}(f^\star)\\
f^\star(a)>f^\star(b)}}
\frac{
a(1-\alpha_n)+b\alpha_n-f^\star(b)
-D_f(P_0^n\|P_1^n)
}{
f^\star(a)-f^\star(b)
}.
\label{eq:second-corollary-bound-proof}
\end{equation}
Combining \eqref{eq:first-corollary-bound-proof} and
\eqref{eq:second-corollary-bound-proof} yields \eqref{eq:f_div_gen_beta_lowerbound}, proving the first claim.

We now derive a bound also for the optimal Type~II error probability $\beta_n(\varepsilon)$ for any $\varepsilon\in(0,1)$. 

To this end, consider the bound in \eqref{eq:first-corollary-bound-proof}. For any $a>b$ such that
$f^\star(a)\geq f^\star(b)$,  the right-hand side of \eqref{eq:first-corollary-bound-proof}  is non-increasing in $\alpha_n$. Indeed, the numerator in \eqref{eq:first-corollary-bound-proof} can be written as
\[
a-f^\star(b)
-\alpha_n\bigl(f^\star(a)-f^\star(b)\bigr)
-D_f(P_1^n\|P_0^n).
\]
Since $\alpha_n\leq\varepsilon$ and
$f^\star(a)-f^\star(b)\geq0$, we have
\[
-\alpha_n\bigl(f^\star(a)-f^\star(b)\bigr)
\geq
-\varepsilon\bigl(f^\star(a)-f^\star(b)\bigr).
\]
Consequently, for every test satisfying $\alpha_n\leq \varepsilon$ we have
\begin{align*}
    \beta_n
&\geq \frac{
a-\alpha_n f^\star(a)
-(1-\alpha_n)f^\star(b)
-D_f(P_1^n\|P_0^n)
}{a-b}.\\
&\geq\frac{
a-\varepsilon f^\star(a)
-(1-\varepsilon)f^\star(b)
-D_f(P_1^n\|P_0^n)
}{a-b}.
\end{align*}
Since the above inequality holds for every test satisfying $\alpha_n\leq\varepsilon$, it also holds for the optimal one. Thus, taking the supremum over all admissible
$a,b$ yields
\begin{equation}
\beta_n(\varepsilon)
\geq
\sup_{\substack{a,b\in\operatorname{dom}(f^\star)\\
a>b,\ f^\star(a)\geq f^\star(b)}}
\frac{
a-\varepsilon f^\star(a)
-(1-\varepsilon)f^\star(b)
-D_f(P_1^n\|P_0^n)
}{a-b}.
\label{eq:first-eps-bound-proof}
\end{equation}

For the second bound, we proceed analogously from
\eqref{eq:second-corollary-bound-proof}. Fix
$a\geq b$ such that
$f^\star(a)>f^\star(b)$. The numerator in \eqref{eq:second-corollary-bound-proof} can be written as
\[
a-f^\star(b)
-\alpha_n(a-b)
-D_f(P_0^n\|P_1^n)
\]
Since $\alpha_n\leq\varepsilon$ and $a-b\geq0$,
\[
-\alpha_n(a-b)
\geq
-\varepsilon(a-b).
\]
Therefore, for any test satisfying $\alpha_n\leq \varepsilon$ we have
\[
\beta_n
\geq
\frac{
a(1-\varepsilon)+b\varepsilon-f^\star(b)
-D_f(P_0^n\|P_1^n)
}{
f^\star(a)-f^\star(b)
}.
\]
Again, since the right-hand side in the above inequality is independent of the particular test,
the same bound holds for $\beta_n(\varepsilon)$. Taking the supremum over
all admissible
$a,b\in\operatorname{dom}(f^\star)$ satisfying
$a\geq b$ and $f^\star(a)>f^\star(b)$ gives
\begin{equation}
\beta_n(\varepsilon)
\geq
\sup_{\substack{a,b\in\operatorname{dom}(f^\star)\\
a\geq b,\ f^\star(a)>f^\star(b)}}
\frac{
a(1-\varepsilon)+b\varepsilon-f^\star(b)
-D_f(P_0^n\|P_1^n)
}{
f^\star(a)-f^\star(b)
}.
\label{eq:second-eps-bound-proof}
\end{equation}

Combining \eqref{eq:first-eps-bound-proof} and
\eqref{eq:second-eps-bound-proof} proves \eqref{eq:f_div_gen_beta_eps_lowerbound}, which concludes the proof.
\end{proof}

specialising Theorem~\ref{thm:fdiv_beta_lowerbound} to the $E_\gamma$ divergence gives the following result.
\begin{corollary}\label{cor:E_gamma_bound}
Let $P_1^n$ and $P_0^n$ be mutually absolutely continuous. Then, for any
$\varepsilon\in(0,1)$,
\begin{align}
\beta_n(\varepsilon)
\geq
\max\Bigg\{
&\sup_{\gamma>0}
\frac{
1-\varepsilon-E_\gamma(P_0^n\|P_1^n)-(1-\gamma)_+
}{\gamma},\sup_{\gamma>0}\left(
1-\gamma\varepsilon
-E_\gamma(P_1^n\|P_0^n)-(1-\gamma)_+
\right)
\Bigg\},
\label{eq:E_gamma_bound}
\end{align}
where $E_\gamma(P\| Q)=D_{f_\gamma}(P\|Q)$ denotes the $E_\gamma$ divergence, which is an $f$-divergence with function
\[f_\gamma(t)=(t-\gamma)_+-(1-\gamma)_+\quad\text{for }\gamma>0,\]
where $(x)_+= \max\{x,0\}$.
\end{corollary}
\begin{proof}
    Fix $\varepsilon\in(0,1)$, and let $\phi$ be a possibly randomised test achieving $\beta_n(\varepsilon)$, with Type~I error $\alpha_n\leq\varepsilon$. Let $U\sim\mathrm{Unif}[0,1]$ be independent under both hypotheses and define
\begin{equation*}
A=\{(\mathbf{x},u):u\leq\phi(\mathbf{x})\}.
\end{equation*}
Under the augmented laws $P_j^n\otimes\mathrm{Unif}[0,1]$, this set satisfies $P_0^n(A)=\alpha_n$ and $P_1^n(A)=1-\beta_n(\varepsilon)$. We suppress the common auxiliary factor below; adjoining it leaves every $f$-divergence unchanged.
For $\gamma>0$, the $E_\gamma$ divergence corresponds to the $f$-divergence associated with the generator $f_\gamma(x)=(x-\gamma)_+ -(1-\gamma)_+$. Its convex conjugate is
\begin{equation}
    f_\gamma^*(t)=\begin{cases}
\gamma t + (1-\gamma)_+, & 0\leq t\leq 1,\\
+\infty, & \text{otherwise},
\end{cases}
\end{equation}
so that $\operatorname{dom}(f_\gamma^*)=[0,1]$.

We can now apply Lemma~\ref{lemma:variational_fdiv}. Specifically, applying Lemma~\ref{lemma:variational_fdiv} with
$P=P_0^n$, $Q=P_1^n$, the set $A$ and $f=f_\gamma$ gives

\begin{align}\label{eq:app:supremum_thm_e_gamma}
    E_\gamma(P_0^n\|P_1^n)
&\geq
\sup_{a,b\in[0,1]}
\Big\{a\alpha_n+b(1-\alpha_n)-f_\gamma^*(a)+(f^\star_\gamma(a)-f^\star_\gamma(b))\beta_n(\varepsilon)
\Big\}\nonumber\\
&=\sup_{a,b\in[0,1]}
\Big\{a\alpha_n+b(1-\alpha_n)-\gamma a -(1-\gamma)_++\gamma(a-b)\beta_n(\varepsilon)
\Big\}
\end{align}
We now choose $a=0$ and $b=1$ in the supremum \eqref{eq:app:supremum_thm_e_gamma}. This yields
\begin{align}\label{eq:app:e_gamma_first_direction}
    E_\gamma(P_0^n\|P_1^n)
&\geq
1-\alpha_n-(1-\gamma)_+-\gamma\beta_n(\varepsilon)\nonumber\\
&\geq 1-\varepsilon-(1-\gamma)_+-\gamma\beta_n(\varepsilon),
\end{align}
where the last inequality follows from $\alpha_n\leq \varepsilon$.

Rewriting \eqref{eq:app:e_gamma_first_direction} to isolate $\beta_n(\varepsilon)$ gives
\begin{equation}
    \beta_n(\varepsilon)
\geq
\frac{
1-\varepsilon-E_\gamma(P_0^n\|P_1^n)-(1-\gamma)_+
}{\gamma}.
\end{equation}

Because the above bound holds for any $\gamma>0$, taking the supremum over $\gamma>0$ yields
\begin{equation}
    \beta_n(\varepsilon)
\geq
\sup_{\gamma>0}
\frac{
1-\varepsilon-E_\gamma(P_0^n\|P_1^n)-(1-\gamma)_+
}{\gamma}.
\end{equation}

For the reverse direction, applying Lemma~\ref{lemma:variational_fdiv} with
$P=P_1^n$, $Q=P_0^n$, the set $A$ and $f=f_\gamma$ gives

\begin{align}
    E_\gamma(P_1^n\|P_0^n)
&\geq
\sup_{a,b\in[0,1]}\Bigg\{
a-(a-b)\beta_n(\varepsilon)-f_\gamma^*(b)-(f^\star_\gamma(a)-f^\star_\gamma(b))\alpha_n\Bigg\}\nonumber\\
&=\sup_{a,b\in[0,1]}\Bigg\{
a-(a-b)\beta_n(\varepsilon)-\gamma b -(1-\gamma)_+-\gamma(a-b)\alpha_n\Bigg\}
\end{align}
Choosing $a=1$ and $b=0$ in the supremum, we obtain
\begin{align}
    E_\gamma(P_1^n\|P_0^n)
&\geq
1-\beta_n(\varepsilon)-(1-\gamma)_+-\gamma\alpha_n\\
&\geq 1-\beta_n(\varepsilon)-(1-\gamma)_+-\gamma\varepsilon,
\end{align}
where the last inequality follows from $\alpha_n\leq \varepsilon$.

From the above inequality by isolating $\beta_n(\varepsilon)$, we obtain
\begin{equation}
    \beta_n(\varepsilon)
\geq
1-\gamma\varepsilon-E_\gamma(P_1^n\|P_0^n)-(1-\gamma)_+.
\end{equation}

Again, because the above bound holds for any $\gamma>0$, taking the supremum over $\gamma>0$ gives
\begin{equation}
    \beta_n(\varepsilon)
\geq
\sup_{\gamma>0}\{
1-\gamma\varepsilon-E_\gamma(P_1^n\|P_0^n)-(1-\gamma)_+\}
.
\end{equation}

Combining the two bounds concludes the proof.
\end{proof}
Setting $\gamma=1$ in \eqref{eq:E_gamma_bound} recovers the classical total-variation converse.

The same approach yields a Hellinger-$\lambda$ and a KL divergence converse bound, as shown in the corollaries below. 

\begin{corollary}\label{cor:Hellinger_div_bound}
    Let $P_1^n$ and $P_0^n$ be mutually absolutely continuous. Then, for any
$\varepsilon\in(0,1)$ and any $\lambda\in(0,1)$,
\begin{equation}
h_\lambda(P_0^n,P_1^n)
\leq
\varepsilon^\lambda
\bigl(1-\beta_n(\varepsilon)\bigr)^{1-\lambda}
+
(1-\varepsilon)^\lambda
\beta_n(\varepsilon)^{1-\lambda},
\label{eq:Hellinger_div_bound}
\end{equation}
where
\begin{equation}
h_\lambda(P,Q)=
1-H_\lambda(P,Q)
\end{equation}
denotes the Hellinger affinity of order $\lambda$, and 
$H_\lambda$ denotes the Hellinger divergence of order $\lambda$, which 
is an $f$-divergence with function
\begin{equation*}
f_\lambda(t)=1-t^\lambda.
\end{equation*}
\end{corollary}
\begin{proof}
    Fix $\varepsilon\in(0,1)$. The constant test $\phi\equiv\varepsilon$ gives $\beta_n(\varepsilon)\leq1-\varepsilon$. Let $\phi$ now be a possibly randomised test achieving $\beta_n(\varepsilon)$, with Type~I error $\alpha_n\leq\varepsilon$. Introduce an independent $U\sim\mathrm{Unif}[0,1]$ under both hypotheses and set
\begin{equation*}
A=\{(\mathbf{x},u):u\leq\phi(\mathbf{x})\}.
\end{equation*}
Under the augmented laws, $P_0^n(A)=\alpha_n$ and $P_1^n(A)=1-\beta_n(\varepsilon)$. We again suppress the common auxiliary factor, which leaves the Hellinger divergence unchanged.
For $\lambda\in(0,1)$, the Hellinger divergence of order $\lambda$ is the $f$-divergence generated by
\begin{equation}
f_\lambda(t)=1-t^\lambda,
\qquad \text{for } 0<\lambda<1.
\end{equation}
Therefore, its convex conjugate is
\begin{equation}
f_\lambda^*(t)
=
\begin{cases}
-1
+(1-\lambda)
\left(\dfrac{\lambda}{-t}\right)^{\frac{\lambda}{1-\lambda}},
& t<0,\\
+\infty,
& t\geq0.
\end{cases}
\end{equation}
Hence, $\operatorname{dom}(f_\lambda^*)=(-\infty,0).$
Applying Lemma~\ref{lemma:variational_fdiv} with
$P=P_0^n$, $Q=P_1^n$, $f=f_\lambda$, and the set $A$, we obtain
\begin{align}
H_\lambda(P_0^n,P_1^n)
&\geq
\sup_{a,b<0}
\Big\{
a\alpha_n
+b(1-\alpha_n)
-f_\lambda^*(b)
-\bigl(f_\lambda^*(a)-f_\lambda^*(b)\bigr)
(1-\beta_n(\varepsilon))
\Big\}\nonumber\\
&=\sup_{a,b<0}
\Big\{
a\alpha_n 
-f_\lambda^*(a)(1-\beta_n(\varepsilon))
+b(1-\alpha_n)
-f_\lambda^*(b)\beta_n(\varepsilon)
\Big\}\nonumber\\
&=\sup_{a<0} \Big\{a\alpha_n-f^\star_\lambda(a)(1-\beta_n(\varepsilon))\Big\} + \sup_{b<0} \Big\{b(1-\alpha_n)-f^\star_\lambda(b)\beta_n(\varepsilon)\Big\}.
\label{eq:app:thm1_hellinger}
\end{align}
From \eqref{eq:app:thm1_hellinger} one can verify that the value of $a$ and $b$ that maximises the two suprema are
\begin{equation*}
    a = -\lambda\left(\frac{\alpha_n}{1-\beta_n(\varepsilon)}\right)^{\lambda-1}
\end{equation*}
and
\begin{equation*}
    b=-\lambda\left(\frac{1-\alpha_n}{\beta_n(\varepsilon)}\right)^{\lambda-1}.
\end{equation*}
Substituting these values into \eqref{eq:app:thm1_hellinger} gives
\begin{align}
H_\lambda(P_0^n,P_1^n)
\geq
1-\alpha_n^\lambda\bigl(1-\beta_n(\varepsilon)\bigr)^{1-\lambda}-
(1-\alpha_n)^\lambda
\beta_n(\varepsilon)^{1-\lambda}.
\label{eq:app:thm_1_hellinger_step2}
\end{align}
From which we obtain
\begin{equation}\label{eq:app:hellinger_bound}
    h_\lambda(P_0^n,P_1^n)=1-H_\lambda(P_0^n,P_1^n)\leq \alpha_n^\lambda\bigl(1-\beta_n(\varepsilon)\bigr)^{1-\lambda}+
(1-\alpha_n)^\lambda
\beta_n(\varepsilon)^{1-\lambda}.
\end{equation}
Let us now show that the right hand side in \eqref{eq:app:hellinger_bound} is increasing in $\alpha_n$. 

Define
\begin{equation}
g(\alpha)
=
\alpha^\lambda
\bigl(1-\beta_n(\varepsilon)\bigr)^{1-\lambda}
+
(1-\alpha)^\lambda
\beta_n(\varepsilon)^{1-\lambda}.
\end{equation}
Its derivative is
\begin{equation}
g'(\alpha)
=
\lambda
\left[
\alpha^{\lambda-1}
\bigl(1-\beta_n(\varepsilon)\bigr)^{1-\lambda}
-
(1-\alpha)^{\lambda-1}
\beta_n(\varepsilon)^{1-\lambda}
\right].
\end{equation}
For $\alpha\leq1-\beta_n(\varepsilon)$, we have
$g'(\alpha)\geq0$. Thus, since 
$\varepsilon+\beta_n(\varepsilon)\leq1$, we have
$\alpha_n\leq\varepsilon\leq1-\beta_n(\varepsilon)$, and 
\begin{equation}
g(\alpha_n)\leq g(\varepsilon).
\end{equation}
Applying this inequality to \eqref{eq:app:hellinger_bound}
yields
\begin{equation*}
    h_\lambda(P_0^n,P_1^n)\leq \varepsilon^\lambda\bigl(1-\beta_n(\varepsilon)\bigr)^{1-\lambda}+
(1-\varepsilon)^\lambda
\beta_n(\varepsilon)^{1-\lambda},
\end{equation*}
which concludes the proof.
\end{proof}

\begin{corollary}\label{cor:KL_converse}
Let $P_0$ and $P_1$ be mutually absolutely continuous. Then for any $\varepsilon\in(0,1)$, \begin{align}
\beta_n(\varepsilon)
&\geq\max\Bigg\{
\exp\left\{
-\frac{nD(P_0\|P_1)+\log 2}{1-\varepsilon}
\right\}, 1-\frac{nD(P_1\|P_0)+\log2}{\log(1/\varepsilon)}\Bigg\}.
\label{eq:KL_bound}
\end{align}
\end{corollary}

It is worth noting that the bound derived in the above corollary coincides precisely with the classical Fano bound~\cite{fano1961transmission}.

\section{KL-based converse bounds with refinements}\label{app:KL_based_bound}
Let $P_1$ and $P_0$ be mutually absolutely continuous, and let $\varepsilon\in(0,1)$.

In finite-alphabet settings with common support, a more refined pair of bounds than the classical KL-based converse presented in Corollary \ref{cor:KL_converse} can be obtained using the smoothing-out method of~\cite{liu2020second}. Specifically, from~\cite[Theorem 2.4]{liu2020second}, for the forward KL $D(P_0\|P_1)$, one gets
\begin{equation*}
    \beta_n(\varepsilon)\geq \exp\Bigg\{-nD(P_0\|P_1)-2\sqrt{\log\frac{1}{1-\varepsilon}}\sqrt{n(a-1)}-\log\frac{1}{1-\varepsilon}\Bigg\},
\end{equation*}
where $a=\left\|\frac{dP_0}{dP_1}\right\|_{\infty}.$
For the reverse KL divergence $D(P_1\|P_0)$, the same approach gives
\begin{equation*}
    \beta_n(\varepsilon)\geq \sup_{t>0}\Bigg\{1-\Bigg(\varepsilon e^{nD(P_1\|P_0)+(b-1)nt}\Bigg)^{\frac{t}{t+1}}\Bigg\},
\end{equation*}
where $b=\left\|\frac{dP_1}{dP_0}\right\|_{\infty}.$
Moreover, in the regime $nD(P_1\|P_0)<\log(1/\varepsilon)$, the supremum over $t$ in the above bound admits the following closed-form solution:
\begin{equation*}
    \beta_n(\varepsilon)\geq 1- \exp\Bigg\{-(b-1)n\left(\sqrt{1-\frac{nD(P_1\|P_0)+\log \varepsilon}{(b-1)n}}-1\right)^2\Bigg\}.
\end{equation*}

When the alphabet is unbounded, the $\left\|\cdot\right\|_{\infty}$ quantities appearing above may be unbounded as well. However, the smoothing-out argument in \cite[Theorem 2.4]{liu2020second} can still be adapted to obtain finite-sample bounds. 

As an example,  consider the Gaussian setting,  $P_0=\mathcal{N}(\mu, 1)$ and $P_1=\mathcal{N}(\mu+\Delta, 1)$.  

For the forward KL, for any $t>0$, we obtain
\begin{equation*}
    \beta_n(\varepsilon)\geq \exp\Bigg\{-nD(P_0\|P_1)+\frac{1}{1-e^{-2t}}\log(1-\varepsilon)-nt-\frac{\Delta^2}{2}\left(e^t-1\right)^2 -n\big(\cosh(2t)-1\big)\Bigg\},
\end{equation*}
where $\cosh(t) = \frac{e^t + e^{-t}}{2}$.

Instead, for the reverse KL, for any $t>0$, we get
\begin{equation*}
    \beta_n(\varepsilon)\geq 1-\exp\Bigg\{(1-e^{-2t})\Bigg[nD(P_1\|P_0)+\log \varepsilon+nt+\frac{\Delta^2}{2}(e^t-1)^2+n\big(\cosh(2t)-1\big)\Bigg]\Bigg\},
\end{equation*}
where $\cosh(t) = \frac{e^t + e^{-t}}{2}$.

\section{Proof of Theorem \ref{th:phase_transition}}\label{app:th:phase_transition}
    Setting $\varepsilon_n=e^{-nr}$ in the data-processing converse of Remark~\ref{rem:DPI_recover} gives, for every $\lambda>1$,
    \begin{equation}\label{eq:phase_transition_proof_converse}
        \beta_n(e^{-nr})\geq1-\exp\left(-\frac{\lambda-1}{\lambda}n\bigl(r-D_\lambda(P_1\|P_0)\bigr)\right).
    \end{equation}
    Optimising over $\lambda>1$ gives~\eqref{eq:r>}.
    
    Similarly, applying Corollary~\ref{cor:beta_optimal_upper_bound_renyi} with $\varepsilon_n=e^{-nr}$ gives, for every $\lambda\in(0,1)$,
    \begin{align}
        \beta_n(e^{-nr})&\leq\lambda(1-\lambda)^{\frac{1-\lambda}{\lambda}}\exp\left(-\frac{1-\lambda}{\lambda}n\bigl(D_\lambda(P_1\|P_0)-r\bigr)\right)\nonumber\\
        &\leq\exp\left(-\frac{1-\lambda}{\lambda}n\bigl(D_\lambda(P_1\|P_0)-r\bigr)\right).\label{eq:phase_transition_proof_achievability}
    \end{align}
    Optimising over $\lambda\in(0,1)$ gives~\eqref{eq:r<}.
    
    \smallskip
    \noindent
    \textbf{Case $r>D(P_1\|P_0)$}: By~\eqref{eq:renyi_right_continuity}, there exists $\lambda_+>1$ sufficiently close to one such that
    \begin{equation}
        r>D_{\lambda_+}(P_1\|P_0).
    \end{equation}
    For this choice of $\lambda$, the exponential term in~\eqref{eq:phase_transition_proof_converse} vanishes as $n\to\infty$. Hence,
    \begin{equation}
        \lim_{n\to\infty}\beta_n(e^{-nr})=1.
    \end{equation}
    
    \smallskip
    \noindent
    \textbf{Case $r<D(P_1\|P_0)$}: Similarly, since $D_{\lambda}(P_1\| P_0)$ converges to  $D(P_1\| P_0)$ as $\lambda\to 1^-$, there exists $\lambda_-\in(0,1)$ sufficiently close to one such that
    \begin{equation}
        r<D_{\lambda_-}(P_1\|P_0).
    \end{equation}
    For this choice of $\lambda$, the exponential term in~\eqref{eq:phase_transition_proof_achievability} vanishes as $n\to\infty$. Hence,
    \begin{equation}
        \lim_{n\to\infty}\beta_n(e^{-nr})=0.
    \end{equation}
    Combining the two cases proves the result.

\section{Proofs of the Sample-Complexity Bounds}
\subsection{Proof of Corollary \ref{cor:lower_bound_sample_complexity}}\label{app:cor:lower_bound_sample_complexity}
Let $n=n(\varepsilon,\delta)$ and set $\beta_\star=\beta_n(\varepsilon)$. By the definition of $n(\varepsilon,\delta)$, we have $\beta_\star\leq\delta$. The first converse in Theorem~\ref{thm:renyi_converse_bound_via_vr} implies that, for any $\lambda>1$ and $c>0$,
\begin{equation}
    \beta_\star\geq \frac{e^{(\lambda-1)c}-e^{\frac{\lambda-1}{\lambda}nD_\lambda(P_1\|P_0)}(1-\varepsilon+\varepsilon e^{\lambda c})^{\frac{\lambda-1}{\lambda}}}{e^{(\lambda-1)c}-1}.
\end{equation}
Rearranging the above inequality to isolate $n$ gives
\begin{equation}\label{eq:first_step_add_sc}
    n\geq
\frac{1}{D_\lambda(P_1\|P_0)}
\left[
\frac{\lambda}{\lambda-1}
\log\left((1-\beta_\star)e^{(\lambda-1)c}+\beta_\star\right)
-\log\left(\varepsilon e^{\lambda c}+1-\varepsilon\right)
\right].
\end{equation}
Since $\beta_\star\leq\delta$ and the right-hand side of \eqref{eq:first_step_add_sc} is decreasing in $\beta_\star$, we obtain
\begin{equation*}
        n\geq \frac{1}{D_{\lambda}(P_1\|P_0)}\left[\frac{\lambda}{\lambda-1}\log((1-\delta) e^{(\lambda-1)c}+\delta)-\log(\varepsilon e^{\lambda c}+1-\varepsilon)\right].
\end{equation*}
Since $n=n(\varepsilon,\delta)$ and the above inequality holds for all $\lambda>1$ and $c>0$, taking the supremum over $c$ and $\lambda$ yields the first bound in \eqref{eq:lb_sample_complexity}. 

\noindent
The same argument applied to the second bound in Theorem~\ref{thm:renyi_converse_bound_via_vr} gives, for every $\lambda>1$ and $c>0$,
\begin{equation}\label{eq:step_sample_mul}
        n\geq \frac{1}{D_{\lambda}(P_0\|P_1)}\left[\frac{\lambda}{\lambda-1}\log((1-\varepsilon)e^{(\lambda-1)c}+\varepsilon)-\log(\beta_\star e^{\lambda c}+1-\beta_\star)\right]
\end{equation}
Since the right-hand side of \eqref{eq:step_sample_mul} is decreasing in $\beta_\star$ and $\beta_\star\leq\delta$, we get
\begin{equation}
        n\geq \frac{1}{D_{\lambda}(P_0\|P_1)}\left[\frac{\lambda}{\lambda-1}\log((1-\varepsilon)e^{(\lambda-1)c}+\varepsilon)-\log(\delta e^{\lambda c}+1-\delta)\right].
\end{equation}
Here, since $n=n(\varepsilon,\delta)$ and the above inequality holds for any $c>0$ and $\lambda>1$, taking the supremum over $\lambda$ and $c$ yields the second bound in \eqref{eq:lb_sample_complexity}.

Let us now optimise the bounds over the auxiliary parameter $c$. For the first bound, for any fixed $\lambda>1$, define
\begin{equation}
        f(c)=\frac{\lambda}{\lambda-1}\log\left((1-\delta) e^{(\lambda-1)c}+\delta\right)-\log(\varepsilon e^{\lambda c}+1-\varepsilon).
\end{equation}
Its derivative $f^\prime(c)$ is given by
\begin{equation}
    f^\prime(c)=\frac{\lambda(1-\delta)e^{(\lambda-1)c}}{(1-\delta)e^{(\lambda-1)c}+\delta}-\frac{\lambda\varepsilon e^{\lambda c}}{\varepsilon e^{\lambda c}+1-\varepsilon}.
\end{equation}
 By evaluating the sign of $f^\prime(c)$, one can verify that $f^\prime(c)$ is positive for 
    \[
        c<  \log\left(\frac{(1-\varepsilon)(1-\delta)}{\varepsilon\delta}\right)
    \]
    and negative for
    \[
        c>  \log\left(\frac{(1-\varepsilon)(1-\delta)}{\varepsilon\delta}\right).
    \]
    Consequently, $f(c)$ achieves its unique global maximum exactly at
    \[
        c^*=\log\left(\frac{(1-\varepsilon)(1-\delta)}{\varepsilon\delta}\right).
    \]
Thus, whenever $\varepsilon+\delta<1$ we have $c^*>0$. Substituting $c^*$ into the first bound in \eqref{eq:lb_sample_complexity} gives the first bound in \eqref{eq:lb_sample_complexity_opt}.

Similarly, for the second bound, for fixed $\lambda>1$, define
\begin{equation}
g(c)
=
\frac{\lambda}{\lambda-1}
\log\left((1-\varepsilon)e^{(\lambda-1)c}+\varepsilon\right)
-
\log\left(\delta e^{\lambda c}+1-\delta\right).
\end{equation}
Its derivative is
\begin{equation}
g^\prime(c)
=
\frac{\lambda(1-\varepsilon)e^{(\lambda-1)c}}
{(1-\varepsilon)e^{(\lambda-1)c}+\varepsilon}
-
\frac{\lambda\delta e^{\lambda c}}
{\delta e^{\lambda c}+1-\delta}.
\end{equation}
As above, one can verify that $g^\prime(c)$ is positive for $c<c^*$ and negative for $c>c^*$, with
\[
        c^*=\log\left(\frac{(1-\varepsilon)(1-\delta)}{\varepsilon\delta}\right).
\]
Consequently, $g(c)$ also attains its unique global maximum at $c^*$. 

Therefore, as for the first bound, whenever $\varepsilon+\delta<1$, we have $c^*>0$. Substituting it into the second bound in \eqref{eq:lb_sample_complexity} yields the second term in \eqref{eq:lb_sample_complexity_opt}, which concludes the proof.

    \subsection{Proof of Corollary~\ref{cor:upper_bound_sample_complexity}}\label{app:cor:upper_bound_sample_complexity}
    Fix $\lambda\in(0,1)$. By Corollary~\ref{cor:beta_optimal_upper_bound_renyi},
    \begin{equation}
        \beta_n(\varepsilon)\leq\lambda\left(\frac{1-\lambda}{\varepsilon}\right)^{\frac{1-\lambda}{\lambda}}\exp\left(-\frac{1-\lambda}{\lambda}nD_\lambda(P_1\|P_0)\right).
    \end{equation}
    The right-hand side is at most $\delta$ whenever
    \begin{equation}
        n\geq\frac{\log\left(\frac{1-\lambda}{\varepsilon}\right)+\frac{\lambda}{1-\lambda}\log\left(\frac{\lambda}{\delta}\right)}{D_\lambda(P_1\|P_0)}.
    \end{equation}
    Therefore, setting
    \begin{equation}
        n_\lambda=\left\lceil\frac{\log\left(\frac{1-\lambda}{\varepsilon}\right)+\frac{\lambda}{1-\lambda}\log\left(\frac{\lambda}{\delta}\right)}{D_\lambda(P_1\|P_0)}\right\rceil
    \end{equation}
    ensures that $\beta_{n_\lambda}(\varepsilon)\leq\delta$. By the definition of sample complexity, $n(\varepsilon,\delta)\leq n_\lambda$. Taking the infimum over $\lambda\in(0,1)$ proves the result.

\subsection{Proof of Remark \ref{remark:sc_bounds_comparison}}\label{app:remark:sc_bounds_comparison}

Consider the strongly asymmetric regime where the Type~II error bound $\delta \in (0, 1/32]$ is fixed and the Type~I error bound $\varepsilon \to 0$. 
In this regime, by expanding the term $\lambda_*/(1-\lambda_*)$, the bound in \eqref{eq:pensia_sc_eps<delta} simplifies to
\begin{equation}
    n(\varepsilon,\delta)\geq \frac{1}{2}\frac{\log(1/(2\varepsilon))}{D_{\lambda_*}(P_1\|P_0)}.
\end{equation}
As $\varepsilon \to 0$, it follows that $\lambda_* \to 1^-$, which implies $D_{\lambda_*}(P_1\|P_0) \to D(P_1\|P_0)$. Furthermore, since $\log(1/(2\varepsilon)) = \log(1/\varepsilon) - \log 2$, the asymptotic behaviour of \eqref{eq:pensia_sc_eps<delta} is governed by
\begin{equation}\label{app:eq:pensia_lb_sc}
n(\varepsilon,\delta) \gtrsim \frac{1}{2}\frac{\log(1/\varepsilon)}{D(P_1\|P_0)}.
\end{equation}
 We now evaluate the behaviour of our bound in \eqref{eq:lb_sample_complexity}. Specifically, since $\varepsilon+\delta<1$, we can consider the bound in \eqref{eq:lb_sample_complexity_opt} yielding
\begin{align}\label{app:eq:step_1_sample_cc_case_1}
n(\varepsilon,\delta) \geq \sup_{\lambda>1} \Bigg\{ \frac{1}{D_{\lambda}(P_1\|P_0)} \Bigg[ &\frac{\lambda}{\lambda-1}\log\left((1-\delta)\left(\frac{(1-\varepsilon)(1-\delta)}{\varepsilon\delta}\right)^{\lambda-1}+\delta\right) \nonumber\\
&- \log\left(\varepsilon \left(\frac{(1-\varepsilon)(1-\delta)}{\varepsilon\delta}\right)^\lambda+1-\varepsilon\right) \Bigg] \Bigg\}.\end{align}
To study the asymptotic behaviour of \eqref{app:eq:step_1_sample_cc_case_1} as $\varepsilon \to 0$ for a fixed $\lambda > 1$, we isolate the dominant terms in the arguments of the two logarithms. For the first logarithm, we have
\begin{equation}\label{eq:case_1_term1}
(1-\delta)\left(\frac{(1-\varepsilon)(1-\delta)}{\varepsilon\delta}\right)^{\lambda-1}+\delta = \frac{(1-\delta)^\lambda}{\delta^{\lambda-1}} \varepsilon^{1-\lambda}(1-\varepsilon)^{\lambda-1} +\delta.
\end{equation}
Since $\lambda > 1$, the term $\varepsilon^{1-\lambda} \to +\infty$ and the term $(1-\varepsilon)^{\lambda-1}\to 1$ as $\varepsilon \to 0$. Thus, \eqref{eq:case_1_term1} can be written as
\begin{align*}
    \frac{(1-\delta)^\lambda}{\delta^{\lambda-1}} \varepsilon^{1-\lambda}(1-\varepsilon)^{\lambda-1} +\delta&=\frac{(1-\delta)^\lambda}{\delta^{\lambda-1}} \varepsilon^{1-\lambda}\left((1-\varepsilon)^{\lambda-1}+\frac{\delta}{\frac{(1-\delta)^\lambda}{\delta^{\lambda-1}} \varepsilon^{1-\lambda}}\right)\\
    &=\frac{(1-\delta)^\lambda}{\delta^{\lambda-1}} \varepsilon^{1-\lambda}\big(1+o(1)\big).
\end{align*}

Following the same reasoning, the argument of the second logarithm can be rearranged as \begin{equation}\label{eq:case_1_term2}
\varepsilon \left(\frac{(1-\varepsilon)(1-\delta)}{\varepsilon\delta}\right)^\lambda+1-\varepsilon = \left(\frac{1-\delta}{\delta}\right)^\lambda \varepsilon^{1-\lambda} (1-\varepsilon)^{\lambda} + 1-\varepsilon,
\end{equation}
which is asymptotically equivalent to
\[
\left(\frac{1-\delta}{\delta}\right)^\lambda \varepsilon^{1-\lambda}(1 + o(1)).
\]
Substituting these back into \eqref{app:eq:step_1_sample_cc_case_1}, and using the property $\log\bigg(x\big(1+o(1)\big)\bigg) = \log x + o(1)$, we get that the term inside the square brackets in \eqref{app:eq:step_1_sample_cc_case_1} becomes
\begin{align}
&\frac{\lambda}{\lambda-1}\log\left( \frac{(1-\delta)^\lambda}{\delta^{\lambda-1}} \varepsilon^{1-\lambda} \right) - \log\left( \left(\frac{1-\delta}{\delta}\right)^\lambda \varepsilon^{1-\lambda} \right) + o(1) \nonumber \\
&= \frac{\lambda}{\lambda-1} \log\left( \frac{(1-\delta)^\lambda}{\delta^{\lambda-1}} \right) + \frac{\lambda}{\lambda-1}(1-\lambda)\log \varepsilon - \log\left( \frac{1-\delta}{\delta} \right)^\lambda - (1-\lambda)\log \varepsilon + o(1) \nonumber \\
&= -\lambda \log \varepsilon + (\lambda-1)\log \varepsilon + \frac{\lambda}{\lambda-1} \log(1-\delta)+ o(1) \nonumber \\
&= \log(1/\varepsilon) + C(\lambda, \delta) + o(1), \label{eq:step_2_sample_cc_case_1}
\end{align}
where 
\begin{equation*}
    C(\lambda, \delta)=\frac{\lambda}{\lambda-1}\log(1-\delta),
\end{equation*}
which is independent of $\varepsilon$.

For every fixed $\lambda>1$, \eqref{eq:step_2_sample_cc_case_1} implies
\begin{equation*}
    \liminf_{\varepsilon\downarrow0}
    \frac{n(\varepsilon,\delta)}{\log(1/\varepsilon)}
    \geq\frac{1}{D_\lambda(P_1\|P_0)}.
\end{equation*}
Letting $\lambda\downarrow1$ and using \eqref{eq:renyi_right_continuity} gives
\begin{equation}\label{eq:final_step_sample_cc_case_1}
    \liminf_{\varepsilon\downarrow0}
    \frac{n(\varepsilon,\delta)}{\log(1/\varepsilon)}
    \geq\frac{1}{D(P_1\|P_0)}.
\end{equation}
Comparing \eqref{eq:final_step_sample_cc_case_1} with \eqref{app:eq:pensia_lb_sc} proves the factor-two improvement in the leading constant.

\subsection{Proof of Remark \ref{remark:sc_upper_comparison}}\label{app:remark:sc_upper_comparison}
Consider the symmetric regime $\varepsilon=\delta\leq 1/32$. In this regime, $\lambda_*=1/2$ and the bound in \eqref{eq:pensia_sc_eps<delta} simplifies to
\begin{equation}\label{eq:pensia_ub_simpl}
    n(\varepsilon,\delta)\leq \Bigg\lceil\frac{2}{D_{1/2}(P_1\|P_0)}\log\left(\frac{1}{2\delta}\right)\Bigg\rceil.
\end{equation}
On the other hand, evaluating \eqref{eq:ub_sample_complexity} at $\lambda=1/2$ gives
\begin{align}
    n(\varepsilon,\delta)&\leq\Bigg\lceil\frac{2}{D_{1/2}(P_1\|P_0)}\log\left(\frac{1}{2\delta}\right)\Bigg\rceil.
\end{align}
Thus, the two bounds exhibits the same behaviour.

Consider the asymmetric regime $\varepsilon<\delta\leq 1/32$ and let $\lambda_*$ be defined as
\begin{equation*}
\lambda_*= \frac{\log\big(1/(2\varepsilon)\big)}{\log\big(1/(2\delta)\big)+\log\big(1/(2\varepsilon)\big)}\in[0.5,1).    
\end{equation*}
 For
\begin{equation*}
    \varepsilon\leq \frac{1}{2}\left(\frac{\delta}{\lambda_*}\right)^{\frac{\lambda_*}{1-\lambda_*}},
\end{equation*}
it follows that 
\begin{equation}\label{eq:upp_eps}
    \frac{\lambda_*}{1-\lambda_*}\log \left(\frac{\lambda_*}{\delta}\right)\leq \log\left(\frac{1}{2\varepsilon}\right).
\end{equation}
Moreover, since $\lambda_*\in[1/2,1)$, we also have that
\begin{equation}
    \log\left(\frac{1-\lambda_*}{\varepsilon}\right)\leq \log\left(\frac{1/2}{\varepsilon}\right)=\log\left(\frac{1}{2\varepsilon}\right).
\end{equation}
Applying the above inequality along with \eqref{eq:upp_eps} to \eqref{eq:ub_sample_complexity} evaluated at $\lambda=\lambda_*$, we get
\begin{align}
    n(\varepsilon,\delta)&\leq\Bigg\lceil\frac{2}{D_{\lambda_*}(P_1\|P_0)}\log\left(\frac{1}{2\varepsilon}\right)\Bigg\rceil.
\end{align}
The last expression corresponds exactly to \eqref{eq:pensia_sc_eps<delta}.

\subsection{Proof of Corollary \ref{cor:sc_gap}}\label{app:cor:sc_gap}

Let $\varepsilon_m=e^{-mR}$ for $R>0$ and fix $\delta\in(0,1)$. We analyze the lower and upper bounds separately.

Fix any $\lambda>1$ and let $c=mR$ in
the first term of \eqref{eq:lb_sample_complexity}. We obtain the following bound
\begin{align}
L_m
&\geq
\frac{1}{D_\lambda(P_1\|P_0)}
\Bigg[
\frac{\lambda}{\lambda-1}
\log\left((1-\delta)e^{(\lambda-1)mR}+\delta\right)
-\log\left(e^{(\lambda-1)mR}+1-e^{-mR}\right)
\Bigg].
\end{align}
Factoring out $e^{(\lambda-1)mR}$ from both logarithms yields
\begin{align}\label{app:eq:L_m_1}
L_m
&\geq
\frac{1}{D_\lambda(P_1\|P_0)}
\Bigg[
mR
+\frac{\lambda}{\lambda-1}
\log\left(1-\delta+\delta e^{-(\lambda-1)mR}\right)-\log\left(1+e^{-(\lambda-1)mR}-e^{-\lambda mR}\right)
\Bigg]\nonumber\\
&=\frac{mR}{D_\lambda(P_1\|P_0)}
\Bigg[
1
+\frac{1}{mR}\frac{\lambda}{\lambda-1}
\log\left(1-\delta+\delta e^{-(\lambda-1)mR}\right)-\frac{1}{mR}\log\left(1+e^{-(\lambda-1)mR}-e^{-\lambda mR}\right)
\Bigg].
\end{align}
For every fixed $\lambda>1$ and $\delta\in(0,1)$,  the two logarithmic
terms in \eqref{app:eq:L_m_1} remain bounded as $m\to\infty$. Consequently,
\[
L_m
\geq
\frac{mR}{D_\lambda(P_1\|P_0)}\big(1+o(1)\big).
\]
Thus,
\begin{equation}\label{eq:lower_asymptotic_m}
\liminf_{m\to\infty}\frac{L_m}{m}
\geq
\frac{R}{D_\lambda(P_1\|P_0)}.
\end{equation}
Since \eqref{eq:lower_asymptotic_m} holds for every $\lambda\in(1,1+\eta]$, letting
$\lambda\to1^+$ and using $
D_\lambda(P_1\|P_0)\to D(P_1\|P_0)$
gives
\begin{equation}\label{eq:lower_asymptotic_m_final}
\liminf_{m\to\infty}\frac{L_m}{m}
\geq
\frac{R}{D(P_1\|P_0)}.
\end{equation}

We next analyse the upper bound. Fix any $\lambda\in(0,1)$. Corollary~\ref{cor:upper_bound_sample_complexity} gives
\begin{align}
    U_m &\leq \left\lceil
    \frac{\log\left(\frac{1-\lambda}{\varepsilon}\right)+\frac{\lambda}{1-\lambda}\log\left(\frac\lambda\delta\right)}{D_\lambda(P_1\|P_0)}
    \right\rceil 
    \\&\leq \left\lceil
    \frac{\log\left(\frac1\varepsilon\right)+\frac{\lambda}{1-\lambda}\log\left(\frac1\delta\right)}{D_\lambda(P_1\|P_0)}
    \right\rceil 
    \\&= 
    \left\lceil
    \frac{mR+\frac{\lambda}{1-\lambda}\log\left(\frac1\delta\right)}{D_\lambda(P_1\|P_0)}
    \right\rceil.
\end{align}
Because $\lambda$ and $\delta$ are fixed,
\begin{equation}
    \limsup_{m\to\infty}\frac{U_m}{m}
    \leq\frac{R}{D_\lambda(P_1\|P_0)}.
\end{equation}
Letting $\lambda\uparrow1$ yields
\begin{equation}\label{eq:app:upper_asymptotic_sharp}
    \limsup_{m\to\infty}\frac{U_m}{m}
    \leq\frac{R}{D(P_1\|P_0)}.
\end{equation}
Since $L_m\leq n(\varepsilon_m,\delta)\leq U_m$, \eqref{eq:lower_asymptotic_m_final} and \eqref{eq:app:upper_asymptotic_sharp} imply
\begin{equation*}
    \frac{R}{D(P_1\|P_0)}
    \leq\liminf_{m\to\infty}\frac{L_m}{m}
    \leq\limsup_{m\to\infty}\frac{U_m}{m}
    \leq\frac{R}{D(P_1\|P_0)}.
\end{equation*}
Thus both bounds converge to the same positive limit, and $U_m/L_m\to1$.

\subsection{Proof of Corollary~\ref{cor:ldp_upper_bound}}\label{app:cor:ldp_upper_bound}
    Fix $\lambda\in(0,1)$ and choose
    \begin{equation}
    \widehat T_\lambda\in\operatorname*{arg\,max}_{T\in\{T_{\rm B},T_{\rm RR}\}}D_\lambda(TP_1\|TP_0).
    \end{equation}
    By the definition of $D_\lambda^{\rm ach}(\epsilon_{\rm dp})$,
    \begin{equation}
    D_\lambda(\widehat T_\lambda P_1\|\widehat T_\lambda P_0)=D_\lambda^{\rm ach}(\epsilon_{\rm dp}).
    \end{equation}
    Both candidate mechanisms have strictly positive transition probabilities, so their induced distributions are mutually absolutely continuous. Applying Corollary~\ref{cor:beta_optimal_upper_bound_renyi} to $\widehat T_\lambda P_0$ and $\widehat T_\lambda P_1$ gives
    \begin{equation}
    \beta_\varepsilon\left((\widehat T_\lambda P_0)^n,(\widehat T_\lambda P_1)^n\right)\leq\lambda\left(\frac{1-\lambda}{\varepsilon}\right)^{\frac{1-\lambda}{\lambda}}\exp\left(-\frac{1-\lambda}{\lambda}nD_\lambda^{\rm ach}(\epsilon_{\rm dp})\right).
    \end{equation}
    Since $\widehat T_\lambda\in\mathcal C_{\epsilon_{\rm dp}}$, choosing $T_i=\widehat T_\lambda$ for every $i$ in~\eqref{eq:private_beta_definition} yields
    \begin{equation}
    \beta_n^{\rm LDP}(\varepsilon, \epsilon_{\rm dp})\leq\beta_\varepsilon\left((\widehat T_\lambda P_0)^n,(\widehat T_\lambda P_1)^n\right).
    \end{equation}
    Combining the preceding inequalities and taking the infimum over $\lambda\in(0,1)$ proves the result.

\subsection{Proof of Theorem~\ref{th:ldp_phase_transition}}\label{app:th:ldp_phase_transition}
    Set $\varepsilon_n=e^{-nr}$. Taking $c\to\infty$ in the first term of Corollary~\ref{cor:LDP_lower_bound} gives
    \begin{align}
        \beta_n^{\rm LDP}(\varepsilon_n,\epsilon_{\rm dp})
        &\geq
        1-\inf_{\lambda>1}
        \left(
        \varepsilon_n
        e^{n(1-e^{-\epsilon_{\rm dp}})D_\lambda(P_1\|P_0)}
        \right)^{\frac{\lambda-1}{\lambda}}\nonumber\\
        &=
        1-\inf_{\lambda>1}
        \exp\left(
        -\frac{\lambda-1}{\lambda}n
        \left[
        r-(1-e^{-\epsilon_{\rm dp}})D_\lambda(P_1\|P_0)
        \right]
        \right),
    \end{align}
    which proves~\eqref{eq:ldp_phase_converse}. Similarly, Corollary~\ref{cor:ldp_upper_bound} yields
    \begin{align}
        \beta_n^{\rm LDP}(\varepsilon_n,\epsilon_{\rm dp})
        &\leq
        \inf_{\lambda\in(0,1)}
        \lambda(1-\lambda)^{\frac{1-\lambda}{\lambda}}
        \exp\left(
        -\frac{1-\lambda}{\lambda}n
        \left[
        D_\lambda^{\rm ach}(\epsilon_{\rm dp})-r
        \right]
        \right)\nonumber\\
        &\leq
        \inf_{\lambda\in(0,1)}
        \exp\left(
        -\frac{1-\lambda}{\lambda}n
        \left[
        D_\lambda^{\rm ach}(\epsilon_{\rm dp})-r
        \right]
        \right),
    \end{align}
    where the second inequality follows from
    $\lambda(1-\lambda)^{(1-\lambda)/\lambda}\leq1$. This proves~\eqref{eq:ldp_phase_achievability}.
    
    \textbf{Case $r>r_{\rm con}$.}
    By~\eqref{eq:renyi_right_continuity}, $\lim_{\lambda\downarrow1}D_\lambda(P_1\|P_0)=D(P_1\|P_0)$. Hence, there exists $\lambda_+>1$ such that
    \begin{equation}
        r>(1-e^{-\epsilon_{\rm dp}})D_{\lambda_+}(P_1\|P_0).
    \end{equation}
    Evaluating~\eqref{eq:ldp_phase_converse} at $\lambda_+$ gives
    \begin{equation}
        \beta_n^{\rm LDP}(\varepsilon_n,\epsilon_{\rm dp})
        \geq
        1-\exp\left(
        -\frac{\lambda_+-1}{\lambda_+}n
        \left[
        r-(1-e^{-\epsilon_{\rm dp}})D_{\lambda_+}(P_1\|P_0)
        \right]
        \right).
    \end{equation}
    The exponent is strictly positive, and therefore
    $\beta_n^{\rm LDP}(\varepsilon_n,\epsilon_{\rm dp})\to1$
    exponentially fast.
    
    \textbf{Case $0<r<r_{\rm ach}$.}
    By the definition of $D^{\rm ach}(\epsilon_{\rm dp})$ and the left-continuity of R\'enyi divergence at order one, there exists $\lambda_-\in(0,1)$ such that $D_{\lambda_-}^{\rm ach}(\epsilon_{\rm dp})>r$. Evaluating the infimum in~\eqref{eq:ldp_phase_achievability} at $\lambda_-$ gives
    \begin{equation}
        \beta_n^{\rm LDP}(\varepsilon_n,\epsilon_{\rm dp})
        \leq
        \exp\left(
        -\frac{1-\lambda_-}{\lambda_-}n
        \left[
        D_{\lambda_-}^{\rm ach}(\epsilon_{\rm dp})-r
        \right]
        \right).
    \end{equation}
    The exponent is strictly positive, and hence
    $\beta_n^{\rm LDP}(\varepsilon_n,\epsilon_{\rm dp})\to0$
    exponentially fast.
    
    Finally, since $T_{\rm B},T_{\rm RR}\in\mathcal C_{\epsilon_{\rm dp}}$,
    \begin{equation}
        D^{\rm ach}(\epsilon_{\rm dp})\leq
        \sup_{T\in\mathcal C_{\epsilon_{\rm dp}}}
        D(TP_1\|TP_0)\leq
        \sup_{T\in\mathcal C_{\epsilon_{\rm dp}}}
        \eta_1(T)D(P_1\|P_0)\leq
        (1-e^{-\epsilon_{\rm dp}})D(P_1\|P_0).
    \end{equation}
    Thus, $r_{\rm ach}\leq r_{\rm con}$, completing the proof.

\IfFileExists{Bibliography.bib}{%
    \bibliographystyle{IEEEtran}
    \bibliography{Bibliography}%
}{%
    \PackageWarningNoLine{renyi-bht}{Bibliography.bib not found; citations remain unresolved}%
}

\end{document}

%% file: defs.tex
\def\l2lim{\,{\buildrel L_2 \over \rightarrow}\,}

\usepackage{tikz}
\usetikzlibrary{arrows,shapes,automata,petri,positioning,calc}

\tikzset{
    place/.style={
        circle,
        thick,
        draw=black,
        fill=gray!50,
        minimum size=6mm,
    },
        state/.style={
        circle,
        thick,
        draw=blue!75,
        fill=blue!20,
        minimum size=6mm,
    },
}

%% file: Bibliography.bib
@article{polyanskiy2010channel,
  title={Channel coding rate in the finite blocklength regime},
  author={Polyanskiy, Yury and Poor, H. Vincent and Verd{\'u}, Sergio},
  journal={IEEE Transactions on Information Theory},
  volume={56},
  number={5},
  pages={2307--2359},
  year={2010},
  publisher={IEEE}
}

@article{neyman1933ix,
  title={On the problem of the most efficient tests of statistical hypotheses},
  author={Neyman, Jerzy and Pearson, Egon Sharpe},
  journal={Philosophical Transactions of the Royal Society of London. Series A, Containing Papers of a Mathematical or Physical Character},
  volume={231},
  number={694-706},
  pages={289--337},
  year={1933},
  publisher={The Royal Society London}
}

@article{bhattacharyya2022approximating,
  title={On approximating total variation distance},
  author={Bhattacharyya, Arnab and Gayen, Sutanu and Meel, Kuldeep S. and Myrisiotis, Dimitrios and Vinodchandran, N.V.},
  journal={arXiv preprint arXiv:2206.07209},
  year={2022}
}

@article{vandenbroucque2026journal,
  title={Contraction of R\'enyi Divergences for Discrete Channels},
  author={Vandenbroucque, Adrien and Esposito, Amedeo Roberto and Gastpar, Michael},
  journal={arXiv preprint},
  year={2026}
}

@article{ahlswede1976bounds,
  title={Bounds on conditional probabilities with applications in multi-user communication},
  author={Ahlswede, Rudolf and G{\'a}cs, Peter and K{\"o}rner, J{\'a}nos},
  journal={Zeitschrift f{\"u}r Wahrscheinlichkeitstheorie und verwandte Gebiete},
  volume={34},
  number={2},
  pages={157--177},
  year={1976}
}

@article{liu2020second,
  title={Second-order converses via reverse hypercontractivity},
  author={Liu, Jingbo and Van Handel, Ramon and Verd{\'u}, Sergio},
  journal={Mathematical Statistics and Learning},
  volume={2},
  number={2},
  pages={103--163},
  year={2020}
}

@book{bar2002complexity,
  title={The complexity of massive data set computations},
  author={Bar-Yossef, Ziv},
  year={2002},
  publisher={PhD Thesis. University of California, Berkeley}
}

@article{mullhaupt2025bounding,
  title={Bounding Neyman-Pearson Region with $f$-Divergences},
  author={Mullhaupt, Andrew and Peng, Cheng},
  journal={arXiv preprint arXiv:2505.08899},
  year={2025}
}

@article{sason2016f,
  title={$f$-divergence Inequalities},
  author={Sason, Igal and Verd{\'u}, Sergio},
  journal={IEEE Transactions on Information Theory},
  volume={62},
  number={11},
  pages={5973--6006},
  year={2016},
  publisher={IEEE}
}

@article{asoodeh2020contraction,
  title={Contraction of {$E_\gamma$}-Divergence and Its Applications to Privacy},
  author={Asoodeh, Shahab and Diaz, Mario and Calmon, Flavio P.},
  journal={arXiv preprint arXiv:2012.11035},
  year={2020}
}

@ARTICLE{vanerven_renyi,
  author={van Erven, Tim and Harremos, Peter},
  journal={IEEE Transactions on Information Theory}, 
  title={Rényi Divergence and Kullback-Leibler Divergence}, 
  year={2014},
  volume={60},
  number={7},
  pages={3797-3820},
  doi={10.1109/TIT.2014.2320500}
}

@article{bhattacharyya2025total,
  title={Total variation distance for product distributions is \#P-complete},
  author={Bhattacharyya, Arnab and Gayen, Sutanu and Meel, Kuldeep S. and Myrisiotis, Dimitrios and Pavan, A. and Vinodchandran, N.V.},
  journal={Information Processing Letters},
  volume={189},
  pages={106560},
  year={2025},
  publisher={Elsevier}
}

@inproceedings{pensia2024sample,
  title={The sample complexity of simple binary hypothesis testing},
  author={Pensia, Ankit and Jog, Varun and Loh, Po-Ling},
  booktitle={The Thirty Seventh Annual Conference on Learning Theory},
  pages={4205--4206},
  year={2024},
  organization={PMLR}
}

@InProceedings{kazemi2025sample,
  title = 	 {The Sample Complexity of Distributed Simple Binary Hypothesis Testing under Information Constraints},
  author =       {Kazemi, Hadi and Pensia, Ankit and Varun, Jog},
  booktitle = 	 {Proceedings of Thirty Eighth Conference on Learning Theory},
  pages = 	 {3213--3214},
  year = 	 {2025},
  editor = 	 {Haghtalab, Nika and Moitra, Ankur},
  volume = 	 {291},
  series = 	 {Proceedings of Machine Learning Research},
  month = 	 {30 Jun--04 Jul},
  publisher =    {PMLR},
}

@article{pensia2024simple,
  title={Simple binary hypothesis testing under local differential privacy and communication constraints},
  author={Pensia, Ankit and Asadi, Amir R and Jog, Varun and Loh, Po-Ling},
  journal={IEEE Transactions on Information Theory},
  volume={71},
  number={1},
  pages={592--617},
  year={2024},
  publisher={IEEE}
}

@article{kairouz2016extremal,
  title={Extremal mechanisms for local differential privacy},
  author={Kairouz, Peter and Oh, Sewoong and Viswanath, Pramod},
  journal={Journal of Machine Learning Research},
  volume={17},
  number={17},
  pages={1--51},
  year={2016}
}

@article{chernoff1956large,
  title={Large-sample theory: Parametric case},
  author={Chernoff, Herman},
  journal={The Annals of Mathematical Statistics},
  volume={27},
  number={1},
  pages={1--22},
  year={1956},
  publisher={JSTOR}
}

@article{chernoff1952measure,
  title={A measure of asymptotic efficiency for tests of a hypothesis based on the sum of observations},
  author={Chernoff, Herman},
  journal={The Annals of Mathematical Statistics},
  pages={493--507},
  year={1952},
  publisher={JSTOR}
}

@book{cover1999elements,
  title={Elements of information theory},
  author={Cover, Thomas M.},
  year={1999},
  publisher={John Wiley \& Sons}
}

@ARTICLE{Espinosa,
  author={Espinosa, Sebastian and Silva, Jorge F. and Piantanida, Pablo},
  journal={IEEE Signal Processing Letters}, 
  title={Finite-Length Bounds on Hypothesis Testing Subject to Vanishing Type I Error Restrictions}, 
  year={2021},
  volume={28},
  number={},
  pages={229-233},
  doi={10.1109/LSP.2021.3050381}
}

@ARTICLE{Blahut,
  author={Blahut, R.},
  journal={IEEE Transactions on Information Theory}, 
  title={Hypothesis testing and information theory}, 
  year={1974},
  volume={20},
  number={4},
  pages={405-417},
  doi={10.1109/TIT.1974.1055254}
}

@ARTICLE{Han_strongconverse,
  author={Han, T.S. and Kobayashi, K.},
  journal={IEEE Transactions on Information Theory}, 
  title={The strong converse theorem for hypothesis testing}, 
  year={1989},
  volume={35},
  number={1},
  pages={178-180},
  doi={10.1109/18.42188}
}

@inproceedings{bruno2026finite,
      title={A Finite-Sample Strong Converse for Binary Hypothesis Testing via (Reverse) {R\'enyi} Divergence}, 
      author={Roberto Bruno and Adrien Vandenbroucque and Amedeo Roberto Esposito},
      booktitle={IEEE International Symposium on Information Theory (ISIT)},
      year={2026},
      organization={IEEE},
}

@inproceedings{lungu2024optimal,
  title={The optimal finite-sample error probability in asymmetric binary hypothesis testing},
  author={Lungu, V. and Kontoyiannis, I.},
  booktitle={2024 IEEE International Symposium on Information Theory (ISIT)},
  pages={831--836},
  year={2024},
  organization={IEEE}
}

@article{mosonyi2015quantum,
  title={Quantum hypothesis testing and the operational interpretation of the quantum R{\'e}nyi relative entropies},
  author={Mosonyi, Mil{\'a}n and Ogawa, Tomohiro},
  journal={Communications in Mathematical Physics},
  volume={334},
  number={3},
  pages={1617--1648},
  year={2015},
  publisher={Springer}
}

@inproceedings{polyanskiy2010arimoto,
  title={Arimoto channel coding converse and R{\'e}nyi divergence},
  author={Polyanskiy, Yury and Verd{\'u}, Sergio},
  booktitle={2010 48th Annual Allerton Conference on Communication, Control, and Computing (Allerton)},
  pages={1327--1333},
  year={2010},
  organization={IEEE}
}

@article{nguyen2010estimating,
  title={Estimating divergence functionals and the likelihood ratio by convex risk minimization},
  author={Nguyen, XuanLong and Wainwright, Martin J and Jordan, Michael I},
  journal={IEEE Transactions on Information Theory},
  volume={56},
  number={11},
  pages={5847--5861},
  year={2010},
  publisher={IEEE}
}

@article{csiszar1967information,
  title={On information-type measure of difference of probability distributions and indirect observations},
  author={Csisz{\'a}r, Imre},
  journal={Studia Sci. Math. Hungar.},
  volume={2},
  pages={299--318},
  year={1967}
}

@article{anantharam2018variational,
  title={A variational characterization of {R{\'e}}nyi divergences},
  author={Anantharam, Venkat},
  journal={IEEE Transactions on Information Theory},
  volume={64},
  number={11},
  pages={6979--6989},
  year={2018},
  publisher={IEEE}
}

@article{fano1961transmission,
  title={Transmission of information: A statistical theory of communications},
  author={Fano, Robert M. and Hawkins, David},
  journal={American Journal of Physics},
  volume={29},
  number={11},
  pages={793--794},
  year={1961},
  publisher={AIP Publishing}
}

@article{dwork2014algorithmic,
  title={The algorithmic foundations of differential privacy},
  author={Dwork, Cynthia and Roth, Aaron},
  journal={Foundations and trends{\textregistered} in theoretical computer science},
  volume={9},
  number={3-4},
  pages={211--487},
  year={2014},
  publisher={Emerald Publishing Limited}
}

@article{duchi2018minimax,
  title={Minimax optimal procedures for locally private estimation},
  author={Duchi, John C and Jordan, Michael I and Wainwright, Martin J},
  journal={Journal of the American Statistical Association},
  volume={113},
  number={521},
  pages={182--201},
  year={2018},
  publisher={Taylor \& Francis}
}

@article{asoodeh2024contraction,
  title={Contraction of locally differentially private mechanisms},
  author={Asoodeh, Shahab and Zhang, Huanyu},
  journal={IEEE Journal on Selected Areas in Information Theory},
  volume={5},
  pages={385--395},
  year={2024},
  publisher={IEEE}
}
